\documentclass[a4paper,UKenglish,cleveref, 
numberwithinsect, pdfa]{lipics-v2021}
\pdfoutput=1
\hideLIPIcs%
\title{Towards Finite Satisfiability of PDL with Loops}
\titlerunning{Towards Finite Satisfiability of PDL with Loops}

\author{Bartosz Jan Bednarczyk}
{Institute of Logic and Computation, TU Wien, Austria 
\and Computer Science Department, University of Wrocław, Poland 
\and \url{https://bartoszjanbednarczyk.github.io}}{bartek@cs.uni.wroc.pl}
{https://orcid.org/0000-0002-8267-7554}
{}

\author{Mikołaj Swoboda}
{Institute of Logic and Computation, TU Wien, Austria 
\and Computer Science Department, University of Wrocław, Poland}{mikolaj.swoboda@cs.uni.wroc.pl}
{https://orcid.org/0009-0007-8955-3179}
{}

\authorrunning{B. J. Bednarczyk and M. Swoboda} 
\Copyright{Bartosz Jan Bednarczyk and Mikołaj Swoboda} 



\funding{This work was supported by Polish NCN grants 2024/53/N/ST6/01390 (B.~Bednarczyk)
  and 2021/41/B/ST6/00996 (M.~Swoboda).}

\supplement{The authors developed almost all of the mathematical content of this paper the old-school way, with whiteboards and markers.
The only exception is the formal definitions of the tree automata in Appendices~\ref{app:toolbox}, \ref{app:encready-automaton}, and~\ref{app:foldable-cactus-automaton}, which were generated, and then simplified, with the assistance of Claude (Opus 5), starting from the authors' low-level descriptions of their states and transition functions.
Such definitions are usually omitted in the literature, but we chose to be pedantic and spell them out, driven by the ambition to make the paper as self-contained as possible.
Claude also assisted extensively with text editing and reviewing, as well as with vibe-coding our figures.
The authors carefully reviewed all of its output and take full responsibility for~the~content.}

\usepackage{lineno} 

\usepackage{ellipsis} 
\usepackage[abbreviations,british]{foreign} 
\usepackage{hyperref} 
\usepackage[usenames,dvipsnames]{xcolor}
\usepackage{float}

\definecolor{darkmidnightblue}{rgb}{0.0, 0.2, 0.4}
\definecolor{persianplum}{rgb}{0.44, 0.11, 0.11}

\hypersetup{
  colorlinks  = true, 
  citecolor   = persianplum, 
  urlcolor    = persianplum, 
  linkcolor   = persianplum
}

\usepackage{tikz}
\usepackage{todonotes}
\usepackage{xspace}

\usepackage{mathtools}
\usepackage{amsmath}
\usepackage{amsthm}
\usepackage{stmaryrd} 
\theoremstyle{plain}
\newtheorem*{theorem*}{Theorem} 
\newtheorem*{lemma*}{Lemma} 
\newtheorem*{observation*}{Observation} 
\newtheorem*{proposition*}{Proposition} 

\makeatletter
\def\desclabel#1#2{\begingroup
\def\@currentlabel{#1}%
\protected@edef\cref@currentlabel{[desclabel][][]#1}
#1\label{#2}\endgroup
}
\makeatother
\crefname{desclabel}{condition}{conditions}
\Crefname{desclabel}{Condition}{Conditions}

\newenvironment{namedlist}[1][(Cycle)]
  {\begin{list}{}{%
     \setlength{\labelsep}{0.5em}%
     \settowidth{\labelwidth}{#1}%
     \setlength{\leftmargin}{\labelwidth}%
     \addtolength{\leftmargin}{\labelsep}%
     \setlength{\itemsep}{0pt}%
     \setlength{\parsep}{0pt}%
     \setlength{\topsep}{4pt plus 1pt minus 1pt}%
     }}
  {\end{list}}

\DeclareMathAlphabet{\mathdutchcal}{U}{dutchcal}{m}{n}
\SetMathAlphabet{\mathdutchcal}{bold}{U}{dutchcal}{b}{n}
\DeclareMathAlphabet{\mathdutchbcal}{U}{dutchcal}{b}{n}

\definecolor{ao(english)}{rgb}{0.0, 0.5, 0.0}

\usepackage{fontawesome5} 
\usepackage{xurl} 

\usetikzlibrary{arrows, arrows.meta, decorations.markings, shapes, calc}
\colorlet{jaune}{yellow!80!green}

\colorlet{vert}{green!45!black}

\colorlet{bleu}{blue!70!black}

\colorlet{rouge}{red!80!black}

\tikzstyle{minirond}=[draw, circle, minimum height=2mm, inner sep=0pt]
\tikzstyle{ptrond}=[draw, circle, minimum height=2.5mm]
\tikzstyle{medrond}=[draw, circle, minimum height=5mm]
\tikzstyle{rond}=[draw, circle, minimum height=7mm]

\tikzstyle{carre}=[draw,minimum width=6mm,minimum height=6mm]
\tikzstyle{medcarre}=[draw,minimum width=4mm,minimum height=4mm]
\tikzstyle{ptcarre}=[draw,minimum width=2.8mm,minimum height=2.8mm]
\tikzstyle{minicarre}=[draw,minimum width=1.5mm,minimum height=1.5mm,inner sep=0pt]
\tikzstyle{rouge}=[draw=red,fill=red!20!white]
\tikzstyle{vert}=[draw=green!80!black,fill=green!80!black!20!white]
\tikzstyle{jaune}=[draw=yellow!60!red,fill=yellow!60!red!30!white]
\tikzstyle{bleu}=[draw=blue,fill=blue!40!white]
\tikzstyle{gris}=[draw=black!80!white,fill=black!40!white]
\tikzstyle{rougef}=[draw=red,fill=red!60!white]
\tikzstyle{vertf}=[draw=green!80!black,fill=green!80!black!60!white]
\tikzstyle{jaunef}=[draw=yellow!80!black,fill=yellow!80!black!60!white]
\tikzstyle{bleuf}=[draw=blue,fill=blue!70!white]
\tikzstyle{rjaune}=[style=rond,style=jaune]
\tikzstyle{rbleu}=[style=rond,style=bleu]
\tikzstyle{rvert}=[style=rond,style=vert]
\tikzstyle{rrouge}=[style=rond,style=rouge]
\tikzstyle{rgris}=[style=rond,style=gris]
\tikzstyle{cjaune}=[style=carre,style=jaune]
\tikzstyle{cbleu}=[style=carre,style=bleu]
\tikzstyle{cvert}=[style=carre,style=vert]
\tikzstyle{crouge}=[style=carre,style=rouge]
\tikzstyle{cgris}=[style=carre,style=gris]
\tikzstyle{rjaunef}=[style=rond,style=jaunef]
\tikzstyle{rbleuf}=[style=rond,style=bleuf]
\tikzstyle{rvertf}=[style=rond,style=vertf]
\tikzstyle{rrougef}=[style=rond,style=rougef]%

\usetikzlibrary{shadows}
\tikzset{
diagonal fill/.style 2 args={fill=#2, path picture={
\fill[#1, sharp corners] (path picture bounding box.south west) -|
                         (path picture bounding box.north east) -- cycle;}},
reversed diagonal fill/.style 2 args={fill=#2, path picture={
\fill[#1, sharp corners] (path picture bounding box.north west) |- 
                         (path picture bounding box.south east) -- cycle;}}
}

\usetikzlibrary{hobby,backgrounds,calc,trees}

\usepackage{csquotes} 
\usepackage{cite} 
\usepackage[most]{tcolorbox}
\newcommand{\complexityclass}[1]{\textnormal{\textsc{#1}}} 
\newcommand{\ExpTime}{\complexityclass{ExpTime}} 
\newcommand{\deff}{\coloneqq}
\newcommand{\N}{\mathbb{N}} 
\renewcommand{\iff}{\leftrightarrow} 
\newcommand{\cupdot}{\mathbin{\mathaccent\cdot\cup}} 

\newcommand{\powerset}[1]{\mathcal{P}(#1)}

\makeatletter 
\newcommand{\hathat}[1]{%
\begingroup%
  \let\macc@kerna\z@%
  \let\macc@kernb\z@%
  \let\macc@nucleus\@empty%
  \hat{\mathchoice%
    {\raisebox{.2ex}{\vphantom{\ensuremath{\displaystyle #1}}}}%
    {\raisebox{.2ex}{\vphantom{\ensuremath{\textstyle #1}}}}%
    {\raisebox{.16ex}{\vphantom{\ensuremath{\scriptstyle #1}}}}%
    {\raisebox{.14ex}{\vphantom{\ensuremath{\scriptscriptstyle #1}}}}%
    \smash{\hat{#1}}}%
\endgroup%
}
\makeatother

\newcommand{\Logic}[1]{\ensuremath{\mathrm{#1}}} 
\newcommand{\PDL}{\Logic{PDL}}   
\newcommand{\CPDL}{\Logic{CPDL}}   
\newcommand{\LoopPDL}{\Logic{LPDL}}   
\newcommand{\LoopCPDL}{\Logic{LCPDL}}   
\newcommand{\ourLoopCPDL}{\ensuremath{\mathrm{LCPDL}_{\scriptscriptstyle +}}}   

\newcommand{\str}[1]{\mathfrak{#1}} 
\def\mystr#1{\expandafter\def\csname str#1\endcsname{\str{#1}}} 
\def\mystrall#1{\ifx#1\mystrall\else\mystr#1\expandafter\mystrall\fi}
\mystrall ABCDEFGHIJKLMNOPQRSTUVWXYZ\mystrall

\newcommand{\dom}[1]{\mathrm{#1}} 
\def\mydom#1{\expandafter\def\csname dom#1\endcsname{\dom{#1}}} 
\def\mydomall#1{\ifx#1\mydomall\else\mydom#1\expandafter\mydomall\fi}
\mydomall ABCDEFGHIJKLMNOPQRSTUVWXYZ\mydomall

\def\myrm#1{\expandafter\def\csname rm#1\endcsname{\mathrm{#1}}}
\def\myrmall#1{\ifx#1\myrmall\else\myrm#1\expandafter\myrmall\fi}
\myrmall ABCDEFGHIJKLMNOPQRSTUVWXYZ\myrmall

\newcommand{\pred}[1]{\mathrm{#1}} 
\def\mypred#1{\expandafter\def\csname p#1\endcsname{\pred{#1}}}
\def\mypredall#1{\ifx#1\mypredall\else\mypred#1\expandafter\mypredall\fi}
\mypredall ABCDEFGHIJKLMNOPQRSTUVWXYZabcdefghjklmnopqrstuvwxyz \mypredall

\newcommand{\el}[1]{\mathit{#1}} 
\def\myel#1{\expandafter\def\csname el#1\endcsname{\el{#1}}}
\def\myelall#1{\ifx#1\myelall\else\myel#1\expandafter\myelall\fi}
\myelall abcdefghijkmnopqrstuvwxyz\myelall

\newcommand{\aut}[1]{\mathdutchcal{#1}} 
\def\myaut#1{\expandafter\def\csname aut#1\endcsname{\aut{#1}}} 
\def\myautall#1{\ifx#1\myautall\else\myaut#1\expandafter\myautall\fi}
\myautall ABCDEFGHIJKLMNOPQRSTUVWXYZ\myautall

\newcommand{\mldiamond}[1]{\langle #1 \rangle} 
\newcommand{\mlbox}[1]{\left[ #1 \right]} 
\newcommand{\AP}{\mathbb{P}} 
\newcommand{\PR}{\mathbb{A}} 
\newcommand{\sem}[1]{\llbracket#1\rrbracket} 
\newcommand{\univprog}{\star} 
\newcommand{\inv}[1]{\overline{#1}} 
\newcommand{\lop}{\mathrm{loop}} 
\newcommand{\wdem}{\Rrightarrow}  
\newcommand{\ldem}{\circlearrowleft}  

\newcommand{\homo}[1]{\mathfrak{#1}} 
\newcommand{\homof}{\homo{f}} 
\newcommand{\homoh}{\homo{h}} 

\newcommand{\tr}{\mathrm{tr}}
\newcommand{\vtr}{\mathrm{vtr}}
\newcommand{\Pref}[1]{\mathsf{pref}(#1)} 
\newcommand{\Seg}[1]{\mathsf{seg}(#1)}   

\usepackage[most]{tcolorbox}
\newcommand{\enc}{\homo{en}}
\newcommand{\dec}{\homo{de}}
\newcommand{\NULL}{\pred{NULL}}

\newtcolorbox{procedurebox}[1]{enhanced, breakable,
  drop shadow={black!50!white},
  colback=white, colframe=darkmidnightblue,
  boxrule=0.5pt, left=0.5em, right=0.5em, top=0.2em, bottom=0.2em, boxsep=1pt,
  before upper={\noindent\colorbox{pink}{\bfseries\,#1\,}\par\smallskip}}
\newcommand{\why}[1]{\textcolor{darkgray}{\scriptsize $\vartriangleright$~\emph{#1}}}

\newcommand{\posbool}[1]{\autB^{+}(#1)} 
\newcommand{\dirs}{\mathrm{Dir}} 
\newcommand{\Inf}{\mathrm{Inf}} 
\newcommand{\Lang}{\mathcal{L}} 
\newcommand{\ind}{\mathrm{ind}} 
\newcommand{\treeT}{\autT}
\newcommand{\stay}{\mathrm{stay}}
\newcommand{\up}{\mathrm{up}}
\nolinenumbers

\begin{document}
\maketitle

\setcounter{page}{0}
\begin{abstract}
Propositional Dynamic Logic (PDL) is a well-established modal logic of programs. Among its many extensions, the loop operator stands out for capturing the cyclic behaviour of programs: a world satisfies $\lop(\pi)$ precisely when it can return to itself along a path matching the regular expression~$\pi$. While the satisfiability problem for the resulting logic~$\LoopPDL$ was settled already in the 1980s, the decidability of its \emph{finite} satisfiability problem has remained open ever since.
The core difficulty is inherent to the finite setting: requiring and forbidding cycles at once is easily reconciled over arbitrary models, but over finite ones it defeats the existing finite-model-theoretic techniques.

We make the first substantial progress towards closing this gap. We prove that finite satisfiability is $\ExpTime$-complete for $\LoopPDL$ formul{\ae} (also in the presence of the converse operator) in negation normal form, under the restriction that every \emph{negated} non-atomic loop has the form~$\lop(\pi^+)$ for some arbitrarily complex program~$\pi$---a generalised form of self-reachability; positive occurrences of loops remain unrestricted. At the heart of our approach lies a novel way of representing intended finite models as their infinite, cactus-like unravellings: these are recognisable by two-way tree automata and satisfy a tailored $\omega$-regular foldability condition that lets us finitise the accepted~structures.
\end{abstract}

\clearpage

\section{Introduction}\label{sec:introduction}

\noindent \textbf{PDL.}\,
Propositional Dynamic Logic ($\PDL$) is a well-established modal logic for reasoning about program behaviour. Introduced by Fischer and Ladner~\cite{FischerL79} in the late 1970s, it has since inspired extensive research on the complexity, expressiveness, proof theory, and model theory of dynamic logics; see the surveys by Harel et al.~\cite{HarelEtAlTextbook} and Troquard and Balbiani~\cite{TroquardB23}.
From the perspective of graph databases, $\PDL$ can also be viewed as a unifying framework combining ontology-based data access with expressive graph query languages such as regular path queries.
In essence, $\PDL$ extends classical multi-modal logic with modalities indexed by \emph{programs} $\pi$ (\ie regular path expressions), built from atomic programs (actions) $\pe \in \PR$ and tests $\pP?$ for atomic propositions $\pP \in \AP$.
Semantically, a world $\ela$ of a Kripke structure $\strM$ satisfies $\mldiamond{\pi}\varphi$ if it reaches a world satisfying $\varphi$ along a path whose labels match~$\pi$, where a test $\psi?$ is read as a self-loop on every world satisfying~$\psi$.
Fischer and Ladner showed that $\PDL$ enjoys the finite model property (FMP)~\cite[Thm.~3.2]{FischerL79}---every satisfiable formula has a finite model. Their work, together with that of Pratt~\cite{Pratt80}, pins down its finite and unrestricted satisfiability problems as $\ExpTime$-complete~\cite[Thm.~4.4]{FischerL79}. The lower bound already applies to $\PDL$ with $\mldiamond{\pe}$ and $\mldiamond{\pe^*}$ alone, for a single $\pe \in \PR$, hence to all logics considered here.
Analogous results hold for $\CPDL$ ($\PDL$ with converse)~and~for~many~generalisations~\cite[Thm.~6]{CalvaneseOS16}~\cite[Thm.~3.1]{BednarczykK22}.

\vspace{0.25em}
\noindent \textbf{Extensions of PDL.}\,
Despite its expressive power, $\PDL$ cannot capture cyclic behaviour of programs.  
To~overcome this limitation, Danecki~\cite{DaneckiLoopPDL84} extended $\PDL$ in 1984 with the $\lop(\pi)$ operator, yielding the logic $\LoopPDL$ (also known as \emph{LoopPDL}).
The semantics of this operator is as expected: a world satisfies $\lop(\pi)$ if it can reach itself via a path matching~$\pi$. We note that, in contrast to the $\mathsf{Self}$ operator from the description logic $\mathcal{SROIQ}$~\cite{HorrocksKS06}, programs in loops need not be atomic.
Danecki also established $\ExpTime$-completeness of the satisfiability problem for $\LoopPDL$~\cite[Thm.~1]{DaneckiLoopPDL84} using an automata-theoretic argument based on the ``cactus model property'' of the logic.
This line of work was further generalised by Harel~\cite[p.~182]{Harel83}, who enriched $\PDL$ with program intersection~$\cap$ (parallel execution of programs), a construct that was studied intensively in the early 2000s~\cite{Lange04,LangeL05,Lutz05,GollerLL07,GollerLL09} and also more recently~\cite{FigueiraFB23,BednarczykK25,FigueiraFN26}. 

\vspace{0.25em}
\noindent  \textbf{Finite-model reasoning.}\,
Even tiny fragments of $\LoopPDL$ fail to enjoy the FMP. Indeed, consider $\psi \deff \mlbox{\pr^*} \left( \mldiamond{\pr}\top \land \neg \lop(\pr^+) \right)$,
stating that all reachable worlds have successors, but none lies on a cycle. Clearly $\psi$ has a model (\eg an infinite chain) but no finite models, as in any finite one the chain of $\pr$-successors must eventually produce an $\pr$-cycle.
This motivates the study of the finite satisfiability problem for $\LoopPDL$ and its fragments.
Decidability of finite satisfiability for $\LoopPDL$ (and its generalisations with the $\cap$ operator) has been posed as an open problem in numerous recent works~\cite[p.~94]{BednarczykK25}\cite[p.~26]{FigueiraFN26}\cite[Open~question~1]{FigueiraFB23}\cite[Sec.~3.8]{GoellerThesis}\cite[§7]{GollerLL09}.
Surprisingly, we are not aware of any positive results even for the unimodal logic extended with reachability~$\pr^+$ and self-reachability~$\lop(\pr^+)$.
Informally, the main obstacle can be summarised by Gottlob's famous motto: ``\emph{All the evil in computer science comes from cycles}''.

\vspace{0.5em}
\noindent \textbf{Related work.}
We briefly survey the positive results for rather restricted, somewhat artificial fragments of $\LoopPDL$ that can be extracted from the literature. 
\begin{itemize}\itemsep0em
\item The \emph{atomic-loop fragment} of $\LoopPDL$, \ie $\PDL$ extended with the $\lop$ operator restricted to atomic programs (an $\ExpTime$-complete fragment of $\mathcal{ZOI}$ with the FMP~\cite[Thm.~6]{CalvaneseOS16}).
\item The \emph{transitive fragment} of $\LoopPDL$, where programs are either atomic or interpreted as transitive but not as transitive closures (expressible in the two-variable guarded fragment with transitive guards, which has decidable finite satisfiability~\cite[Thm.~6.11]{KieronskiTendera18}).
\item The \emph{query fragment} of $\LoopPDL$, \ie $\LoopPDL$-formul{\ae} in negation normal form in which loops appear only under negation (decidable even in the broader setting of CRPQ entailment modulo ontologies, but without any complexity bound~\cite[Thm.~1.1]{Gutowski_2026}).
\end{itemize}
To the best of our knowledge, no other fragments of our logic are known to have a decidable finite satisfiability problem, leaving the state of the art highly unsatisfactory. Moreover, none of the fragments discussed supports the dangerous combination of simultaneously enforcing and forbidding cycles: in the first, cycles are limited to self-loops; in the second, they become cliques (over atomic programs!) due to transitivity; and the last one cannot enforce cycles.

\vspace{0.5em}
\noindent \textbf{Our motivation.}\,
Our primary goal is to settle the finite satisfiability problem for $\LoopPDL$, and we see two reasons why it deserves attention.
First, the question has been open for four decades and posed repeatedly in the recent works cited above, yet even its most basic fragments (\eg the aforementioned unimodal logic with ``self-reachability'') have withstood analysis.
Second, and most important to us, the interaction between enforced and forbidden cycles plays out over finite models in ways that the standard finite-model techniques---filtrations, mosaics, and the like---were never designed to control.
Coming to grips with it is, we believe, the real prize: it forces the development of finite-model methods that ought to prove useful well beyond $\LoopPDL$, \eg~for the family of guarded logics with transitive closure~\cite{Michaliszyn09,FigueiraFN26} and more.

\vspace{0.5em}
\noindent \textbf{Our results.}\,
We make the first substantial progress towards deciding finite satisfiability for $\LoopPDL$. Throughout, input formul{\ae} are assumed to be in negation normal form (so tests are atomic).
We study $\ourLoopCPDL$, the fragment of $\LoopCPDL$ (\ie $\LoopPDL$ with converse) in which every \emph{negated} non-atomic loop has the form $\lop(\pi^+)$, a generalised self-reachability condition; here $\pi$ may be arbitrarily complex, and positive loops remain unrestricted.
Restricting to~$\pi^+$ makes our fragment capture precisely the \emph{concatenation-closed} regular languages, a very natural first step towards the ultimate decidability goal---yet, unlike the aforementioned sublogics, one that retains the difficult interaction between cycles.
Moreover, the programs allowed in negated loops generalise (modulo trivial rewriting) $\circ$-free regular expressions, as well as the atoms of the practically relevant \emph{simple} regular path expressions~\cite{FigueiraGKMNT20}. Our~main~result~is:

\begin{theorem*}[Main result]
    The finite satisfiability problem for $\ourLoopCPDL$ is $\ExpTime$-complete.
\end{theorem*}

Our decision procedure rests on a single guiding idea---``\emph{reason in the infinite, but be ready to return to the finite}''.
The closest precursor we know of is Bojańczyk's work on finite satisfiability of the two-way $\mu$-calculus~\cite{Bojanczyk02,BojanczykThesis}, a great inspiration for us (even though his logic can neither enforce nor forbid cycles).
The key idea is to represent finite models by infinite trees recognisable by two-way parity tree automata, intended to capture the cactus-like unravellings of finite structures.
Since not every accepted tree can be obtained by unravelling a finite model, we impose an additional $\omega$-regular \emph{foldability} condition, certifying precisely when a tree folds back into a finite structure.
Unravelling a finite model into an accepted, foldable tree is immediate; folding an accepted tree back into a finite model is the crux, and demands a careful analysis of regular cactus-like models.
The main ingredient is a new notion of \emph{trace}: an invariant maintained throughout the finitisation, recording roughly how many obstacles a world sees along its incoming paths that could prevent its branch from folding.\\
\text{\faHandPointRight}~\emph{Remaining technical details and further guiding examples are deferred~to~the~appendix.}

\section{Preliminaries}\label{sec:preliminaries}

\noindent \textbf{PDL.}\,
We recall the syntax and semantics of $\LoopCPDL$.
Throughout, we fix disjoint sets of \emph{atomic propositions} $\AP \deff \{ \pP, \pQ, \pR, \ldots \}$ and \emph{atomic programs} (actions) $\PR \deff \{ \pe, \pr, \ps, \ldots \}$.
By convention, fraktur letters denote structures, and the corresponding Roman letters denote their domains.
A~\emph{Kripke structure} is a tuple $\strA \deff (\domA, \{ \pe^{\strA} \mid \pe \in \PR \}, \{ \pP^{\strA} \mid \pP \in \AP \})$, where $\domA \deff \{ \elu, \elv, \ldots \}$ is a set of \emph{worlds}, each $\pe^{\strA} \subseteq \domA \times \domA$ is a \emph{transition relation}, and each $\pP^{\strA} \subseteq \domA$~is a \emph{unary relation}.
The~\emph{formul{\ae}} $\varphi, \psi$ and \emph{programs} $\pi, \rho$ of $\LoopCPDL$ are defined by the~grammar:

\begin{align*}
\pi, \rho &\Coloneqq \pe \mid \inv{\pe} \mid \pi \circ \rho \mid \pi \cup \rho \mid \pi^{*} \mid \pP{?} \mid \top{?}, &\\
\varphi, \psi &\Coloneqq \pP \mid \neg \pP \mid \varphi \vee \psi \mid \varphi \land \psi \mid \mldiamond{\pi}\varphi \mid \mlbox{\pi}\varphi \mid \lop(\pi) \mid \neg \lop(\pi). & (\pP \in \AP, \pe \in \PR)
\end{align*}
\noindent We define $\ourLoopCPDL$ by replacing ``$\neg \lop(\pi)$'' in the grammar with ``$\neg \lop(\pe) \mid \neg \lop(\pi^+)$''.
\begin{table}[H]
  \centering
  \begin{tabular}{@{}l@{\ \ }c@{\ \ }r@{}}
    \hline
    Name & Syntax & Semantics $\sem{\cdot}_{\strA}$ of $\cdot$ in a structure $\strA$ \\\hline
    Program/Formula & $\pi$ / $\varphi$ & $\sem{\pi}_{\strA} \subseteq \domA \times \domA$ / $\sem{\varphi}_{\strA} \subseteq \domA$,\\
    Action/Unary Predicate & $\pe$ / $\pP$ & $\pe^{\strA} \subseteq \domA \times \domA$ / $\pP^{\strA} \subseteq \domA$,\\
    Composition/Union & $\pi \oplus \rho$ & $\sem{\pi}_{\strA}\; \oplus \;\sem{\rho}_{\strA}$ for $\oplus \in \{ \circ, \cup \}$,\\
    Kleene star/plus & $\pi^{*}$ / $\pi^{+}$ & (reflexive and) transitive closure of $\sem{\pi}_{\strA}$\\
    Converse action & $\inv{\pe}$ & $\{ (\elu, \elv) \mid (\elv, \elu) \in \pe^{\strA} \}$,\\
    Test / Top / Falsum & $\varphi?$ / $\top$ / $\bot$ & $\{ (\elu, \elu) \mid \elu \in \sem{\varphi}_{\strA} \}$ / $\domA$ / $\emptyset$,\\
    Negation/Conjunction & $\neg \varphi$ / $\varphi \land \psi$ & $\domA \setminus \sem{\varphi}_{\strA}$ / $\sem{\varphi}_{\strA} \cap \sem{\psi}_{\strA}$,\\
    Diamond modality & $\mldiamond{\pi}\varphi$ & $\{ \elu \mid \exists{\elv}\, (\elu, \elv) \in \sem{\pi}_{\strA} \; \land \; \elv \in \sem{\varphi}_{\strA}\}$,\\
    Loop operator & $\lop(\pi)$ & $\{ \elu \mid (\elu, \elu) \in \sem{\pi}_{\strA}\}$,\\\hline
    Automaton program & $\autA_{q,q'}$ & $\{ (\elu, \elv) \mid \exists \text{ accepting run of } \autA_{q, q'} \text{ from } \elu \text{ to } \elv \text{ in } \strA \},$\\
    Universal program & $\univprog$ & $\domA \times \domA$.\\\hline%
  \end{tabular}
  \caption{Semantics of $\LoopCPDL$. Omitted cases: $\varphi \vee \psi \equiv \neg(\neg \varphi \land \neg \psi)$ and $\mlbox{\pi}\varphi \equiv \neg \mldiamond{\pi} \neg \varphi$.}\label{tab:semantics}
\end{table}

\vspace{-\baselineskip}
\noindent Let $\Sigma \deff \{ \pe, \inv{\pe}, \top?, \pP? \mid \pe \in \PR, \pP \in \AP \}$; note that programs are just regular expressions over~$\Sigma$.
We lift $\inv{\cdot}$ to $\Sigma$, putting $\inv{(\inv{\pe})} \deff \pe$ and $\inv{(\pP?)} \deff \pP?$ for all $\pe \in \PR$ and $\pP \in \AP$.
The \emph{length}~$|\varphi|$ of a formula~$\varphi$ is the total number of symbols in~$\varphi$.
We write $\strA, \elw \models \varphi$ if $\elw \in \sem{\varphi}_{\strA}$, and call $(\strA, \elw)$ a  \emph{model} of~$\varphi$.
A~formula is \emph{(finitely) satisfiable} if it has a model $(\strA, \elw)$ with (finite)~$\strA$.
In the \emph{(finite) satisfiability problem} for a logic~$\Logic{L}$, we ask if an input $\Logic{L}$-formula is (finitely)~satisfiable.

\vspace{0.35em}
\noindent \textbf{Walks and NFAs.}\,
A~\emph{walk} in $\strA$ is a finite word $\xi \deff \ela_0, \sigma_1, \ela_1, \ldots, \sigma_{n}, \ela_{n}$ in $\domA (\Sigma \cdot \domA)^*$ with $(\ela_i, \ela_{i+1}) \in \sem{\sigma_{i+1}}_{\strA}$ for all $i < n$. 
Here $\xi$ \emph{starts from}~$\ela_0$ and \emph{ends on}~$\ela_{n}$.
Its \emph{reversal} $\inv{\xi}$ is defined as $\ela_n, \inv{\sigma_n}, \ela_{n-1}, \ldots, \inv{\sigma_1}, \ela_0$.
A~\emph{cycle} is a walk with $\ela_{n} = \ela_0$. 
An $\omega$-walk is an infinite word in $\domA (\Sigma \cdot \domA)^\omega$, whose every prefix of odd length is a walk.
A~walk is \emph{stationary} if all of its worlds are equal to~$\ela_0$.
A~\emph{semiautomaton}~\cite{Ginzburg} is a pair $\autA \deff (Q, \Delta)$ where $Q$ is a finite set of states and $\Delta \subseteq Q \times \Sigma \times Q$ is a finite transition relation.
Given states $q, q'$ of $\autA$, we write $\autA_{q, q'}$ for the NFA obtained from $\autA$ by setting the initial and the final state to $q$ and $q'$, respectively.
A~\emph{run} of $\autA_{q, q'}$ over $\strA$ is a finite sequence of configurations $(\ela_0, q_0), (\ela_1, q_1), \ldots, (\ela_\ell, q_\ell)$ in $\domA \times Q$, interleaved with letters $\sigma_1, \ldots, \sigma_\ell \in \Sigma$, such that $(q_i, \sigma_{i+1}, q_{i+1}) \in \Delta$ and $(\ela_i, \ela_{i+1}) \in \sem{\sigma_{i+1}}_{\strA}$ hold for all $i < \ell$.
Such a run \emph{starts from}~$\ela_0$, \emph{ends on}~$\ela_\ell$, and is \emph{accepting} if $q_0 = q$ and $q_\ell = q'$.
With every accepting run of $\autA_{q, q'}$ starting from~$\ela_0$ we associate a corresponding \emph{witnessing walk} $\ela_0, \sigma_1, \ela_1, \ldots, \sigma_\ell, \ela_\ell$ for $\autA_{q, q'}$ and~$\ela_0$.
If $\ela_\ell = \ela_0$, we call such a walk a \emph{witnessing cycle} for~$\ela_0$.
If $\xi$ and $\zeta$ are walks such that $\xi$ ends on the same world on which $\zeta$ starts, we write $\xi\zeta$ for their concatenation.
An NFA is \emph{concatenation-closed} if $\xi\zeta$ is an accepting walk from $\ela$ to~$\elc$ whenever $\xi$ is an accepting walk from $\ela$ to $\elb$ and $\zeta$ is one from $\elb$ to $\elc$. Modulo the presence of the empty word, an NFA is concatenation-closed if it corresponds to $\pi^+$ for a~program~$\pi$.
At the opposite extreme, an NFA is \emph{stationary} if each of its accepted words has at most one non-test letter. The name is justified: all witnessing \emph{cycles} of such an NFA~are~stationary.

\vspace{0.35em}
\noindent \textbf{Trees.}\,
A \emph{tree} $\domT$ is a prefix-closed subset of $\{0,1,\dots, \rmN\}^*$ for a positive $\rmN \in \N$ ($\rmN{+}1$ is called the \emph{arity} of $\domT$). Its members are called \emph{vertices}, and the empty sequence $\varepsilon$ is called the \emph{root}.
The \emph{children} of a vertex~$\elv$ are the vertices of the form $\elv i$ with $0 \leqslant i \leqslant \rmN$, and its \emph{descendants} are those of the form $\elv w$ with $w \in \{0,1,\dots,\rmN\}^{+}$; the notions of a~\emph{parent} and of an~\emph{ancestor} are defined dually.
By a~\emph{neighbour} of $\elv$ we mean either a child or the parent of~$\elv$.
The \emph{depth} of a vertex $\elv$ is just $|\elv|$.
A~(possibly infinite) sequence of vertices $\elv_0, \elv_1, \dots$ in which every $\elv_{i+1}$ is a neighbour of $\elv_i$ is called a \emph{path} (not to be confused with a walk).
Such a path is \emph{downward} if every $\elv_{i+1}$ is a child of $\elv_i$, and it is a \emph{$j$-path} if, for a single~$j \in \N$, every $\elv_{i+1}$ is the $j$-th child of $\elv_i$ (we will often speak of $0$-paths).
A~path~$\xi$ is \emph{anchored} if it has the form $\xi = \elv, \elv j0^{0}, \elv j0^{1}, \ldots, \elv j0^{k}$, for some vertex~$\elv$, non-zero direction~$j$, and $k \in \N$, and its final vertex $\elw \deff \elv j0^{k}$ has no~$0$-child.
We then call~$\elw$ an \emph{anchored descendant} of~$\elv$, and~$\elv$ the \emph{anchor} of~$\elw$.
Each anchored descendant has a unique anchor, and every vertex has at most one anchored descendant in each non-zero direction~$j$.
A~Kripke structure~$\strT$ is \emph{tree-like} if its domain is a tree and, for every $\pe \in \PR$, each pair $(\elu, \elv) \in \pe^{\strT}$ consists of two neighbours or satisfies $\elu = \elv$.
Whenever $\domT$ is a tree, we say that $\strS$ is a \emph{subtree} of $\strT$ if, for some vertex~$\elv \in \domT$, the map $w \mapsto \elv w$ is an isomorphism between $\strS$ and the substructure of $\strT$ induced by $\elv$ and its descendants; and we call $\strT$ \emph{regular} if it has only finitely many~subtrees~up~to~isomorphism.

\vspace{-0.5em}
\begin{figure}[h]
  \centering
  \begin{tikzpicture}[
      x=1cm, y=1cm, scale=1.03, every node/.style={transform shape},
      world/.style   = {circle, draw=black!65, fill=white, line width=0.5pt,
                        inner sep=0pt, minimum size=4.6pt},
      root/.style    = {world, minimum size=5.8pt, line width=0.7pt, draw=black!80},
      tedge/.style   = {->, >=stealth, draw=black!55, semithick},
      zedge/.style   = {->, >=stealth, draw=blue!55!black, thick},
      bedge/.style   = {->, >=stealth, draw=red!70!black, thick,
                        dash pattern=on 3pt off 2pt},
      wlab/.style    = {font=\scriptsize, inner sep=1.5pt, text=black!70},
      elab/.style    = {font=\tiny, inner sep=1.3pt, fill=white, midway},
      leg/.style     = {font=\scriptsize, inner sep=0pt, anchor=west},
      arcL/.style    = {out=196, in=164, looseness=1.2},
      arcR/.style    = {out=-16, in=16,  looseness=2.0},
      spath/.style   = {line width=4.5pt, draw=orange!35, line cap=round, line join=round},
      tag/.style     = {font=\tiny, text=black!75, inner sep=1.5pt, align=center},
      bush/.style    = {draw=black!30, fill=black!6, rounded corners=1pt, line width=0.4pt,
                        drop shadow={shadow xshift=1.5pt, shadow yshift=-1.5pt,
                                   opacity=0.28, fill=black!60}},
    ]
    \def\bushat#1{\draw[bush] (#1.south) -- ++(-0.38,-0.56)
                              -- ++(0.76,0) -- cycle;}
    \def\bushsmall#1{\draw[bush] (#1.south) -- ++(-0.30,-0.45)
                                 -- ++(0.60,0) -- cycle;}

    \node[root]  (r)  at (0,0)     {};
    \node[world] (a1) at (0,-0.90) {};
    \node[world] (a2) at (0,-1.80) {};
    \node[world] (a3) at (0,-2.70) {};

    \node[wlab, anchor=south] at ([yshift=2pt]r.north) {$\varepsilon$};
    \node[wlab, anchor=east]  at ([xshift=-3pt]a1.west) {$1$};
    \node[wlab, anchor=east]  at ([xshift=-3pt]a2.west) {$10$};
    \node[wlab, anchor=north] at ([yshift=-2pt]a3.south) {$100$};

    \draw[tedge] (r)  -- node[elab] {$\pe$}        (a1);
    \draw[zedge] (a2) -- node[elab] {$\pr$}  (a1);
    \draw[zedge] (a2) -- node[elab] {$\pr$}        (a3);
    \draw[bedge] (a3) to[arcL] node[elab] {$\pe$} (r);
    \begin{scope}[on background layer]
      \draw[spath] (r) -- (a1) -- (a2) -- (a3);
    \end{scope}
    \node[tag, anchor=south west] at ([shift={(3pt,1pt)}]r.north east) {anchor};
    \node[tag, anchor=north] at ([yshift=-8pt]a3.south) {anchored descendant\;\;\\(no $0$-child)};

    \node[world] (t1) at (1.15,-1.35) {};
    \draw[tedge] (a1) -- node[elab] {$\pr$} (t1);   \bushat{t1}

    \node[world] (c1) at (1.15,-2.20) {};
    \node[world] (c2) at (1.15,-3.00) {};
    \draw[tedge] (a2) -- node[elab] {$\pe$} (c1);
    \draw[tedge] (c1) -- node[elab] {$\pr$} (c2);
    \draw[bedge] (c1) to[arcR] node[elab] {$\pe$} (c2);

    \node[world] (t0) at (1.90,-0.85) {};
    \draw[tedge] (r) -- node[elab] {$\pe$} (t0);   \bushsmall{t0}

    \node[world] (b1) at (3.15,-0.90) {};
    \node[world] (b2) at (3.15,-1.80) {};
    \node[world] (b3) at (3.15,-2.70) {};

    \node[wlab, anchor=south west] at ([shift={(-1pt,-1pt)}]b1.north east) {$\elv$};
    \node[wlab, anchor=east] at ([xshift=-3pt]b2.west) {$\elv i$};
    \node[wlab, anchor=north] at ([yshift=-2pt]b3.south) {$\elv i 0$};

    \draw[tedge] (r)  -- node[elab] {$\pr$} (b1);
    \draw[tedge] (b1) -- node[elab] {$\pe$} (b2);
    \draw[zedge] (b3) -- node[elab] {$\pe$} (b2);
    \draw[bedge] (b1) to[arcL] node[elab] {$\pr$} (b3);

    \node[world] (t2) at (4.30,-1.35) {};
    \draw[tedge] (t2) -- node[elab] {$\pr$} (b1);   \bushat{t2}
    \node[world] (t3) at (4.30,-2.55) {};
    \draw[tedge] (b2) -- node[elab] {$\pe$} (t3);   \bushat{t3}

    \node[world] (d1) at (5.85,-0.90) {};
    \node[world] (d2) at (5.85,-1.80) {};
    \node[world] (d3) at (5.85,-2.70) {};

    \draw[tedge] (r)  -- node[elab] {$\pe$} (d1);
    \draw[zedge] (d1) -- node[elab] {$\pr$} (d2);
    \draw[zedge] (d3) -- node[elab] {$\pr$} (d2);
    \node[tag, anchor=north] at ([yshift=-3pt]d3.south) {anchored path,\\no back edge};

    \node[world] (t5) at (7.00,-1.35) {};
    \draw[tedge] (d1) -- node[elab] {$\pe$} (t5);   \bushat{t5}
    \node[world] (t4) at (7.00,-2.55) {};
    \draw[tedge] (t4) -- node[elab] {$\pe$} (d2);   \bushat{t4}

    \begin{scope}[shift={(7.85,0)}]
      \draw[zedge] (0,-0.35) -- (0.55,-0.35);
      \node[leg] at (0.68,-0.35) {$0$-edge (to the $0$-th child)};
      \draw[tedge] (0,-1.00) -- (0.55,-1.00);
      \node[leg] at (0.68,-1.00) {tree edge (other child)};
      \draw[bedge] (0,-1.65) -- (0.55,-1.65);
      \node[leg] at (0.68,-1.65) {back edge};
      \draw[bush] (0.28,-2.10) -- (0,-2.55) -- (0.55,-2.55) -- cycle;
      \node[leg] at (0.68,-2.32) {arbitrary tree-like part};
      \draw[spath] (0.03,-2.95) -- (0.52,-2.95);
      \node[leg] at (0.68,-2.95) {anchored path};
    \end{scope}
  \end{tikzpicture}
  \vspace{-1em}
  \caption{A cactus-like structure $\strC$. The sequence $1, 10, 100$ is an example $0$-path; the sequence $\xi \deff \varepsilon, 1, 10, 100$ is an anchored path with $\varepsilon$ being the anchor of $100$~and~$\xi, \varepsilon$ an anchored cycle.}
  \label{fig:cactus}
\end{figure}
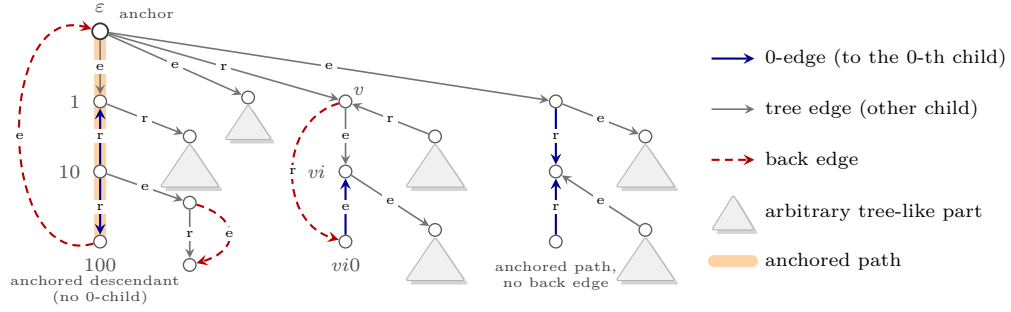

\noindent \textbf{Cacti.}\,
$\strC$ is \emph{cactus-like} if it is obtained from a tree-like structure~$\strT$ by \emph{adding}, for every $\pe \in \PR$ and every pair consisting of an anchored descendant~$\elw$ and its anchor~$\elv$, at most one of $(\elv, \elw)$ and $(\elw, \elv)$ to $\pe^\strT$.
We call the pairs added in this way \emph{back edges}.
Note that a back edge closes a cycle together with the downward path from~$\elv$ to~$\elw$---one non-zero step followed by a (possibly empty) $0$-path.
Consult Figure~\ref{fig:cactus} for a visualisation.
The notion of regularity lifts from tree-like to cactus-like structures (the isomorphism takes also back edges~into~account). In the forthcoming notion of $\varphi$-cacti, a key role is played by \emph{anchored witnessing walks} and \emph{cycles}: these are the ones whose underlying sequence of vertices, obtained after removing consecutive duplicates (as well as the last vertex in the case of cycles), forms~an~anchored~path.

\vspace{0.25em}
\noindent \textbf{Normal Form.}\,
To avoid nested modalities and to track the partial progress of programs within modalities, we define a Scott-like normal form for $\ourLoopCPDL$ over a slightly extended syntax: modalities may contain automata instead of regular expressions, as well as the universal program~$\univprog$ (consult Table~\ref{tab:semantics}).
Fix a semiautomaton~$\autA \deff (Q, \Delta)$.
A formula~$\varphi$ is in \emph{$\autA$-normal form} if it conforms~to~the~pattern:
\[
\pP_{\varphi} \land \mlbox{\univprog} \left( \textstyle\bigwedge_{i=1}^{N_\varphi} \left( \pP_i \iff \psi_i \right) \land \textstyle\bigwedge_{q,q' \in Q} \left( \wdem_{q,q'} \iff \mldiamond{\autA_{q,q'}}\top \right) \land \textstyle\bigwedge_{q,q' \in Q} \left( \ldem_{q,q'} \,\leftrightarrows_{q,q'}\, \lop(\autA_{q,q'}) \right) \right),
\]
where the decorated letters $\pP$, $\wdem$, and $\ldem$ denote atomic propositions (the latter two are the \emph{walk} and \emph{cycle} demands), each $\psi_i$ is a Boolean combination of atomic propositions, and $\leftrightarrows_{q,q'}$ is either $\to$ or $\iff$. Automata with $\iff$ in place of~$\leftrightarrows_{q,q'}$ are called \emph{bannable} (as only they may issue negative cycle demands).
We identify the bannable NFAs purely syntactically while computing the normal form, and record which of them are concatenation-closed. By design of our logics, bannable NFAs are stationary or concatenation-closed. We prove:

\begin{lemma}\label{lemma:normal-form}
  For every $\ourLoopCPDL$-formula~$\varphi$ one can compute in time polynomial in $|\varphi|$ a semiautomaton~$\autA$ and a formula~$\varphi'$ in $\autA$-normal form that is (finitely)~equisatisfiable~with~$\varphi$.
\end{lemma}

\newtcbtheorem{Proviso}{\bfseries Proviso for the rest of the paper}{enhanced,drop shadow={black!50!white},
  coltitle=black, fontupper=\small, fonttitle=\small,
  boxsep=0pt, top=6.5pt, bottom=0.5pt, left=4pt, right=4pt,
  before skip=2pt, after skip=1pt,
  attach boxed title to top left=
  {xshift=1.5em,yshift=-\tcboxedtitleheight/2},
  boxed title style={size=small,colback=pink,boxsep=1pt,top=0.5pt,bottom=0.5pt}
}{Proviso}

\begin{Proviso*}{}{}
  To simplify the presentation, we fix a semiautomaton $\autA \deff (Q, \Delta)$ and an $\ourLoopCPDL$-formula~$\varphi$ in $\autA$-normal form, with its components named as stated above. We use $\Sigma_{\varphi}$ for the set of atomic propositions and atomic programs (with their inverses)~from~$\varphi$.
\end{Proviso*}

\section{Foldable Cacti: A Way of Representing Finite Models}\label{sec:cactus-models}

The classical route to satisfiability in modal logics runs through automata: one represents models by infinite trees (obtained from models via the so-called \emph{unravelling}) and then reduces satisfiability of a formula to non-emptiness of a two-way parity tree automaton (2APTA), as pioneered by Vardi~\cite{VardiW86,Vardi98}.
Two features of our setting resist such representations.
Unravelling removes cycles, but a loop demand may require a cycle of unbounded length; our representation must therefore keep it.
What is more, unravellings are typically infinite, which at first sight makes them useless for finite satisfiability.
We therefore unravel into \emph{cactus-like} structures---trees except for the back edges that turn anchored paths into cycles---and single out, by an appropriate $\omega$-regular foldability condition, precisely those that fold~back~into~a~finite~model.

\begin{definition}[$\varphi$-cactus]\label{def:phi-cactus}
  Fix a bijection $\kappa\colon \{ \ldem_{q,q'}, \wdem_{q,q'} \mid q,q' \in Q \} \to \{ 1, 2, \ldots, 2|Q|^2 \}$, reserving one non-zero direction $\kappa(\delta)$ for each demand~$\delta$.
  A~cactus-like structure~$\strC$ of arity~$2|Q|^2{+}1$ is a \emph{$\varphi$-cactus} if the following hold for all vertices $\elv \in \domC$ and all states $q, q' \in Q$:
    \begin{namedlist}[$(\wdem)\hspace{0.8em}$]
        \item[\desclabel{$(\models)$}{phicactus:bool}] $\varepsilon$ satisfies $\pP_{\varphi}$; every $\elv$ satisfies $\textstyle\bigwedge_{i=1}^{N_\varphi} (\pP_i \iff \psi_i)$ and $\textstyle\bigwedge_{q\in Q} (\ldem_{q,q} \land \wdem_{q,q})$;
        \item[\desclabel{$(\mathbb{F})$}{phicactus:fini}] all $0$-paths in $\strC$ are finite; $\varepsilon$ has no $0$-child; every $\elv \neq \varepsilon$ has a $\Sigma_{\varphi}$-edge to its parent;
\item[\desclabel{$(\wdem)$}{phicactus:walk}] if $\elv$ satisfies $\delta \,\deff\, \wdem_{q,q'}$, it has a witnessing walk $\xi_{\delta}^{\elv}$ for $\autA_{q,q'}$ that is either stationary, or anchored passing through $\elv\cdot\kappa(\delta)$; otherwise it has no witnessing walk for~$\autA_{q,q'}$;
\item[\desclabel{$(\ldem)$}{phicactus:cycle}] if $\elv$ satisfies $\delta \,\deff\, \ldem_{q,q'}$, it has a stationary or anchored witnessing cycle $\xi_{\delta}^{\elv}$ for $\autA_{q,q'}$ passing through $\elv\cdot\kappa(\delta)$; otherwise it has no witnessing cycle, provided $\autA_{q,q'}$~is~bannable;
      \end{namedlist}
\end{definition}

\noindent Comparing the shape of~$\varphi$ with Conditions~\ref{phicactus:bool},~\ref{phicactus:walk}, and~\ref{phicactus:cycle}, we infer that $(\strC,\varepsilon) \models \varphi$ holds for every $\varphi$-cactus~$\strC$. The benefit of $\varphi$-cacti lies in their restricted shape (almost tree-like, and thus amenable to 2APTA running over their tree encodings) and in the way they constrain witnesses (always stationary or anchored).
We next design a suitable unravelling~(see~Figure~\ref{fig:unravelling}).

\vspace{-0.75em}
\begin{figure}[h]
  \centering

  \tikzset{
    wld/.style   = {circle, draw=black!70, fill=white, line width=0.6pt,
                    inner sep=0.5pt, minimum size=16pt, font=\small},
    rt/.style    = {wld, line width=1.1pt, draw=black!90},
    tedge/.style = {->, >=stealth, draw=black!55, semithick},
    zedge/.style = {->, >=stealth, draw=blue!55!black, thick},
    bedge/.style = {->, >=stealth, draw=red!70!black, thick,
                    dash pattern=on 3pt off 2pt},
    elab/.style  = {font=\scriptsize, inner sep=1.4pt, fill=white, midway},
    adr/.style   = {font=\scriptsize, text=black!60, inner sep=1.5pt},
    nam/.style   = {font=\fontsize{34}{36}\selectfont, inner sep=1pt},
    bush/.style  = {draw=black!35, fill=black!8, rounded corners=1pt,
                    line width=0.5pt,
                    drop shadow={shadow xshift=1.5pt, shadow yshift=-1.5pt,
                                 opacity=0.28, fill=black!60}},
    wa/.style    = {fill=green!20!white},
    wb/.style    = {fill=yellow!35!white},
    wc/.style    = {fill=red!15!white},
    aut/.style   = {circle, draw=black!70, fill=white, line width=0.6pt,
                    inner sep=0.5pt, minimum size=15pt, font=\scriptsize},
    panel/.style = {rounded corners=2pt, draw=black!22, fill=black!4,
                    line width=0.5pt},
    alab/.style  = {font=\scriptsize, inner sep=1.4pt, fill=black!4, midway},
    dem/.style   = {font=\scriptsize, text=black!80, inner sep=1.5pt},
    dots/.style  = {font=\scriptsize, text=black!55, inner sep=0pt},
    more/.style  = {font=\normalsize, text=black!55, inner sep=0pt},
  }
  \tikzset{
    unravbox/.style = {x=1cm, y=1cm, scale=.85, baseline={(0,0)},
                       every node/.style={transform shape}},
  }

  \def\bushat#1{\draw[bush] (#1.south) -- ++(-0.45,-0.70)
                            -- ++(0.90,0) -- cycle;
                \node[dots] at ($(#1.south)+(0,-0.44)$) {$\ldots$};}
  \def\bushsmall#1{\draw[bush] ($(#1.south)+(0,-0.06)$) -- ++(-0.30,-0.48)
                               -- ++(0.60,0) -- cycle;}

  \makebox[\textwidth][s]{%
    \begin{tikzpicture}[unravbox]
      \path[use as bounding box] (-0.40,-3.70) rectangle (3.55,0);

      \draw[panel] (-0.40,-3.65) rectangle (3.55,-2.79);
      \node[aut] (q0) at (0.15,-3.22) {$q_0$};
      \node[aut] (q1) at (1.10,-3.22) {$q_1$};
      \node[aut] (q2) at (2.05,-3.22) {$q_2$};
      \node[aut] (q3) at (3.00,-3.22) {$q_3$};
      \draw[tedge] (q0) -- node[alab] {$\pe$} (q1);
      \draw[tedge] (q1) -- node[alab] {$\pr$} (q2);
      \draw[tedge] (q2) -- node[alab] {$\pe$} (q3);

      \node[rt, wa]  (ma) at (1.10,-0.35) {$\ela$};
      \node[wld, wb] (mb) at (2.20,-1.75) {$\elb$};
      \node[wld, wc] (mc) at (0   ,-1.75) {$\elc$};

      \draw[tedge] (ma) -- node[elab] {$\pe$} (mb);
      \draw[tedge] (mb) -- node[elab] {$\pr$} (mc);
      \draw[tedge] (mc) to[bend left=22]  node[elab, pos=0.64] {$\pe$} (ma);
      \draw[tedge] (ma) -- node[elab, pos=0.64] {$\pr$} (mc);
      \draw[tedge] (mb) to[out=-52,in=12,looseness=11] node[elab] {$\pe$} (mb);
    \end{tikzpicture}%
    \hfill
    \begin{tikzpicture}[unravbox]
      \path[use as bounding box] (-0.62,-1.70) rectangle (0.62,0);
      \node[nam] at (0,-1.05) {$\leadsto$};
    \end{tikzpicture}%
    \hfill
    \begin{tikzpicture}[unravbox]
      \path[use as bounding box] (4.30,-4.14) rectangle (11.05,0.30);

      \node[rt, wa]  (ce)  at (7.05, 0)    {$\ela$};
      \node[wld, wb] (c1)  at (5.65,-1.40) {$\elb$};
      \node[wld, wb] (c2)  at (7.50,-1.40) {$\elb$};
      \node[wld, wc] (c3)  at (9.45,-1.40) {$\elc$};
      \node[wld, wc] (c10) at (4.90,-2.80) {$\elc$};
      \node[wld, wc] (c12) at (6.35,-2.80) {$\elc$};

      \node[adr, anchor=east]       at ([xshift=-6pt]ce.west)          {$\varepsilon$};
      \node[adr, anchor=east]       at ([xshift=-4pt]c1.west)          {$1$};
      \node[adr, anchor=east]       at ([xshift=-3pt]c2.west)          {$4$};
      \node[adr, anchor=west]       at ([xshift=3pt]c3.east)           {$7$};
      \node[adr, anchor=north east] at ([shift={(-2pt,-2pt)}]c10.west) {$10$};
      \node[adr, anchor=north west] at ([shift={(2pt,-2pt)}]c12.east)  {$12$};

      \draw[tedge] (ce) -- node[elab] {$\pe$} (c1);
      \draw[tedge] (ce) -- node[elab, pos=0.42] {$\pr$} (c3);
      \draw[tedge] (ce) -- node[elab] {$\pe$} (c2);
      \draw[zedge] (c1) -- node[elab] {$\pr$} (c10);
      \draw[tedge] (c1) -- node[elab] {$\pr$} (c12);
      \draw[bedge] (c10) to[out=176,in=152,looseness=1.15] node[elab, pos=0.72] {$\pe$} (ce);
      \draw[bedge] (c3)  to[bend left=16] node[elab, pos=0.30] {$\pe$} (ce);
      \draw[tedge] (c1)  to[out=158,in=98,looseness=11] node[elab] {$\pe$} (c1);
      \draw[tedge] (c2)  to[out=-52,in=12,looseness=11] node[elab] {$\pe$} (c2);

      \bushat{c10} \bushat{c12} \bushat{c3} \bushat{c2}

      \node[more] (bmore) at (5.63,-2.92) {$\cdots$};
      \draw[tedge] (c1) -- (bmore);
      \bushsmall{bmore}

      \node[more] (cmore) at (10.65,-1.52) {$\cdots$};
      \draw[tedge] (ce) to[bend left=11] (cmore);
      \bushsmall{cmore}

      \node[dem] at (5.15,-4.02) {$\ldem_{q_0,q_3}$};
      \node[dem] at (6.35,-4.02) {$\wdem_{q_1,q_2}$};
      \node[dem] at (7.50,-2.62) {$\wdem_{q_0,q_1}$};
      \node[dem] at (9.45,-2.62) {$\ldem_{q_1,q_3}$};
    \end{tikzpicture}%
  }
  \vspace{-1em}
  \caption{Unravelling a model into a $\varphi$-cactus~$\strC$, for the
  semiautomaton~$\autA$. Each vertex of~$\strC$ carries its $\homof$-image inside and its
  address in grey; tags name realised demands, dots elide further~branches.}
  \label{fig:unravelling}
\end{figure}
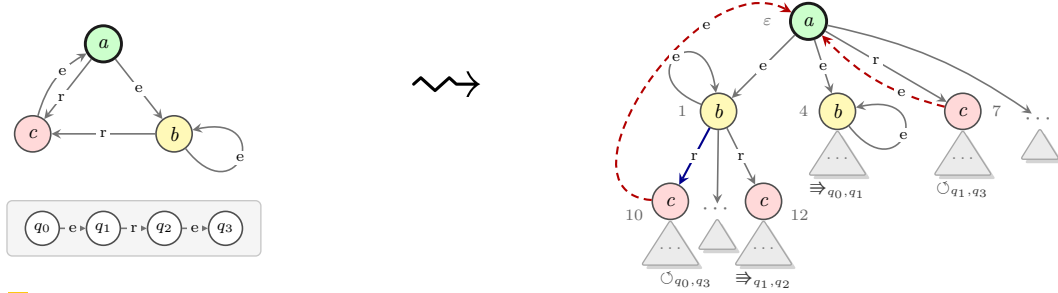

\vspace{-1.5em}
\begin{lemma}\label{lemma:model-to-cactus}
  From every model $(\strA, \ela)$ of $\varphi$ one can construct a $\varphi$-cactus~$\strC$ with $(\strC, \varepsilon) \models \varphi$ and a homomorphism $\homof \colon \strC \to \strA$ satisfying $\homof(\varepsilon) = \ela$.
  Moreover, if $\strA$ is finite, then $\strC$ is regular.
\end{lemma}
\begin{proof}[Proof sketch]
For all $\elu \in \domA$ and $q, q' \in Q$, fix (if they exist) a witnessing walk $\xi^{\elu}_{\wdem_{q,q'}}$ and cycle $\xi^{\elu}_{\ldem_{q,q'}}$ for $\autA_{q,q'}$ and~$\elu$.
We construct the unravelling $\strC$ with a map $\homof \colon \domC \to \domA$, starting from $\domC \deff \{\varepsilon\}$ and $\homof(\varepsilon) \deff \ela$, by repeatedly processing vertices: a vertex stays \emph{active} until processed, in first-in first-out order; $\strC$ is the limit of this run.
See the appendix for correctness.
\noindent Processing an active vertex~$\elv$ takes two stages, after which $\elv$ becomes processed and every fresh vertex active.
First, $\elv$ inherits from~$\homof(\elv)$ its atomic propositions and its atomic self-loops.
Second, we attach a \emph{branch} for every demand~$\delta$ satisfied by $\homof(\elv)$ whose fixed witness $\xi_{\delta}^{\homof(\elv)}$ is not stationary.
Let $\ela_0, \sigma_1, \ldots, \sigma_{n}, \ela_{n}$ be the \emph{reduction} of that witness, obtained by deleting every world that is equal to its predecessor (with a letter preceding it). Note that $\ela_0 = \homof(\elv)$ and $n \geqslant 1$.
Put $j \deff \kappa(\delta)$, $\elv_0 \deff \elv$ and $\elv_i \deff \elv j 0^{i-1}$, and add the fresh vertices $\elv_1, \ldots, \elv_{m}$ with $\homof(\elv_i) \deff \ela_i$, where $m \deff n$ for a walk demand and $m \deff n{-}1$ for a cycle demand (then $\ela_{n} = \ela_0$, and $\elv_{n}$ denotes~$\elv$).
As consecutive $\ela_i$ differ, each $\sigma_i$ is an action~$\pe$ or its converse~$\inv{\pe}$; for every $i < n$ we add to $\pe^{\strC}$ the unique pair making $(\elv_i, \elv_{i+1}) \in \sem{\sigma_{i+1}}_{\strC}$.
For a cycle demand, the last of these pairs joins~$\elv$ with its anchored descendant~$\elv j 0^{m-1}$ and is thus a back edge.
\end{proof}

\noindent From now on, it will be convenient to let a walk discharge a cycle demand in one~step~rather than spell out the required cycle. We thus pass to the \emph{extension} $\hat{\autA} \deff(Q, \hat{\Delta})$ of $\autA$, where~$\hat{\Delta}$~is $\Delta \cup \{ (q, \ldem_{q,q'}?, q') \mid q, q' \in Q \}$. In models of $\varphi$ this is harmless:~for~all~worlds~$\elu, \elv$ and~$q, q' \in Q$, the NFA $\autA_{q,q'}$ has a witnessing walk from $\elu$ to $\elv$ iff $\hat{\autA}_{q,q'}$ does. Indeed, $\Delta \subseteq \hat{\Delta}$; conversely, replace every fresh step $\ela, \ldem_{p,p'}?, \ela$ by a witnessing cycle for $\autA_{p,p'}$ at $\ela$, which exists~as~$\varphi$~holds.

\vspace{0.5em}
\noindent \textbf{Encoding cacti.}\,
We will decide the existence of $\varphi$-cacti~$\strC$ by testing non-emptiness of a carefully designed 2APTA running over their tree~encodings~$\enc(\strC)$.
These are complete infinite $(2|Q|^2{+}1)$-ary $2^{\Gamma_{\varphi}}$-vertex-labelled trees for
$\Gamma_{\varphi} \deff \Sigma_{\varphi} \cup \Gamma^{\mathrm{dir}}_{\varphi} \cup \{ \NULL, \pred{slf}_{\sigma}, \pred{bck}_{\sigma}, \pred{bck}_{\inv{\sigma}} \mid \sigma \in \Sigma_{\varphi} \cap \PR \}$ and \emph{direction letters} $\Gamma^{\mathrm{dir}}_{\varphi} \deff \{ \pD_{\varepsilon}, \pD_{0}, \ldots, \pD_{2|Q|^2} \}$.
In our tree~encoding~$\enc$, a~vertex~$\elv$ carries its atomic propositions and the letters recording the actions between its parent and itself; its direction letter, \ie $\pD_{i}$ if it is the $i$-th child of its parent and $\pD_{\varepsilon}$ if it is the root; a letter $\pred{slf}_{\sigma}$ for each self-loop $(\elv,\elv) \in \sem{\sigma}_{\strC}$; and, if $\elv$ is an anchored descendant, a letter $\pred{bck}_{\sigma}$ for its $\sigma$-back edge (this determines the edge uniquely, since anchors are unique). Vertices outside the intended domain carry $\NULL$ alone. 
Figure~\ref{fig:encoding} visualises~$\enc$,~which~the~appendix~defines~in~full.

\vspace{-0.75em}
\begin{figure}[h]
  \centering
  \tikzset{
    wld/.style   = {circle, draw=black!70, fill=white, line width=0.7pt,
                    inner sep=0pt, minimum size=7.5pt},
    rt/.style    = {wld, minimum size=9.5pt, line width=1pt, draw=black!85,
                    fill=black!12},
    nul/.style   = {circle, draw=black!45, densely dashed, fill=white,
                    line width=0.6pt, inner sep=0pt, minimum size=6pt},
    tedge/.style = {->, >=stealth, draw=black!55, semithick},
    zedge/.style = {->, >=stealth, draw=blue!55!black, thick},
    bedge/.style = {->, >=stealth, draw=red!70!black, thick,
                    dash pattern=on 3pt off 2pt},
    elab/.style  = {font=\scriptsize, inner sep=1.5pt, fill=white},
    adr/.style   = {font=\footnotesize, text=black!55, inner sep=1.6pt},
    prp/.style   = {font=\footnotesize, text=black!85, inner sep=1.6pt},
    lbl/.style   = {font=\scriptsize, inner sep=1.5pt, text=black!85},
    nlbl/.style  = {font=\scriptsize, inner sep=1.5pt, text=black!50},
    bush/.style  = {draw=black!30, fill=black!7, rounded corners=1pt,
                    line width=0.5pt},
    nbush/.style = {draw=black!40, densely dashed, fill=white,
                    rounded corners=1pt, line width=0.5pt},
    tlink/.style = {draw=black!55, semithick},
    zlink/.style = {draw=blue!55!black, thick},
    slink/.style = {draw=black!40, semithick},
    nlink/.style = {draw=black!45, densely dashed, semithick},
  }
  \def\dirhl#1{\textcolor{orange!75!black}{\mathbf{#1}}}
  \def\dirlet#1{\textcolor{orange!75!black}{#1}}
  \def\bushat#1{\draw[bush] ($(#1.south)+(0,-0.03)$) -- ++(-0.36,-0.58)
                            -- ++(0.72,0) -- cycle;}
  \def\subtree#1#2#3{%
    \coordinate (sbt) at ($(#1)+(#2,#3)$);
    \draw[slink] (#1) -- (sbt);
    \draw[bush] (sbt) -- ++(-0.36,-0.58) -- ++(0.72,0) -- cycle;}
  \def\nullkid#1#2#3{%
    \coordinate (nkd) at ($(#1)+(#2,#3)$);
    \draw[nlink] (#1) -- (nkd);
    \node[nul] at (nkd) {};
    \draw[nbush] ($(nkd)+(0,-0.13)$) -- ++(-0.20,-0.32) -- ++(0.40,0) -- cycle;}
  \makebox[\textwidth]{%
    \begin{tikzpicture}[x=1cm, y=1cm, baseline, every node/.style={transform shape}]
      \node[rt]  (r)    at (1.70, 0)    {};
      \node[wld] (v1)   at (0.35,-1.20) {};
      \node[wld] (v2)   at (3.05,-1.20) {};
      \node[wld] (v20)  at (3.05,-2.40) {};
      \node[wld] (v200) at (3.05,-3.60) {};
      \bushat{v1}
      \subtree{r}{-0.35}{-1.20}
      \subtree{v2}{-1.30}{-0.90}
      \subtree{v20}{-1.30}{-0.90}

      \draw[tedge] (v1)   -- node[elab, pos=0.45] {$\ps$} (r);
      \draw[tedge] (r)    -- node[elab, pos=0.45] {$\pe$} (v2);
      \draw[zedge] (v2)   -- node[elab, midway]   {$\pr$} (v20);
      \draw[zedge] (v200) -- node[elab, midway]   {$\ps$} (v20);
      \draw[tedge] (v20)  to[out=-28,in=28,looseness=9] (v20);
      \node[elab, anchor=west] at ([xshift=13pt]v20.east) {$\pe$};
      \draw[bedge] (v200) to[out=-6,in=-20,looseness=1.55]
                   node[elab, pos=0.35] {$\pe$} (r);

      \node[adr, anchor=south]      at ([yshift=2pt]r.north)                {$\varepsilon$};
      \node[adr, anchor=east]       at ([xshift=-2pt]v1.west)               {$1$};
      \node[adr, anchor=east]       at ([xshift=-2pt]v2.west)               {$2$};
      \node[adr, anchor=east]       at ([xshift=-2pt]v20.west)              {$2\dirhl{0}$};
      \node[adr, anchor=north]      at ([yshift=-2pt]v200.south)            {$2\dirhl{00}$};

      \node[prp, anchor=south west] at ([shift={(3pt,-1pt)}]r.north east)   {$\pP, \ldem_{q,q'}$};
      \node[prp, anchor=north west] at ([shift={(3pt,1pt)}]v2.south east)   {$\pQ$};
      \node[prp, anchor=east]       at ([xshift=-5pt]v200.west)             {$\pP$};
    \end{tikzpicture}%
    \hfill
    \raisebox{-2.5cm}{\LARGE\scalebox{2.5}{$\leadsto$}}%
    \hfill
    \begin{tikzpicture}[x=1cm, y=1cm, baseline, every node/.style={transform shape}]
      \node[rt]  (u)    at (1.45, 0)    {};
      \node[wld] (u1)   at (-0.15,-1.20) {};
      \node[wld] (u2)   at (3.05,-1.20) {};
      \node[wld] (u20)  at (3.05,-2.40) {};
      \node[wld] (u200) at (3.05,-3.60) {};
      \bushat{u1}
      \subtree{u}{-0.10}{-1.30}
      \subtree{u2}{-1.30}{-0.90}
      \subtree{u20}{-1.30}{-0.90}

      \draw[tlink] (u)   -- (u1);
      \draw[tlink] (u)   -- (u2);
      \draw[zlink] (u2)  -- (u20);
      \draw[zlink] (u20) -- (u200);

      \nullkid{u}{0.55}{-0.95}
      \nullkid{u1}{0.72}{-0.85}
      \nullkid{u2}{0.65}{-0.45}
      \nullkid{u20}{0.65}{-0.45}
      \node[nlbl] at (1.65,-0.95) {$\ldots$};
      \node[nlbl] at (0.92,-2.05) {$\ldots$};
      \node[nlbl] at (4.05,-1.65) {$\ldots$};
      \node[nlbl] at (4.05,-2.85) {$\ldots$};
      \draw[nlink] (u200) -- ++(0,-0.42);
      \node[nlbl]              at (3.05,-4.22) {$\ldots$};
      \node[nlbl, anchor=west] at (3.32,-4.22) {$\{\NULL\}$};

      \node[adr, anchor=south]      at ([yshift=2pt]u.north)               {$\varepsilon$};
      \node[adr, anchor=east]       at ([xshift=-2pt]u1.west)              {$1$};
      \node[adr, anchor=east]       at ([xshift=-2pt]u2.west)              {$2$};
      \node[adr, anchor=east]       at ([xshift=-2pt]u20.west)             {$2\dirhl{0}$};
      \node[adr, anchor=east]       at ([xshift=-2pt]u200.west)            {$2\dirhl{00}$};

      \node[lbl, anchor=west] at ([xshift=5pt]u.east)    {$\{\dirlet{\pD_{\varepsilon}}, \pP, \ldem_{q,q'}\}$};
      \node[lbl, anchor=west] at ([xshift=5pt]u1.east)   {$\{\dirlet{\pD_{1}}, \inv{\ps}\}$};
      \node[lbl, anchor=west] at ([xshift=5pt]u2.east)   {$\{\dirlet{\pD_{2}}, \pe, \pQ\}$};
      \node[lbl, anchor=west] at ([xshift=5pt]u20.east)  {$\{\dirlet{\pD_{0}}, \pr, \pred{slf}_{\pe}\}$};
      \node[lbl, anchor=west] at ([xshift=5pt]u200.east) {$\{\dirlet{\pD_{0}}, \inv{\ps}, \pP, \pred{bck}_{\pe}\}$};
    \end{tikzpicture}%
  }
  \vspace{-1em}
  \caption{A visual representation of a $\varphi$-cactus~$\strC$ together with its tree encoding $\enc(\strC)$.}
  \label{fig:encoding}
\end{figure}
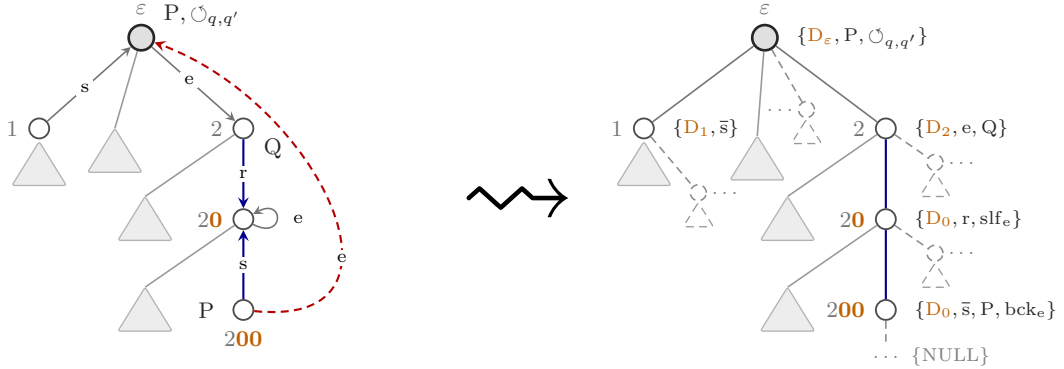

\vspace{-1.25em}
\begin{lemma}\label{lemma:cactus-automaton}
  There exist a map~$\enc$ from $\varphi$-cacti to complete infinite $(2|Q|^2{+}1)$-ary $2^{\Gamma_{\varphi}}$-vertex-labelled trees and a 2APTA~$\autT_{\varphi}$ that accepts exactly the trees~$\enc(\strC)$ for $\varphi$-cacti~$\strC$.~Also, both the construction and the non-emptiness test for $\autT_{\varphi}$ can be done in time~exponential~in~$|\varphi|$.
\end{lemma}

The construction of~$\enc$ and~$\autT_{\varphi}$ (deferred to the appendix) is mostly classical, following the logic-to-automata translations for the $\mu$-calculus~\cite{Vardi98} and similar logics. The only tricky part is the verification of the witnessing cycle condition \ref{phicactus:cycle}. Fortunately, in $\varphi$-cacti every non-trivial cycle decomposes into more ``basic'' ones, readily detectable by tree automata~and~$\hat{\autA}$. The $\ldem$ labels witnessing such cycles then let us simulate their traversal, certifying that no forbidden cycle arises anywhere, \ie no cycle witnessing a bannable~$\autA_{q,q'}$~at~an~element~that~fails~$\ldem_{q,q'}$.
\begin{remark}\label{remark:general-loop-pdl}
  Nothing so far in this section is specific to our fragment (taking all $\leftrightarrows$ to be $\iff$). Lemma~\ref{lemma:cactus-automaton} and its proof apply verbatim to full~$\LoopPDL$, giving a modern and far more detailed rendering of Danecki's sketched $\ExpTime$-completeness proof for general satisfiability~\cite{DaneckiLoopPDL84}.
\end{remark}

All the developments so far concern the general satisfiability problem rather than the finite one.
Worse still, Lemma~\ref{lemma:model-to-cactus} does not preserve finiteness: applied to a finite model $(\strA, \ela) \models \varphi$, it returns a cactus~$\strC$ that is almost always infinite.
A natural remedy is to ``fold'' $\strC$ back into a finite model: cut its infinite branches and redirect the resulting dangling edges to ancestors.
Intuitively this should work, as along every infinite path of~$\strC$ some world of~$\strA$ recurs infinitely often.
Unfortunately, it turns out that not every $\varphi$-cactus folds into~a~finite~model~(Fig.~\ref{fig:no-folding}).

\vspace{-0.25em}
\begin{figure}[h]
  \centering

  \tikzset{
    wld/.style     = {circle, draw=black!70, fill=white, line width=0.6pt,
                      inner sep=0.5pt, minimum size=17pt, font=\scriptsize},
    rt/.style      = {wld, line width=1.1pt, draw=black!90},
    aut/.style     = {circle, draw=black!70, fill=white, line width=0.6pt,
                      inner sep=0.5pt, minimum size=17pt, font=\scriptsize},
    tedge/.style   = {->, >=stealth, draw=black!55, semithick},
    cutedge/.style = {->, >=stealth, draw=black!28, semithick},
    bedge/.style   = {->, >=stealth, draw=red!70!black, thick,
                      dash pattern=on 3pt off 2pt},
    cutmark/.style = {draw=red!70!black, line width=1pt, line cap=round},
    elab/.style    = {font=\tiny, inner sep=1.3pt, fill=white, midway},
    alab/.style    = {font=\tiny, inner sep=1.3pt, fill=black!4, midway},
    panel/.style   = {rounded corners=2pt, draw=black!22, fill=black!4,
                      line width=0.5pt},
    dem/.style     = {font=\scriptsize, text=black!80, inner sep=1.5pt},
    bad/.style     = {font=\scriptsize, text=red!70!black, inner sep=1.5pt},
    dots/.style    = {font=\normalsize, text=black!55, inner sep=0pt},
    upar/.style    = {-{Stealth[length=1.4pt, width=1.2pt]}, draw=black!55,
                      line width=0.3pt},
    updots/.style  = {text=black!55, inner sep=0pt, scale=0.5},
    nfbox/.style   = {x=1cm, y=1cm, baseline={(0,0)},
                      every node/.style={transform shape}},
  }

  \makebox[\textwidth][s]{%
    \begin{tikzpicture}[nfbox]
      \node[font=\small, anchor=west, inner sep=0pt] (f) at (0,0.25)
        {$\psi \deff \mlbox{\univprog} \left( \wdem_{q,q'} \land \neg \ldem_{q,q''} \right) \land\, \ldots$};
      \coordinate (P) at ([xshift=0.75cm]f.east);

      \draw[panel] ([yshift=-0.783cm]P) rectangle ([xshift=3.00cm,yshift=0.774cm]P);
      \node[aut] (q)   at ([xshift=0.45cm,yshift=-0.235cm]P) {$q$};
      \node[aut] (qp)  at ([xshift=1.31cm,yshift=0.345cm]P)  {$q'$};
      \node[aut] (qpp) at ([xshift=1.93cm,yshift=-0.235cm]P) {$q''$};

      \draw[tedge] (q)   -- node[alab] {$\pr$} (qp);
      \draw[tedge] (q)   -- node[alab] {$\pr$} (qpp);
      \draw[tedge] (qpp) to[out=-40,in=40,looseness=7] node[alab] {$\pr$} (qpp);
      \node[dem, anchor=north west] at ([xshift=0.10cm,yshift=0.744cm]P) {$\autA$};
    \end{tikzpicture}%
    \hfill
    \begin{tikzpicture}[nfbox]
      \node[rt]   (v0) at (0,0)    {$\varepsilon$};
      \node[wld]  (v1) at (1.26,0) {$1$};
      \node[wld]  (v2) at (2.52,0) {$11$};
      \node[wld]  (v3) at (3.78,0) {$111$};
      \node[dots] (vd) at (4.90,0) {$\cdots$};

      \foreach \v/\arrs in {%
          v0/{-0.1/132,-0.05/108,0/88,0.05/65,0.1/44},
          v1/{-0.15/150,-0.1/132,-0.05/108,0/88,0.05/65},
          v2/{-0.1/132,-0.05/108,0/88,0.05/65,0.1/44},
          v3/{-0.1/132,-0.05/108,0/88,0.05/65,0.1/44}} {
        \foreach \dx/\ang in \arrs
          \draw[upar] ([xshift=\dx cm,yshift=0.04cm]\v.north)
                      -- ++(\ang:0.12);
        \node[updots] at ([yshift=0.2cm]\v.north) {$\ldots$};
      }

      \draw[tedge] (v0) -- node[elab] {$\pr$} (v1);
      \draw[tedge] (v1) -- node[elab] {$\pr$} (v2);
      \draw[tedge] (v2) -- node[elab] {$\pr$} (v3);

      \draw[cutedge] (v3) -- (vd);
      \draw[cutmark] (4.26,-0.17) -- (4.52,0.17);
      \draw[cutmark] (4.26, 0.17) -- (4.52,-0.17);

      \draw[bedge] (v3) to[out=118,in=62,looseness=0.6] node[elab] {$\pr$} (v1);
      \node[bad, anchor=south] at (2.52,0.65)
        {closes an $\pr^{+}$-cycle~\faSkullCrossbones};

      \node[dem, anchor=north] at (2.32,-0.36)
        {every world satisfies $\wdem_{q,q'}$ and $\neg \ldem_{q,q''}$};
    \end{tikzpicture}%
  }
  \vspace{-1em}
  \caption{A satisfiable but not finitely satisfiable $\psi$ in $\autA$-normal form, with an $\autA$ and a $\psi$-cactus.}
  \label{fig:no-folding}
\end{figure}
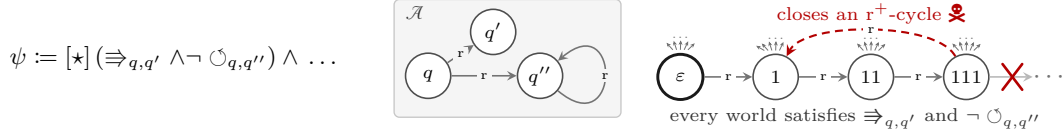

\vspace{-0.5em}

\noindent Figure~\ref{fig:no-folding} revisits $\mlbox{\pr^*} \left( \mldiamond{\pr}\top \land \neg \lop(\pr^+) \right)$: its $\autA$-normal form~$\psi$ (left, with~$\autA$ in the middle) has no finite model, yet it has the $\psi$-cactus model $(\strC, \varepsilon)$ containing an infinite chain $1^*$ of $\pr$-successors (right).
Cut this chain at a vertex~$1^{n}$ and redirect its outgoing $\pr$-edge to an ancestor.
As no vertex of the chain carries~$\ldem_{q,q''}$, each of them forbids an $\pr^{+}$-cycle, yet the redirected edge closes precisely such a cycle. This contradicts the conjunct $\ldem_{q,q''} \iff \lop(\autA_{q,q''})$ of~$\psi$.
Hence no fold of~$\strC$ is a model of~$\psi$.
This exposes the main obstacle to finite satisfiability~of~$\LoopPDL$: a $\varphi$-cactus may have an infinite branch that no cut can fold, its cycle demands never permitting the resulting cycle.
Ruling out such $\omega$-walks is precisely the job of \emph{foldable} $\varphi$-cacti (only concatenation-closed automata matter here, as cycles~for~stationary~ones~are~treated~``locally'').

\begin{definition}[Foldability]\label{def:foldability}
  A \emph{bad $(q,q')$-walk} is a witnessing walk for $\hat{\autA}_{q,q'}$ whose ends both fail $\ldem_{q,q'}$.
  An $\omega$-walk is \emph{bad} if it factorises as $\xi_1 \xi_2 \ldots$ with all $\xi_i$ or all~$\inv{\xi_i}$~bad $(q,q')$-walks, for some bannable concatenation-closed $\autA_{q,q'}$.
  A structure is \emph{foldable} if~it~lacks~bad~$\omega$-walks.
\end{definition}
For a concatenation-closed $\autA_{q,q'}$, repeating its forbidden cycle keeps it forbidden, so the cycle unfolds into a bad $\omega$-walk.
Moreover, as cycles can be traversed backwards, the definition of bad $\omega$-walks also covers their reversals.
Consequently, a procedure finitising $\varphi$-cacti while maintaining foldability never creates a forbidden cycle, and hence likely preserves satisfaction of~$\varphi$.
Maintaining foldability is the key technical challenge. We meet it with \emph{traces}~(Sec.~\ref{sec:traces}) and the finitisation procedure (Sec.~\ref{sec:soundness}), roughly as in~Fig.~\ref{fig:no-folding} but with the mild twist that some extra edges are needed to regain the witnessing cycles. Before that, we show how foldability applies to finite models and their cacti, and how it yields a decision~procedure~for~$\ourLoopCPDL$.
As some world repeats along every $\omega$-walk of a finite model, closing a cycle,~we~can~prove:
\begin{lemma}\label{lemma:if-finite-then-cactus-foldable}
  For finite $\strA$, the $\varphi$-cactus~$\strC$ constructed in Lemma~\ref{lemma:model-to-cactus} is regular and foldable.
\end{lemma}

\noindent What is more, foldability is an $\omega$-regular property of tree encodings, and thus we can prove:
\begin{lemma}\label{lemma:foldable-cactus-automaton}
  Lemma~\ref{lemma:cactus-automaton} holds with ``foldable $\varphi$-cactus'' in place of ``$\varphi$-cactus''.
\end{lemma}

\noindent We can now finally state our main technical contribution: its ``if'' direction is immediate from Lemma~\ref{lemma:if-finite-then-cactus-foldable}, whereas the ``only if'' part is much harder and occupies the~remainder~of~the~paper.
\begin{theorem}\label{thm:main-characterisation}
A regular foldable $\varphi$-cactus exists if and only if $\varphi$ is finitely satisfiable.
\end{theorem}

\noindent With Theorem~\ref{thm:main-characterisation} in hand, we reap the benefits of our automata-based machinery and obtain the following exponential-time decision procedure for finite satisfiability of $\ourLoopCPDL$. 
Its correctness follows from Theorem~\ref{thm:main-characterisation}: every regular foldable $\varphi$-cactus is encoded by a tree accepted by $\autT^{f}_{\varphi}$ and, conversely, if the language of $\autT^{f}_{\varphi}$ is non-empty, then by Rabin's basis 
theorem~\cite[Thm.~3.7]{PerrinPin04InfiniteTrees} it contains a \emph{regular} tree, which encodes a regular foldable $\varphi$-cactus.

\begin{procedurebox}{Deciding finite satisfiability of $\ourLoopCPDL$ in $\ExpTime$}
\begin{namedlist}[3.]
  \item[\textbf{1.}] Given an input $\ourLoopCPDL$-formula~$\varphi'$, compute a semiautomaton~$\autA$ and a formula~$\varphi$ in $\autA$-normal form. \why{Lemma~\ref{lemma:normal-form}; polynomial time, and $|\varphi|$ is polynomial in~$|\varphi'|$.}
  \item[\textbf{2.}] Construct a 2APTA~$\autT^{f}_{\varphi}$ recognising the encodings of foldable $\varphi$-cacti. \why{Lemmas \ref{lemma:cactus-automaton}\&\ref{lemma:foldable-cactus-automaton}.}
  \item[\textbf{3.}] Return \textsc{Yes} if $\autT^{f}_{\varphi}$ has non-empty language, and \textsc{No} otherwise. \why{Lemma~\ref{lemma:foldable-cactus-automaton}; exp.~time~in~$|\varphi|$.}
\end{namedlist}
\end{procedurebox}

\noindent As every step of the presented procedure runs in exponential time in $|\varphi|$, we can conclude:
\begin{theorem}\label{thm:main-result}
  The finite satisfiability problem for $\ourLoopCPDL$ is $\ExpTime$-complete.
\end{theorem}

\section{Traces: A Folding-Friendly Invariant}\label{sec:traces}

By Definition~\ref{def:foldability}, a prospective fold-operation turning a $\varphi$-cactus
into a finite model is safe exactly when it creates no bad $\omega$-walk. Being infinitary,
this property is hard to maintain fold by fold, so we replace it by a finitary strengthening.
The key idea is to label every world with a summary of the walks reaching it, recording how
far each of them has already progressed towards a bad $\omega$-walk. We call such labels
\emph{traces}, and it is they, rather than foldability itself, that we maintain throughout
the finitisation process, fully defined in the follow-up section. 

\begin{definition}[Trace]\label{def:trace}
  Consider a bannable and concatenation-closed NFA $\autA_{q, q'}$.
  The~\emph{trace} $\tr_{q,q'}(\xi)$ of a walk~$\xi$ in a structure $\strA$ is 
  defined as the
  set of all pairs $(m, p) \in \N \times Q$ such that $\xi$ factorises as
  $\xi_0 \xi_1 \cdots \xi_m \xi_{m+1}$, where $\xi_1, \ldots, \xi_m$ are bad
  $(q, q')$-walks, the last world of $\xi_0$ satisfies $\neg\ldem_{q,q'}$
  and $\xi_{m+1}$ is a witnessing walk for $\hat{\autA}_{q, p}$.
  The \emph{vertex trace} $\vtr^{\strA}_{q,q'}(\elw)$ of a world~$\elw$ of a
  structure~$\strA$ is the union of the traces $\tr_{q, q'}(\xi)$ of all walks~$\xi$
  in~$\strA$ ending~in~$\elw$.
\end{definition}

\vspace{-1.25em}
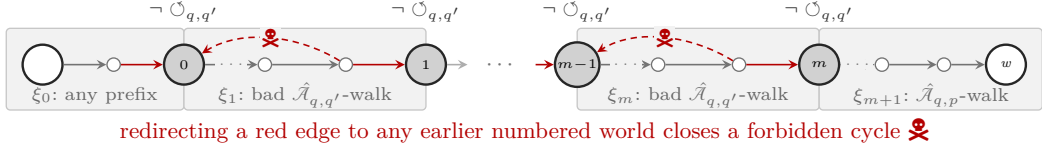
\begin{figure}[h]
  \centering
  \begin{tikzpicture}[
      x=1cm, y=1cm, every node/.style={transform shape},
      seg/.style     = {rounded corners=1.5pt, draw=black!22, fill=black!5,
                        line width=0.5pt},
      wld/.style     = {circle, draw=black!55, fill=white, line width=0.5pt,
                        inner sep=0pt, minimum size=5pt},
      brd/.style     = {circle, draw=black!85, fill=black!18, line width=0.9pt,
                        inner sep=0.5pt, minimum size=15pt, font=\tiny},
      endp/.style    = {circle, draw=black!85, fill=white, line width=0.9pt,
                        inner sep=0.5pt, minimum size=15pt, font=\tiny},
      mid/.style     = {inner sep=0pt, outer sep=0pt, minimum size=0pt},
      tedge/.style   = {->, >=stealth, draw=black!55, semithick},
      dedge/.style   = {->, >=stealth, draw=black!32, semithick},
      cedge/.style   = {->, >=stealth, draw=red!70!black, semithick},
      bedge/.style   = {->, >=stealth, draw=red!70!black, semithick,
                        dash pattern=on 2.4pt off 1.8pt},
      skull/.style   = {font=\scriptsize, text=red!70!black, inner sep=0.8pt,
                        fill=black!5},
      dem/.style     = {font=\scriptsize, text=black!80, inner sep=1.2pt},
      tag/.style     = {font=\scriptsize, text=black!60, inner sep=1.2pt},
      bad/.style     = {font=\footnotesize, text=red!70!black, inner sep=1.2pt},
      dots/.style    = {font=\footnotesize, text=black!55, inner sep=0pt},
      elab/.style    = {font=\tiny, text=black!55, inner sep=1.3pt,
                        fill=black!5, midway},
    ]
    \def\segbox#1#2{\draw[seg] (#1,-0.60) rectangle (#2,0.48);}
    \def\skullat#1{\node[skull] at #1 {\faSkullCrossbones};}

    \segbox{0.00}{2.35}   \segbox{2.35}{5.55}
    \segbox{7.55}{10.75}  \segbox{10.75}{13.70}

    \node[endp] (s)  at ( 0.48,0) {};
    \node[wld]  (a1) at ( 1.42,0) {};
    \node[brd]  (b0) at ( 2.35,0) {$0$};
    \node[wld]  (c1) at ( 3.42,0) {};
    \node[wld]  (c2) at ( 4.49,0) {};
    \node[brd]  (b1) at ( 5.55,0) {$1$};
    \node[brd]  (b2) at ( 7.55,0) {$m{-}1$};
    \node[wld]  (e1) at ( 8.62,0) {};
    \node[wld]  (e2) at ( 9.69,0) {};
    \node[brd]  (b3) at (10.75,0) {$m$};
    \node[wld]  (f1) at (11.58,0) {};
    \node[wld]  (f2) at (12.40,0) {};
    \node[endp] (x)  at (13.22,0) {$\elw$};

    \draw[tedge] (s)  -- (a1);   \draw[cedge] (a1) -- (b0);
    \draw[tedge] (b0) -- node[elab] {$\ldots$} (c1);
    \draw[tedge] (c1) -- (c2);
    \draw[cedge] (c2) -- (b1);
    \draw[dedge] (b1) -- (6.10,0);
    \node[dots]  at (6.55,0) {$\cdots$};
    \draw[cedge] (7.00,0) -- (b2);
    \draw[tedge] (b2) -- node[elab] {$\ldots$} (e1);
    \draw[tedge] (e1) -- (e2);
    \draw[cedge] (e2) -- (b3);
    \draw[tedge] (b3) -- node[elab] {$\ldots$} (f1);
    \draw[tedge] (f1) -- (f2);   \draw[tedge] (f2) -- (x);

    \draw[bedge] (c2) to[out=148,in=32,looseness=0.88]
                 node[mid, pos=0.5] (mA) {} (b0);
    \skullat{(mA)}
    \draw[bedge] (e2) to[out=148,in=32,looseness=0.88]
                 node[mid, pos=0.5] (mB) {} (b2);
    \skullat{(mB)}

    \node[tag] at ( 1.17,-0.40) {$\xi_0$: any prefix};
    \node[tag] at ( 3.95,-0.40) {$\xi_1$: bad $\hat{\autA}_{q,q'}$-walk};
    \node[tag] at ( 9.15,-0.40) {$\xi_m$: bad $\hat{\autA}_{q,q'}$-walk};
    \node[tag] at (12.23,-0.40) {$\xi_{m+1}$: $\hat{\autA}_{q,p}$-walk};

    \node[dem, anchor=south] at ( 2.35,0.52) {$\neg \ldem_{q,q'}$};
    \node[dem, anchor=south] at ( 5.55,0.52) {$\neg \ldem_{q,q'}$};
    \node[dem, anchor=south] at ( 7.55,0.52) {$\neg \ldem_{q,q'}$};
    \node[dem, anchor=south] at (10.75,0.52) {$\neg \ldem_{q,q'}$};

    \node[bad, anchor=north] at (6.85,-0.66)
      {redirecting a red edge to any earlier numbered world closes a forbidden cycle~\faSkullCrossbones};
  \end{tikzpicture}
  \vspace{-0.6em}
  \caption{A visualisation of a walk~$\xi$ of~$\strA$ ending in a world~$\elw$, witnessing $(m,p) \in \tr_{q,q'}(\xi)$.}
  \label{fig:trace}
\end{figure}

\vspace{-0.5em}

\noindent Intuitively, $(m, p) \in \vtr^{\strA}_{q,q'}(\elw)$ indicates that 
a world~$\elw$ is reached by a walk that has already accumulated $m$ obstacles 
to folding, leaving $\hat{\autA}$ in state~$p$; the larger the~$m$, the closer 
that walk is to being a prefix of a bad $\omega$-walk. Foldability bounds 
these numbers: a foldable regular $\varphi$-cactus has finitely many vertex 
traces, which is what makes them usable in our approach.

\begin{lemma}\label{lemma:trace-finiteness}
  Let $\strC$ be a foldable regular $\varphi$-cactus.
  Then there is a $\rmB \in \N$ such that
  for every concatenation-closed bannable $\autA_{q,q'}$
  the values of $\vtr^\strC_{q,q'}$ only contain pairs $(m,p)$~with~${m \leqslant \rmB}$.
\end{lemma}
\begin{proof}[Proof idea]
  Suppose no such bound exists. Only finitely many automata are present in $\varphi$, so we
  may fix a single concatenation-closed bannable $\autA_{q,q'}$ for which, for
  arbitrarily large $m$, there are vertices $\elv$ and states $p$ with
  $(m,p) \in \vtr^\strC_{q,q'}(\elv)$. For this automaton we collect the walks
  witnessing such trace values into a language $\mathcal{L}$. A~word of
  $\mathcal{L}$ is a sequence over $\Sigma_\varphi$ interspersed with a fresh
  separator $\sharp$, which cuts it into a prefix walk followed by the
  successive bad $(q,q')$-walks it certifies; the number of separators
  therefore is one more than the trace value $m$. 
  In the appendix we give a grammar for $\mathcal{L}$, witnessing that 
  $\mathcal{L}$ is context-free; the construction is technical 
  and, we admit, rather tedious (\eg one needs to care about detours and traversals of anchored cycles).
  It is concluded that $\mathcal{L}$ contains words with arbitrarily many $\sharp$, so
  Ogden's lemma applies with the separators marked: any word carrying enough of
  them decomposes as $sztyx$, where $z$ or $y$ is pumpable and contains at
  least one $\sharp$. Pumping $z$ yields a bad $\omega$-walk with all $\xi_i$
  bad, and pumping $y$ yields one with~all~$\inv{\xi_i}$~bad.
\end{proof}

\section{Soundness: Folding Infinite Cacti into Finite Models}\label{sec:soundness}

Equipped with traces, we are ready to prove the soundness of our approach, 
\ie the ``only if'' direction of Theorem~\ref{thm:main-characterisation}. 
We first introduce \emph{$\varphi$-quasicacti}: intermediate structures 
resulting from the application of fold-operations to $\varphi$-cacti. 
They are no longer cacti, but they are still models of $\varphi$. 
After that we introduce the similarity relation~$\approx$, 
and the fold-operation itself.

\begin{definition}\label{def:quasicactus}
  A structure is \emph{quasicactus-like} if it is obtained from 
  a cactus-like structure by adding to its transition relations 
  any number of edges of two kinds, each for some $\elv, \elw, i$ 
  with $\elv$ an ancestor of~$\elw$ and $i$ non-zero: 
  a \emph{quasi-child edge} $(\elw, \elv i)$, and a \emph{quasi-back edge} 
  joining $\elw$ with an anchored descendant $\elv i 0^{k}$ of~$\elv$.
  A witnessing walk or cycle is \emph{quasi-anchored} if it is anchored, 
  or becomes anchored once a quasi-child edge is read as the step from~$\elv$ 
  to its child~$\elv i$, and a quasi-back edge as a back edge.
  Finally, $\strD$ is a \emph{$\varphi$-quasicactus} if it is quasicactus-like 
  and satisfies \ref{phicactus:bool}--\ref{phicactus:cycle} from 
  Definition~\ref{def:phi-cactus} with ``quasi-anchored''~in~place~of~``anchored''.
  As for $\varphi$-cacti, for all $\elv$ and its realised demands $\delta$, we fix witnessing stationary~or~quasi-anchored~$\xi_{\delta}^{\elv}$.
\end{definition}

\vspace{-1em}
\begin{figure}[h]
  \centering
  \begin{tikzpicture}[
      x=1cm, y=0.87cm, scale=0.95, every node/.style={transform shape},
      world/.style  = {circle, draw=black!65, fill=white, line width=0.5pt,
                       inner sep=0pt, minimum size=4.6pt},
      root/.style   = {world, minimum size=5.8pt, line width=0.7pt, draw=black!80},
      gworld/.style = {world, draw=black!30},
      tedge/.style  = {->, >=stealth, draw=black!55, semithick},
      gedge/.style  = {->, >=stealth, draw=black!25, semithick},
      zedge/.style  = {->, >=stealth, draw=blue!55!black, thick},
      bedge/.style  = {->, >=stealth, draw=red!70!black, thick,
                       dash pattern=on 3pt off 2pt},
      qcedge/.style = {->, >=stealth, draw=violet!80!black, thick},
      qbedge/.style = {->, >=stealth, draw=violet!80!black, thick,
                       dash pattern=on 3pt off 2pt},
      wlab/.style   = {font=\scriptsize, inner sep=1.5pt, text=black!70},
      glab/.style   = {wlab, text=black!40},
      elab/.style   = {font=\tiny, inner sep=1.3pt, fill=white, midway},
      gelab/.style  = {font=\tiny, inner sep=1.3pt, midway, text=black!40},
      qlab/.style   = {elab, fill=violet!20, inner sep=0.8pt},
      slab/.style   = {elab, fill=orange!35, inner sep=0.8pt},
      leg/.style    = {font=\scriptsize, inner sep=0pt, anchor=west},
      ttl/.style    = {font=\small, inner sep=0pt},
      qpath/.style  = {line width=4.2pt, draw=violet!20, line cap=round,
                       line join=round},
      simpath/.style= {line width=4.2pt, draw=orange!35, line cap=round},
      cutm/.style   = {draw=red!70!black, line width=1.0pt, line cap=round},
      cuthalo/.style= {draw=white, line width=2.6pt, line cap=round},
      bush/.style   = {draw=black!30, fill=black!6, rounded corners=1pt,
                       line width=0.4pt,
                       drop shadow={shadow xshift=1.5pt, shadow yshift=-1.5pt,
                                    opacity=0.28, fill=black!60}},
      gbush/.style  = {draw=black!25, fill=black!4, rounded corners=1pt,
                       line width=0.4pt},
      rule/.style   = {draw=black!20, line width=0.6pt},
      slink/.style  = {draw=black!45, semithick},
      arcL/.style   = {out=205, in=155, looseness=0.78},
      arcR/.style   = {out=-26, in=28,  looseness=1.0},
      qc/.style     = {out=115, in=198, looseness=1.0},
      qb/.style     = {out=198, in=-45, looseness=0.9},
    ]
    \def\bushat#1{\draw[bush] (#1.south) -- ++(-0.40,-0.62)
                              -- ++(0.80,0) -- cycle;}
    \def\bushsmall#1{\draw[bush] (#1.south) -- ++(-0.32,-0.50)
                                 -- ++(0.64,0) -- cycle;}
    \def\bushtiny#1{\draw[bush] (#1.south) -- ++(-0.30,-0.44)
                                -- ++(0.60,0) -- cycle;}
    \def\minisub#1#2#3{\coordinate (ms) at ($(#1)+(#2,#3)$);
      \draw[slink] (#1) -- (ms);
      \draw[bush] (ms) -- ++(-0.20,-0.32) -- ++(0.40,0) -- cycle;}
    \def\cutmark#1#2#3{\begin{scope}[shift={#1}, rotate=#2]
        \draw[cuthalo] (-#3,-#3) -- (#3,#3);   \draw[cuthalo] (-#3,#3) -- (#3,-#3);
        \draw[cutm]    (-#3,-#3) -- (#3,#3);   \draw[cutm]    (-#3,#3) -- (#3,-#3);
      \end{scope}}

    \begin{scope}[shift={(0,0.75)}]
      \draw[tedge]  (0,-0.05) -- (0.55,-0.05);
      \node[leg] at (0.68,-0.05) {tree edge};
      \draw[zedge]  (0,-0.53) -- (0.55,-0.53);
      \node[leg] at (0.68,-0.53) {$0$-edge};
      \draw[bedge]  (0,-1.01) -- (0.55,-1.01);
      \node[leg] at (0.68,-1.01) {back edge};
      \draw[bush] (0.28,-1.42) -- (0.02,-1.82) -- (0.54,-1.82) -- cycle;
      \node[leg] at (0.68,-1.62) {subtree};
      \draw[qcedge] (0,-2.10) -- (0.55,-2.10);
      \node[leg] at (0.68,-2.10) {quasi-child edge};
      \draw[qbedge] (0,-2.58) -- (0.55,-2.58);
      \node[leg] at (0.68,-2.58) {quasi-back edge};
      \draw[qpath]  (0.03,-3.06) -- (0.52,-3.06);
      \node[leg] at (0.68,-3.06) {quasi-anch.\ cycle};
      \draw[simpath] (0.03,-3.54) -- (0.52,-3.54);
      \node[leg] at (0.68,-3.54) {similar edges};
      \draw[gbush] (0.28,-3.90) -- (0.02,-4.30) -- (0.54,-4.30) -- cycle;
      \cutmark{(0.28,-4.11)}{0}{0.17}
      \node[leg] at (0.68,-4.11) {cut part};
    \end{scope}

    \begin{scope}[shift={(5.80,0)}]
      \node[ttl] at (0.10,0.70) {\textbf{(a)} a $\varphi$-quasicactus};

      \node[root]  (r)  at ( 0.55, 0)    {};
      \node[world] (a)  at (-1.75,-0.80) {};
      \node[world] (a0) at (-1.75,-1.60) {};
      \node[world] (t)  at (-0.70,-0.90) {};
      \node[world] (s)  at ( 1.95,-0.75) {};
      \draw[tedge] (r)  -- node[elab] {$\pe$} (a);
      \draw[zedge] (a)  -- node[elab] {$\pr$} (a0);
      \draw[bedge] (a0) to[arcL] node[elab] {$\pe$} (a);
      \bushtiny{a0}
      \draw[tedge] (r)  -- node[elab, pos=0.68] {$\pe$} (t);
      \bushsmall{t}
      \draw[tedge] (r)  -- node[elab, pos=0.62] {$\pr$} (s);
      \bushat{s}

      \node[world] (v)   at (0.55,-0.80) {};
      \node[world] (vi)  at (0.55,-1.60) {};
      \node[world] (vi0) at (0.55,-2.40) {};
      \node[world] (vik) at (0.55,-3.20) {};
      \draw[tedge] (r)   -- node[elab] {$\ps$} (v);
      \draw[tedge] (v)   -- node[elab] {$\pr$} (vi);
      \draw[zedge] (vi)  -- (vi0);
      \draw[zedge] (vi0) -- node[qlab] {$\ldots$} (vik);
      \draw[bedge] (vik) to[arcR] node[elab] {$\ps$} (v);

      \node[world] (u)  at (-0.30,-2.02) {};
      \node[world] (w)  at (-1.10,-2.42) {};
      \node[world] (wj) at (-2.05,-2.96) {};
      \draw[tedge] (vi) -- node[elab, pos=0.6] {$\ldots$} (u);
      \draw[zedge] (u)  -- node[elab, pos=0.55] {$\pe$} (w);
      \draw[tedge] (w)  -- node[elab, pos=0.55] {$\pr$} (wj);
      \bushat{wj}

      \minisub{vi}{0.40}{-0.30} \minisub{u}{0.42}{-0.32} \minisub{w}{0.05}{-0.52}

      \draw[qcedge] (w)   to[qc] node[qlab, pos=0.42] {$\pr$}  (vi);
      \draw[qbedge] (vik) to[qb] node[qlab, pos=0.55] {$\ps$} (w);
      \begin{scope}[on background layer]
        \draw[qpath] (w) to[qc] (vi) -- (vi0) -- (vik) to[qb] (w);
      \end{scope}

      \node[wlab, anchor=south]      at ([yshift=0.5pt]r.north)  {$\varepsilon$};
      \node[wlab, anchor=east]       at ([xshift=-1.5pt]v.west)  {$\elv$};
      \node[wlab, anchor=west]       at ([xshift=1.5pt]vi.east)  {$\elv i$};
      \node[wlab, anchor=north]      at ([yshift=-2pt]vik.south) {$\elv i 0^{k}$};
      \node[wlab, anchor=east]       at ([xshift=-1.5pt]w.west)  {$\elw$};
    \end{scope}

    \draw[rule] (8.75,0.90) -- (8.75,-3.72);

    \begin{scope}[shift={(11.60,0)}]
      \node[ttl] at (0.10,0.70) {\textbf{(b)} the fold-operation};

      \node[root]  (br)  at ( 0.55, 0)    {};
      \node[world] (ba)  at (-1.75,-0.80) {};
      \node[world] (ba0) at (-1.75,-1.60) {};
      \node[world] (bt)  at (-0.70,-0.90) {};
      \node[world] (bs)  at ( 1.95,-0.75) {};
      \draw[tedge] (br)  -- node[elab] {$\pe$} (ba);
      \draw[zedge] (ba)  -- node[elab] {$\pr$} (ba0);
      \draw[bedge] (ba0) to[arcL] node[elab] {$\pe$} (ba);
      \bushtiny{ba0}
      \draw[tedge] (br)  -- node[elab, pos=0.68] {$\pe$} (bt);
      \bushsmall{bt}
      \draw[tedge] (br)  -- node[elab, pos=0.62] {$\pr$} (bs);
      \bushat{bs}

      \node[world] (bv)   at (0.55,-0.80) {};
      \node[world] (bvi)  at (0.55,-1.60) {};
      \node[world] (bvi0) at (0.55,-2.40) {};
      \node[world] (bvik) at (0.55,-3.20) {};
      \draw[tedge] (br)   -- node[elab] {$\ps$} (bv);
      \draw[tedge] (bv)   -- node[slab] {$\pr$} (bvi);
      \draw[zedge] (bvi)  -- (bvi0);
      \draw[zedge] (bvi0) -- node[elab] {$\ldots$} (bvik);
      \draw[bedge] (bvik) to[arcR] node[elab] {$\ps$} (bv);

      \node[world]  (bu)  at (-0.30,-2.02) {};
      \node[world]  (bw)  at (-1.10,-2.42) {};
      \node[gworld] (bwi) at (-2.05,-2.96) {};
      \draw[tedge] (bvi) -- node[elab, pos=0.6] {$\ldots$} (bu);
      \draw[zedge] (bu)  -- node[elab, pos=0.55] {$\pe$} (bw);
      \draw[gedge] (bw)  -- node[gelab, fill=orange!35, inner sep=0.8pt, pos=0.34] {$\pr$} (bwi);
      \draw[gbush] (bwi.south) -- ++(-0.40,-0.62) -- ++(0.80,0) -- cycle;
      \cutmark{(-1.84,-2.84)}{28}{0.13}
      \cutmark{(-2.05,-3.42)}{0}{0.22}

      \minisub{bvi}{0.40}{-0.30} \minisub{bu}{0.42}{-0.32} \minisub{bw}{0.05}{-0.52}

      \begin{scope}[on background layer]
        \draw[simpath] (bv) -- (bvi);
        \draw[simpath] (bw) -- (bwi);
      \end{scope}
      \draw[qcedge] (bw)   to[qc] node[elab, pos=0.42] {$\pr$}  (bvi);
      \draw[qbedge] (bvik) to[qb] node[elab, pos=0.55] {$\ps$} (bw);

      \node[wlab, anchor=south]      at ([yshift=0.5pt]br.north)  {$\varepsilon$};
      \node[wlab, anchor=east]       at ([xshift=-1.5pt]bv.west)  {$\elv$};
      \node[wlab, anchor=west]       at ([xshift=1.5pt]bvi.east)  {$\elv i$};
      \node[wlab, anchor=north]      at ([yshift=-2pt]bvik.south) {$\elv i 0^{k}$};
      \node[wlab, anchor=east]       at ([xshift=-1.5pt]bw.west)  {$\elw$};
      \node[glab, anchor=south east] at ([shift={(-0.5pt,0.5pt)}]bwi.north west) {$\elw i$};
    \end{scope}
  \end{tikzpicture}
  \vspace{-1em}
  \caption{\textbf{(a)} An example quasicactus. \textbf{(b)} A visualisation of the fold operation for~$(\elv, \elv i) \approx (\elw, \elw i)$.}
  \label{fig:quasicactus}
\end{figure}
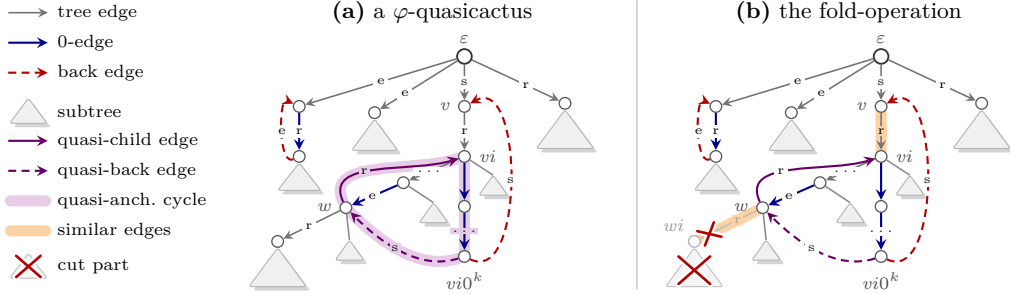

\noindent The extra edges added to quasicacti during the folding process serve two purposes:
quasi-child edges mimic the removed child edge, whereas quasi-back edges
mimic the removed back edge and allow us to recover the witnessing cycle.
Every $\varphi$-cactus is trivially a $\varphi$-quasicactus and, by design,
$(\strD, \varepsilon) \models \varphi$~for~every~$\varphi$-quasicactus~$\strD$.
\noindent We next define the similarity relation~$\approx$.
By the \emph{type} $\lambda^{\strD}(\elw)$ of a world~$\elw$ in~$\strD$ we mean the 
set of propositions $\pP$ and self-loops $\lop(\pe)$ over $\Sigma_\varphi$ satisfied by $\elw$ in~$\strD$.
Similarity is what will let us select a vertex for removal: once a vertex is gone, every world
that used it as a witness must be able to receive a replacement. This is why we compare not only
vertices, but also~the~edges~leading~to~them~from~their~parents.

\begin{definition}[Similarity]\label{def:similar}
  In a $\varphi$-quasicactus~$\strD$, vertices $\elv, \elw$ are \emph{similar}, written
  $\elv \approx \elw$, if $\lambda^{\strD}(\elv) = \lambda^{\strD}(\elw)$ and
  $\vtr^{\strD}_{q,q'}(\elv) = \vtr^{\strD}_{q,q'}(\elw)$ for every 
  concatenation-closed bannable~$\autA_{q,q'}$.
  Fix a demand~$\delta$ satisfied by both~$\elv$ and~$\elw$, put $i \deff \kappa(\delta)$,
  and let the \emph{initial transition}~$\sigma$ be an action or its converse.
  The edges $(\elv, \elv i)$ and $(\elw, \elw i)$ are $(\delta, \sigma)$-\emph{similar}, written
  $(\elv, \elv i) \approx_{\delta, \sigma} (\elw, \elw i)$, if $\elv \approx \elw$,
  $\elv i \approx \elw i$, and $\elv, \sigma, \elv i$ is an infix of $\xi_{\delta}^{\elv}$,
  and likewise $\elw, \sigma, \elw i$ of~$\xi_{\delta}^{\elw}$.
\end{definition}

\noindent We now define fold-operations; see Figure~\ref{fig:quasicactus}(b), where~the~freshly~added~edges~are~purple.

\begin{definition}[Folding]\label{def:fold-op}
  Let $\strD$ be a $\varphi$-quasicactus containing a pair of $(\delta, \sigma)$-similar edges
  $(\elv, \elv i) \approx_{\delta, \sigma} (\elw, \elw i)$, where $i \deff \kappa(\delta)$ and
  $\elw \not\in \elv i 0^*$ is a descendant of~$\elv i$.
  The \emph{fold-operation} $((\elv, \elv i), (\elw, \elw i))$
  turns~$\strD$ into the structure~$\strD'$ with domain
  $\domD' \deff \domD \setminus \domD_{\elw i}$, where $\domD_{\elw i}$ is the subtree rooted
  at~$\elw i$: all interpretations are restricted to~$\domD'$ and then extended with
  the quasi-child edge $(\elw, \elv i)$, put into $\sem{\sigma}_{\strD'}$, and,
  whenever $\delta$ is a cycle demand, 
  also with the quasi-back edge $(\elv i 0^{k}, \elw)$,
  put into $\sem{\sigma'}_{\strD'}$ for the closing step
  $\elv i 0^{k}, \sigma', \elv$ of the~witnessing~cycle~$\xi_{\delta}^{\elv}$.
\end{definition}

\noindent Folding preserves types and traces in the sense of the next lemma, making them nice~invariants.

\begin{lemma}\label{lemma:invariant}
  Let $\strD'$ be obtained from $\varphi$-quasicactus $\strD$ by a fold-operation.
  Then $\vtr_{q, q'}^{\strD'}(\elu) \subseteq \vtr_{q, q'}^{\strD}(\elu)$
  and $\lambda^{\strD'}(\elu) = \lambda^\strD(\elu)$ for 
  all concatenation-closed bannable~$\autA_{q, q'}$~and~all~$\elu \in \domD'$.
\end{lemma}
\begin{proof}[Proof idea]
  Types are untouched, as folding never alters atomic propositions. For traces we
  do induction on the number of new edges of a witnessing walk $\xi$, inspecting the last
  one traversed by $\xi$: the prefix of $\xi$ before it uses fewer new edges, hence
  already exists in $\strD$ by the induction hyp.  
  Every new edge duplicates one incident to an $\approx$-equivalent vertex,
  so the run reroutes through that vertex and continues, yielding the same trace
  in~$\strD$;~see~Figure~\ref{fig:soundness}(a).
\end{proof}

\vspace{-1.25em}
\begin{figure}[h]
  \centering
  \begin{tikzpicture}[
      x=1cm, y=1cm, scale=0.95, every node/.style={transform shape},
      world/.style  = {circle, draw=black!60, fill=white, line width=0.5pt,
                       inner sep=0pt, minimum size=4.6pt},
      bworld/.style = {circle, draw=black!80, fill=white, line width=0.8pt,
                       inner sep=0.5pt, minimum size=13.5pt, font=\tiny},
      tedge/.style  = {->, >=stealth, draw=black!55, semithick},
      qcedge/.style = {->, >=stealth, draw=violet!80!black, thick},
      cedge/.style  = {->, >=stealth, draw=red!70!black, semithick},
      cring/.style  = {draw=red!70!black, semithick},
      fade/.style   = {->, >=stealth, draw=red!70!black, semithick,
                       dash pattern=on 2.6pt off 2pt},
      sim/.style    = {draw=black!45, semithick,
                       dash pattern=on 1.8pt off 1.8pt},
      halo/.style   = {line width=4.2pt, draw=teal!22, line cap=round,
                       line join=round},
      elab/.style   = {font=\tiny, inner sep=1.3pt, fill=white, midway},
      qlab/.style   = {elab, fill=violet!20, inner sep=0.8pt},
      wlab/.style   = {font=\tiny, text=black!70, inner sep=1.2pt},
      tlab/.style   = {font=\tiny, text=black!70, inner sep=1.2pt},
      rlab/.style   = {font=\tiny, text=red!70!black, inner sep=1.2pt},
      dem/.style    = {font=\scriptsize, text=black!80, inner sep=1.2pt},
      tag/.style    = {font=\tiny, text=black!60, inner sep=1.2pt,
                       align=center},
      dots/.style   = {font=\footnotesize, text=black!55, inner sep=1pt},
      ttl/.style    = {font=\small, inner sep=0pt},
      rule/.style   = {draw=black!20, line width=0.6pt},
    ]

    \begin{scope}[shift={(0,0)}]
      \node[ttl] at (2.56,0.75)
        {\textbf{(a)} rerouting a new edge (Lemma~\ref{lemma:invariant})};

      \node[dots]   (p0) at (0.18, 0.00) {$\cdots$};
      \node[world]  (p1) at (0.74, 0.00) {};
      \node[world]  (p2) at (1.32, 0.00) {};
      \node[bworld] (w)  at (2.05, 0.00) {$\elw$};
      \draw[tedge] (p0) -- (p1);  \draw[tedge] (p1) -- (p2);
      \draw[tedge] (p2) -- (w);

      \node[dots]   (r0) at (0.18,-1.75) {$\cdots$};
      \node[world]  (r1) at (0.74,-1.75) {};
      \node[world]  (r2) at (1.32,-1.75) {};
      \node[bworld] (v)  at (2.05,-1.75) {$\elv$};
      \draw[tedge] (r0) -- (r1);  \draw[tedge] (r1) -- (r2);
      \draw[tedge] (r2) -- (v);

      \node[bworld] (vi) at (3.15,-0.88) {$\elv i$};
      \node[world]  (s1) at (3.85,-0.88) {};
      \node[world]  (s2) at (4.45,-0.88) {};
      \node[bworld] (u)  at (5.15,-0.88) {$\elu$};
      \draw[qcedge] (w) -- node[qlab, pos=0.46] {$\sigma$} (vi);
      \draw[tedge]  (v) -- node[elab, pos=0.46] {$\sigma$} (vi);
      \draw[tedge] (vi) -- (s1);  \draw[tedge] (s1) -- (s2);
      \draw[tedge] (s2) -- (u);

      \draw[sim] (w) -- node[elab] {$\approx$} (v);

      \begin{scope}[on background layer]
        \draw[halo] (r0) -- (v) -- (vi) -- (u);
      \end{scope}

      \node[tag, anchor=north] at (1.03,-0.18) {$\xi$ in $\strD'$};
      \node[tag, anchor=south] at (1.03,-1.57) {$\zeta$ in $\strD$};
      \node[tag, anchor=south] at (4.15,-0.70) {suffix of $\xi$};
      \node[tag, text=violet!65!black, anchor=south west] at (2.50,-0.22)
        {new edge};
      \node[tag, anchor=north west] at (2.62,-1.38) {twin edge};

      \node[tlab, anchor=south] at (w.north) {$(i{-}1,q_k)$};
      \node[tlab, anchor=north] at (v.south) {$(i{-}1,q_k)$};
      \node[tlab, anchor=north] at (u.south) {$(m,q'')$};
    \end{scope}

    \draw[rule] (5.75,0.95) -- (5.75,-2.15);

    \begin{scope}[shift={(6.45,0)}]
      \node[ttl] at (3.75,0.75)
        {\textbf{(b)} negative cycle demands (Lemma~\ref{lemma:quasicactus})};

      \draw[cring, decoration={markings,
              mark=at position 0.09 with {\arrow{stealth}},
              mark=at position 0.26 with {\arrow{stealth}},
              mark=at position 0.42 with {\arrow{stealth}},
              mark=at position 0.59 with {\arrow{stealth}},
              mark=at position 0.76 with {\arrow{stealth}},
              mark=at position 0.92 with {\arrow{stealth}}},
            postaction={decorate}] (1.60,-0.95) circle (0.78);
      \node[world]  at (1.21,-0.27) {};   \node[world] at (1.99,-0.27) {};
      \node[world]  at (2.38,-0.95) {};   \node[world] at (1.99,-1.63) {};
      \node[world]  at (1.21,-1.63) {};
      \node[bworld] (x) at (0.82,-0.95) {$\elx$};
      \node[rlab, align=center, fill=white, inner sep=1.8pt] at (1.60,-0.95)
        {$\xi$\\$\autA_{q,q'}$};
      \node[dem, anchor=east] at ([xshift=-3pt]x.west) {$\neg \ldem_{q,q'}$};

      \draw[->, >=stealth, draw=black!40] (2.50,-0.95) -- (3.08,-0.95);
      \node[tag, anchor=south] at (2.79,-0.84) {unroll};

      \node[bworld] (x1) at (3.46,-0.95) {$\elx$};
      \node[world]  (m1) at (3.99,-0.95) {};
      \node[world]  (m2) at (4.51,-0.95) {};
      \node[bworld] (x2) at (5.04,-0.95) {$\elx$};
      \node[world]  (m3) at (5.57,-0.95) {};
      \node[world]  (m4) at (6.09,-0.95) {};
      \node[bworld] (x3) at (6.62,-0.95) {$\elx$};
      \draw[cedge] (x1) -- (m1);  \draw[cedge] (m1) -- (m2);
      \draw[cedge] (m2) -- (x2);  \draw[cedge] (x2) -- (m3);
      \draw[cedge] (m3) -- (m4);  \draw[cedge] (m4) -- (x3);
      \draw[fade]  (x3) -- (7.28,-0.95);
      \node[dots]  at (7.52,-0.95) {$\cdots$};

      \node[rlab, anchor=south] at (4.25,-0.80) {$\xi$};
      \node[rlab, anchor=south] at (5.83,-0.80) {$\xi$};
      \node[rlab, anchor=south] at (7.08,-0.80) {$\xi$};
      \node[tlab, anchor=north] at (x2.south) {$(1,q)$};
      \node[tlab, anchor=north] at (x3.south) {$(2,q)$};
      \node[tlab, anchor=north] at (7.52,-1.19) {$(m,q)$};

      \node[tag, anchor=north] at (5.48,-1.55)
        {$\xi^{\omega}$: a bad $\omega$-walk};
    \end{scope}
  \end{tikzpicture}
  \vspace{-1em}
  \caption{\textbf{(a)} The last fresh edge of a walk~$\xi$ in $\strD'$ is replaced
  by its twin at the $\approx$-equivalent~$\elv$, giving the highlighted walk
  of~$\strD$. \textbf{(b)} A forbidden cycle unrolls into a bad
  $\omega$-walk with~infinite~trace.}
  \label{fig:soundness}
\end{figure}

\noindent The following lemma proves that fold-operations preserve foldability and ``$\varphi$-quasicactushood''.
\begin{lemma}\label{lemma:quasicactus}
  Let $\strD_0, \strD_1, \ldots, \strD_{n}$ be structures such that $\strD_0$ is a regular
  foldable $\varphi$-cactus~$\strC$ and each~$\strD_{i+1}$ is obtained from~$\strD_i$ by a
  fold-operation. Then all $\strD_i$ are foldable~$\varphi$-quasicacti.
\end{lemma}
\begin{proof}[Proof sketch]
Induction on~$i$; the base case is immediate, as every $\varphi$-cactus is a
$\varphi$-quasicactus. Assume the claim for~$\strD_i$ and abbreviate
$\strD \deff \strD_i$ and $\strD' \deff \strD_{i+1}$. We only treat the 
``only if'' part of Condition~\ref{phicactus:cycle}; the remaining ones are 
immediate (see appendix). We name the vertices as in Definition~\ref{def:fold-op}.
Fix~$\autA_{q,q'}$. If it is not bannable, a new cycle is
harmless: an implication cannot be falsified by affirming its consequent. Otherwise it is
stationary, and no new cycle arises, as the fold adds no self-loop ($\elw \not\in \elv i 0^*$)
and alters no type; or concatenation-closed, where we use traces.
Iterating Lemma~\ref{lemma:invariant} gives
$\vtr^{\strD'}_{q,q'} \subseteq \vtr^{\strC}_{q,q'}$~on~$\domD'$, where
Lemma~\ref{lemma:trace-finiteness} ensures their finiteness.
Foldability of~$\strD'$ follows.
No vertex~$\elv$ failing~$\ldem_{q,q'}$ carries a witnessing cycle: 
being a bad $(q,q')$-walk, 
it~could~be~repeated,~implying~infiniteness~of~$\vtr^{\strD'}_{q,q'}(\elv)$.~See~Figure~\ref{fig:soundness}(b).\qedhere
\end{proof}

\noindent We are ready for the final result of the paper, establishing soundness of our algorithm.

\begin{lemma}[Soundness]\label{lemma:soundness}
  If a foldable regular $\varphi$-cactus $\strC$ exists, then $\varphi$~has~a~finite~model.
\end{lemma}
\begin{proof}[Proof idea]
  For a vertex $\elv$ let a map $C_\elv$ send each concatenation-closed bannable 
  $\autA_{q,q'}$ to $\vtr^\strC_{q,q'}(\elv)$; by Lemma~\ref{lemma:trace-finiteness}
  only finitely many arise, so an edge $(\elv, \elv i)$ has one of finitely many
  \emph{extended types} $(\lambda^\strC(\elv), \lambda^\strC(\elv i), C_\elv, C_{\elv i}, i, \sigma)$,
  with $i, \sigma$ read off the witnessing walk of~$\elv$, \ie exactly the data making edges similar.
  For a suitable $\rmM$, every path from~$\varepsilon$ to depth $\rmM$ in $\strC$ has two edges
  $(\elv, \elv i),\,(\elw, \elw i)$ of equal extended type with $\elw\notin\elv i0^*$ below $\elv i$, so
  $((\elv, \elv i), (\elw, \elw i))$ is a permissible fold. After this fold,
  the model remains a quasicactus, some depth-$\rmM$ vertex is no more,
  and $\rmM$ enables further such folds, soon yielding a finite model.
\end{proof}

\section{Conclusions}\label{sec:conclusions}

We pushed the decidability frontier for finite satisfiability 
of $\PDL$ with loops, giving a modern account of the general case and, via $\varphi$-cacti and traces, new decidability results in the finite.
Decidability of the finite satisfiability problem for the full logic remains
open. The obstacles we are aware of, and the places in the
proof where we exploit concatenation-closedness, are discussed 
closely in the appendix. We believe our techniques will eventually 
settle the problem, and that they will find applications, \eg in
guarded~fragments~with~transitive~closure~\cite{Michaliszyn09,FigueiraFN26}.
\clearpage
\bibliography{references}
\clearpage
\appendix
\makeatletter
\edef\thetheorem{\expandafter\noexpand\thesection\@thmcountersep\@thmcounter{theorem}}
\makeatother
\tableofcontents
\clearpage
\section{Additional Information to the Reviewers}\label{app:to-reviewers}

\subsection{Relation to an Earlier Informal Version of This Work}\label{app:workshop-version}

A preliminary account of this work, \emph{Baby Steps towards Finite Satisfiability for LoopPDL}, was presented at the 39th International Workshop on Description Logics (DL~2026)~\cite{BednarczykS26}. The proceedings of that workshop are informal: contributions there are explicitly not treated as formal publications, so the present paper is not a dual submission.
That workshop paper was a simplified proof of concept: loops were restricted to the shapes $\lop(\pe)$ and $\lop(\pe^{+})$ for atomic~$\pe$, and the converse operator was not treated. Passing to $\ourLoopCPDL$ changed the construction throughout: it rests on the more general framework of $\varphi$-cacti, on a more involved foldability condition, and on a more refined notion of trace. Beyond~the~guiding~idea, the~two~papers~share~little.

\subsection{Comparison with Bojańczyk's Work}\label{app:bojanczyk}

\noindent Our approach was inspired by Bojańczyk's solution to finite satisfiability for the $\mu$-calculus with backward modalities~\cite{Bojanczyk02,BojanczykThesis}, so we state precisely what we borrowed.

\vspace{0.25em}
\noindent \textbf{The logics are incomparable.}\,
His logic is invariant under two-way bisimulation and ours is not: a world with an $\pe$-self-loop and a world on a two-element $\pe$-cycle are bisimilar, yet only the former satisfies $\lop(\pe)$. Hence no formula of Bojańczyk's logic defines $\lop(\pe)$, already over finite structures. Conversely, $\LoopCPDL$ has no fixed points beyond the regular patterns of~$\pi^{*}$. Hence, these two logics are incomparable in terms of their expressive power, and to the best of our knowledge, the decidability of one does not imply that of the other.

\vspace{0.25em}
\noindent \textbf{Proof scheme.}\,
Bojańczyk represents a finite graph by its two-way tree unravelling, recognises such trees by an alternating two-way parity automaton, and singles out those that loop back into a finite graph by \emph{bounded signature}~\cite[Lem.~5.2.10]{BojanczykThesis}. We unravel into a $\varphi$-cactus~(Lemma~\ref{lemma:model-to-cactus}), encode cacti as automaton-recognisable trees~(Lemma~\ref{lemma:cactus-automaton}), and single out the finite ones by foldability~(Theorem~\ref{thm:main-characterisation}). Reasoning in the infinite while carrying a certificate that one may return to the finite is Bojańczyk's idea, and we gladly~acknowledge~the~debt.

\vspace{0.25em}
\noindent \textbf{The technique.}\,
Bojańczyk can afford tree models precisely because the two-way $\mu$-calculus cannot enforce cycles. Moreover, his folding operation may create arbitrary cycles, needing only to keep infinite plays within the parity condition. Ours must create none that $\varphi$ forbids while keeping every cycle it enforces---and creating cycles is the very mechanism of finitisation. Finally, Bojańczyk's traces are plays of the acceptance game, and the associated measure counts the occurrences of an odd rank not yet cancelled by a smaller one, that is, progress towards violating the parity condition. Ours are walks in the structure, and $\vtr^{\strA}_{q, q'}(\elx)$ counts the bad $(q,q')$-walks already stacked up, that is, progress towards a bad $\omega$-walk. The analogy ends there, but we keep the word \emph{trace}, and speak of folding where Bojańczyk speaks of looping, to honour that influence.

\section{Appendix for Section~\ref{sec:preliminaries}}\label{app:preliminaries}

\subsection{Additional Definitions and Remarks on the Semantics of $\ourLoopCPDL$}

\begin{enumerate}\itemsep0em
    \item For additional background on complexity theory and formal languages, consult Sipser's textbook~\cite[Chapters~1,~4--5, and~10.3]{sipser13}.

    \item We write $\N$ for the set of non-negative integers, $\cupdot$ for the disjoint union of sets, and either $\powerset{\rmS}$ or $2^{\rmS}$ for the power set of a set~$\rmS$.

    \item A \emph{homomorphism} between pointed Kripke structures $(\strA, \ela)$ and $(\strB, \elb)$ is any map $\homof\colon \domA \to \domB$ such that: (i) for any element $\elv \in \domA$ and any atomic proposition $\pP \in \AP$ we have $\elv \in \pP^{\strA}$ implies $\homof(\elv) \in \pP^{\strB}$, and (ii) for any pair of elements $\elu, \elv \in \domA$ and any action $\pe \in \PR$ we have $(\elu,\elv) \in \pe^{\strA}$ implies $(\homof(\elu),\homof(\elv)) \in \pe^{\strB}$. As a consequence, we also have that $(\elu,\elv) \in \inv{\pe}^{\strA}$ implies $(\homof(\elu),\homof(\elv)) \in \inv{\pe}^{\strB}$. By an \emph{isomorphism} we mean a bijective homomorphism $\homof$ such that its inverse $\homof^{-1}$ is a homomorphism as well. As homomorphisms preserve the semantics of atomic programs and tests, they also carry witnessing walks and cycles for $\autA_{q,q'}$ and~$\ela$ to witnessing walks and cycles for $\autA_{q,q'}$ and~$\homof(\ela)$.
    
    For isomorphic tree-like or cactus-like Kripke structures, we silently assume that the isomorphism $\homof$ preserves the inner structure of the tree: for every vertex $\elv$ and every direction $i$, if $\homof(\elv) = \elw$ then $\homof(\elv \cdot i) = \elw \cdot i$.  

    \item A \emph{nondeterministic finite automaton} (NFA) is a tuple $\autA \deff (Q, \Delta, q_0, q_f)$, where $(Q, \Delta)$ is a semiautomaton and $q_0, q_f$, called the initial and final states, respectively, belong to~$Q$.

    \item Let $\strA$ be a Kripke structure and let $\domB \subseteq \domA$.
    The \emph{substructure} of~$\strA$ \emph{induced by}~$\domB$ is the Kripke structure~$\strB$ with domain~$\domB$ in which $\pP^{\strB} \deff \pP^{\strA} \cap \domB$ and $\pe^{\strB} \deff \pe^{\strA} \cap (\domB \times \domB)$ for all $\pP \in \AP$ and all $\pe \in \PR$.

    \item In the introduction we claimed that every formula $\neg \lop(\pi)$ with $\circ$-free~$\pi$ belongs to $\ourLoopCPDL$. Here is a proof.
    Indeed, $\lop(\pi \cup \rho) \equiv \lop(\pi) \vee \lop(\rho)$, and $\lop(\pi^{*}) \equiv \top$ because $\sem{\pi^{*}}_{\strA}$ is reflexive.
    Writing $\pi$ as a union of components, $\neg\lop(\pi)$ is thus equivalent to~$\bot$ if some component is starred, and to a conjunction of formulae $\neg\lop(\sigma)$ with $\sigma \in \Sigma$ otherwise.
    The latter are in the fragment, as $\lop(\pP?) \equiv \pP$ and $\lop(\inv{\pe}) \equiv \lop(\pe)$.

    \item The syntax of $\ourLoopCPDL$ can easily be extended with $\neg\lop(\inv{\pe})$ for $\pe \in \PR$, since $\neg\lop(\inv{\pe}) \equiv \neg\lop(\pe)$. It can likewise be extended with $\neg\lop(\pi_1^+ \cup \cdots \cup \pi_n^+)$ for any finite list of programs $\pi_1, \ldots, \pi_n$; this, however, adds neither expressive power nor succinctness, as $\neg\lop(\pi_1^+ \cup \cdots \cup \pi_n^+) \equiv \neg\lop(\pi_1^+) \wedge \cdots \wedge \neg\lop(\pi_n^+)$. Finally, one could allow Kleene's star operator, but this would not increase the expressive power either, as $\neg\lop(\pi^*) \equiv \bot$.
    
    \item In contrast to $\PDL$, we allow only atomic propositions in tests, rather than arbitrary formul{\ae}. For $\PDL$ this costs no expressive power, as every test $\varphi?$ can be replaced by $\pP?$ for a fresh atomic proposition~$\pP$, provided one adds the defining equivalence $\mlbox{\univprog}(\pP \iff \varphi)$. For our logic, however, this is problematic: one could write $\mlbox{\lop(\pi)?}\bot$ in place of $\neg\lop(\pi)$, making the resulting extension of $\ourLoopCPDL$ as expressive as the full $\LoopCPDL$.
\end{enumerate}

\subsection{Proof of Lemma~\ref{lemma:normal-form}}

\begin{proof}
  The construction is the routine Scott-style renaming, so we only record the details that are worth checking.
  It proceeds in three steps: we name every subformula of~$\varphi$ by a fresh atomic proposition, then push the naming through the modalities so that the only remaining modal atoms are of the shapes $\mldiamond{\rho}\top$ and $\lop(\rho)$, and finally replace the programs~$\rho$ by automata modalities.
  The one step that is not routine is the last one: because the cycle demands are stated with $\leftrightarrows_{q,q'}$ rather than with~$\iff$, the fresh propositions no longer denote exactly the subformulae they name.
  We therefore replace, in Step~3, the usual ``faithful naming'' argument by a one-sided one, and it is exactly there that the restriction of $\ourLoopCPDL$ to negated loops of the shapes $\neg\lop(\pe)$ and $\neg\lop(\pi^+)$ is used.
  Two conventions are used throughout.
  First, by the syntax of $\ourLoopCPDL$, the input formula~$\varphi$ contains neither the universal program~$\univprog$ nor automata modalities, negation occurs in~$\varphi$ only in front of atomic propositions and of loops, and every negated loop is of the form $\neg\lop(\pe)$ or $\neg\lop(\pi^+)$.
  Second, $\univprog$ is interpreted as $\domB \times \domB$, so $\strB, \elu \models \mlbox{\univprog}{\Theta}$ holds if and only if $\Theta$ is satisfied at every world of~$\strB$; in particular, the truth of such a conjunct does not~depend~on~the~distinguished~world.

  \smallskip\noindent\textbf{Step 1: naming subformulae.}\,
  For every subformula~$\chi$ of~$\varphi$ take a fresh atomic proposition~$\pP_{\chi}$, and for every subformula of the form $\mlbox{\pi}\chi_1$ take one further fresh proposition~$\pQ_{\chi}$, meant to denote the complement of~$\pP_{\chi_1}$.
  Let $\Theta$ be the conjunction of the following \emph{defining equivalences}, taken over all subformulae~$\chi$ of~$\varphi$:
  \begin{itemize}\itemsep0em
    \item $\pP_{\chi} \iff \chi$, if $\chi$ is an atomic proposition or the negation of one;
    \item $\pP_{\chi} \iff (\pP_{\chi_1} \land \pP_{\chi_2})$, if $\chi = \chi_1 \land \chi_2$, and dually for~$\vee$;
    \item $\pP_{\chi} \iff \mldiamond{\pi \circ \pP_{\chi_1}?}{\top}$, if $\chi = \mldiamond{\pi}\chi_1$;
    \item $(\pQ_{\chi} \iff \neg \pP_{\chi_1}) \,\land\, (\pP_{\chi} \iff \neg\mldiamond{\pi \circ \pQ_{\chi}?}{\top})$, if $\chi = \mlbox{\pi}\chi_1$;
    \item $\pP_{\chi} \iff \lop(\pi)$, if $\chi = \lop(\pi)$, and $\pP_{\chi} \iff \neg\lop(\pi)$, if $\chi = \neg\lop(\pi)$; in the latter case $\pi$ is, by the syntax of $\ourLoopCPDL$, of the form $\pe$ or~$\rho^{+}$.
  \end{itemize}
  The third and the fourth item rely on the equivalences $\mldiamond{\pi}\pP \equiv \mldiamond{\pi \circ \pP?}{\top}$ and $\mlbox{\pi}\pP \equiv \neg\mldiamond{\pi \circ \pQ?}{\top}$, the latter under the proviso that $\pQ$ is interpreted as the complement of~$\pP$.
  The detour through~$\pQ_{\chi}$ is needed because tests in $\LoopCPDL$ are restricted to atomic propositions.

  We claim that $\varphi$ and $\varphi_1 \deff \pP_{\varphi} \land \mlbox{\univprog}{\Theta}$ are (finitely) equisatisfiable.
  For the forward direction, take a model $(\strA, \ela)$ of~$\varphi$ and expand~$\strA$ to~$\strA_1$ by interpreting each fresh proposition as the set it is meant to denote, that is, $\pP_{\chi}^{\strA_1} \deff \sem{\chi}_{\strA}$ and $\pQ_{\chi}^{\strA_1} \deff \domA \setminus \sem{\chi_1}_{\strA}$.
  All defining equivalences then hold at every world of~$\strA_1$, so $(\strA_1, \ela) \models \varphi_1$; as the domain is untouched, $\strA_1$ is finite whenever $\strA$ is.
  Conversely, let $(\strB, \elb) \models \varphi_1$.
  Then every defining equivalence holds at every world of~$\strB$, and a routine induction on the structure of~$\chi$ yields $\pP_{\chi}^{\strB} = \sem{\chi}_{\strB}$ for all subformulae~$\chi$ of~$\varphi$.
  Taking $\chi \deff \varphi$ and recalling that $\elb \in \pP_{\varphi}^{\strB}$, we conclude that $(\strB, \elb) \models \varphi$.
  Note that $\varphi_1$ is of polynomial size: it has at most twice as many fresh propositions as $\varphi$ has distinct subformulae, and each defining equivalence is of constant size except for the program that it carries---and programs occupy pairwise disjoint positions in~$\varphi$.

  \smallskip\noindent\textbf{Step 2: from programs to a single semiautomaton.}\,
  Every defining equivalence produced above is of the form $\pP \iff \psi$, where $\psi$ is a Boolean combination of atomic propositions and of atoms of the shape $\mldiamond{\rho}\top$ or $\lop(\rho)$ for a $\CPDL$-program~$\rho$.
  Moreover, as observed at the end of Step~1, the total size of the programs~$\rho$ arising in this way does not exceed~$|\varphi|$.
  Appending a dummy test~$\top?$ where needed (in particular, we replace $\pi^+$ with $(\pi\top?)^+$), we may assume that no such~$\rho$ matches the empty word.
  Hence Thompson's construction~\cite{Thompson68}, followed by the standard elimination of $\varepsilon$-transitions, turns each~$\rho$ in polynomial time into an equivalent NFA without $\varepsilon$-transitions, with a single initial state~$q_0^{\rho}$, a single final state~$q_f^{\rho} \neq q_0^{\rho}$, and polynomially many states in~$|\rho|$.
  Let $\autA \deff (Q, \Delta)$ be the disjoint union of these automata, with their state spaces renamed apart; note that $|Q|$ and $|\Delta|$ are polynomial in~$|\varphi|$.
  As the automata do not share states, inspecting the semantics gives $\sem{\autA_{q_0^{\rho}, q_f^{\rho}}}_{\strB} = \sem{\rho}_{\strB}$ for every such~$\rho$ and every Kripke structure~$\strB$.
  We may thus replace, inside the defining equivalences, every atom $\mldiamond{\rho}\top$ by~$\mldiamond{\autA_{q_0^{\rho}, q_f^{\rho}}}\top$ and every atom $\lop(\rho)$ by $\lop(\autA_{q_0^{\rho}, q_f^{\rho}})$, without affecting satisfaction in any structure.

  It remains to compute the symbols~$\leftrightarrows_{q,q'}$.
  By the first convention a negated loop of~$\varphi$ has the form $\neg\lop(\pe)$ or $\neg\lop(\pi^{+})$. The automaton built for~$\pe$ accepts the single word~$\pe$, hence is stationary; the one built for~$\pi^{+}$ denotes a transitive relation in every structure, hence is concatenation-closed. Appending dummy tests $\top?$ is harmless. We therefore put:
  \begin{center}
    $(\mathsf{Exact})$\quad for every subformula $\neg\lop(\rho)$ of~$\varphi$, the symbol $\leftrightarrows_{q_0^{\rho}, q_f^{\rho}}$ denotes~$\iff$.
  \end{center}
  We set all the remaining $\leftrightarrows$ symbols to $\to$.
  This concludes the second step.

  \smallskip\noindent\textbf{Step 3: naming the automata atoms.}\,
  Finally, take fresh atomic propositions $\wdem_{q,q'}$ and $\ldem_{q,q'}$ for \emph{all} pairs of states $q, q' \in Q$, and let $\Theta'$ be the result of replacing, in the defining equivalences, every atom $\mldiamond{\autA_{q,q'}}\top$ by~$\wdem_{q,q'}$ and every atom $\lop(\autA_{q,q'})$ by~$\ldem_{q,q'}$.
  Adding the walk and cycle demands
  $\textstyle\bigwedge_{q,q' \in Q} ( \wdem_{q,q'} \iff \mldiamond{\autA_{q,q'}}\top )$ and $\textstyle\bigwedge_{q,q' \in Q} ( \ldem_{q,q'} \leftrightarrows_{q,q'} \lop(\autA_{q,q'}) )$
  to~$\Theta'$ inside the universal program, and keeping $\pP_{\varphi}$ as the top-level conjunct, we arrive at the desired formula~$\varphi'$.
  Every conjunct of~$\Theta'$ is of the shape $\pP \iff \psi$ with $\psi$ a Boolean combination of atomic propositions, so $\varphi'$ is indeed in $\autA$-normal form, and it is of polynomial size by the size estimates of Steps~1 and~2.

  The two families of demands are not symmetric: a walk demand is an equivalence, whereas a cycle demand is one only at the pairs at which $\leftrightarrows_{q,q'}$ denotes~$\iff$, and elsewhere merely makes $\ldem_{q,q'}$ \emph{sufficient} for $\lop(\autA_{q,q'})$.
  The fresh propositions therefore need not denote the formulae they name; we obtain one inclusion only, and it is the one we need.

  \smallskip\noindent\emph{From $\varphi$ to $\varphi'$.}\,
  Recall that $\varphi$ and $\varphi_1$ are equivalent, so we work directly with the latter.
  Take any model $(\strA_1, \ela) \models \varphi_1$. Expand it to~$\strA_2$ by interpreting $\wdem_{q,q'}$ and $\ldem_{q,q'}$ as $\sem{\mldiamond{\autA_{q,q'}}\top}_{\strA_1}$ and $\sem{\lop(\autA_{q,q'})}_{\strA_1}$, for all $q, q' \in Q$, respectively.
  All demands then hold as equivalences, hence also as implications, and every fresh proposition denotes exactly the extension of the atom it replaces in~$\Theta$; thus $(\strA_2, \ela) \models \varphi'$, over the same domain.

  \smallskip\noindent\emph{From $\varphi'$ to $\varphi$ (assuming the claim holds).}\,
  Let $(\strB, \elb) \models \varphi'$, so that every conjunct under $\mlbox{\univprog}$ holds at every world of~$\strB$.
  All extensions below are taken in~$\strB$, and we omit the corresponding sub- and superscripts.
  We claim that every subformula~$\chi$ of~$\varphi$ satisfies (i)~$\pP_{\chi} \subseteq \sem{\chi}$, and (ii)~$\domB \setminus \sem{\chi_1} \subseteq \pQ_{\chi}$ in case $\chi = \mlbox{\pi}\chi_1$; the two inclusions point in opposite directions because $\pQ_{\chi}$ names a complement.
  Taking $\chi \deff \varphi$ in~(i) and recalling that $\elb \in \pP_{\varphi}$, we obtain $(\strB, \elb) \models \varphi$.
  Only the last case used the right-to-left implication of a cycle demand, which is why negated loops are restricted: otherwise $\ldem_{q,q'}$ could be a proper subset of $\sem{\lop(\pi)}$, and~(i) would~fail.

  \smallskip\noindent\emph{Inductive Proof of the Claim.}\,
  The proof is by induction on the syntactic complexity of~$\chi$, which is in negation normal form.
  If $\chi$ is an atomic proposition or the negation of one, then $\pP_{\chi} \iff \chi$ gives $\pP_{\chi} = \sem{\chi}$.
  If $\chi = \chi_1 \land \chi_2$, then $\pP_{\chi} = \pP_{\chi_1} \cap \pP_{\chi_2} \subseteq \sem{\chi_1} \cap \sem{\chi_2} = \sem{\chi}$ by~(i), and dually for~$\vee$.
  The two modal cases share the monotonicity of $\sem{\mldiamond{\pi \circ \pP?}{\top}}$ in the extension of the test~$\pP$.
  If $\chi = \mldiamond{\pi}\chi_1$, the walk demand gives $\pP_{\chi} = \sem{\mldiamond{\pi \circ \pP_{\chi_1}?}{\top}}$, which by~(i) and monotonicity is contained in $\sem{\mldiamond{\pi}\chi_1}$.
  If $\chi = \mlbox{\pi}\chi_1$, then (i) gives $\pQ_{\chi} = \domB \setminus \pP_{\chi_1} \supseteq \domB \setminus \sem{\chi_1}$, which is~(ii); monotonicity yields $\sem{\mldiamond{\pi \circ \pQ_{\chi}?}{\top}} \supseteq \sem{\mldiamond{\pi}\neg\chi_1}$, and taking complements gives~(i).
  Finally, let $\chi$ be $\lop(\pi)$ or $\neg\lop(\pi)$, and write $q, q'$ for the initial and the final state of the automaton built for~$\pi$ in Step~2, so that $\sem{\lop(\autA_{q,q'})} = \sem{\lop(\pi)}$.
  In the positive case $\pP_{\chi} = \ldem_{q,q'} \subseteq \sem{\lop(\pi)}$ by the left-to-right implication of the cycle demand, whichever symbol $\leftrightarrows_{q,q'}$ denotes.
  In the negative one $\pi$ is of the form $\pe$ or~$\rho^{+}$, so $(\mathsf{Exact})$ makes that demand an equivalence; hence $\ldem_{q,q'} = \sem{\lop(\pi)}$ and $\pP_{\chi} = \domB \setminus \ldem_{q,q'} = \sem{\neg\lop(\pi)}$.

  Each of the three steps leaves the domain of the given model unchanged and hence preserves finiteness, so $\varphi$ and~$\varphi'$ are (finitely) equisatisfiable. Moreover, $\varphi'$ is of polynomial size, and all the steps take polynomial time~in~$|\varphi|$.
\end{proof}

\section{Appendix for Section~\ref{sec:cactus-models}}\label{app:cactus-models}

As in Condition~\ref{phicactus:fini}, a vertex~$\elv$ of a structure~$\strC$ has a \emph{$\Sigma_{\varphi}$-edge} to a vertex~$\elw$ if $(\elv, \elw) \in \sem{\sigma}_{\strC}$ for some possibly inverted action $\sigma \in \Sigma_{\varphi}$; as $\Sigma_{\varphi}$ is closed under~$\inv{\cdot}$, the direction of such an edge is immaterial.

\subsection{Proof of Lemma~\ref{lemma:model-to-cactus}}
Fix a model $(\strA, \ela) \models \varphi$, and take $\strC$ and $\homof$ as constructed in the main body of the paper, with the notation used there.
Let $\rmD \deff \{ \wdem_{q,q'}, \ldem_{q,q'} \mid q, q' \in Q \}$ be the set of demands, and recall that $\kappa$ maps~$\rmD$ bijectively to $\{1, 2, \ldots, 2|Q|^2\}$.
Two features of the construction of $\strC$ are used throughout: every vertex is created exactly once---the root at the outset, every other vertex inside a single branch---and, the run being fair, every vertex is eventually processed.

\noindent\textbf{Reductions.}\,
Every walk~$\zeta$ in~$\strA$ factorises uniquely as $\zeta = \beta_0, \sigma_1, \beta_1, \ldots, \sigma_{n}, \beta_{n}$, where its \emph{$i$-th block} $\beta_i \deff \ela_i, \tau_i^{1}, \ela_i, \ldots, \tau_i^{k_i}, \ela_i$ repeats a single world and $\ela_i \neq \ela_{i+1}$ for all~$i < n$.
We call $\ela_0, \sigma_1, \ldots, \sigma_{n}, \ela_{n}$ the \emph{reduction} of~$\zeta$; it is what the construction of a branch obtains by removing consecutive duplicate worlds.
Note that the letters inside a block are \emph{self-loop letters}, \ie $(\ela_i, \ela_i) \in \sem{\tau_i^{j}}_{\strA}$, that the letters between blocks are actions or converse actions but never tests, and that $\zeta$ is stationary precisely when $n = 0$.
A non-stationary cycle~has~$n \geqslant 2$.

\noindent\textbf{Two useful facts.}\,
For all vertices $\elv \in \domC$, all $\pP \in \AP$, all $\sigma \in \Sigma$, and all states $q, q' \in Q$:
\begin{namedlist}
    \item[\desclabel{$(\mathsf{Prop})$}{cactus:prop}] $\elv \in \pP^{\strC}$ if and only if $\homof(\elv) \in \pP^{\strA}$;
    \item[\desclabel{$(\mathsf{Loop})$}{cactus:loop}] $(\elv, \elv) \in \sem{\sigma}_{\strC}$ if and only if $(\homof(\elv), \homof(\elv)) \in \sem{\sigma}_{\strA}$.
\end{namedlist}

\noindent
First, $\elv$ inherits from~$\homof(\elv)$ its atomic propositions and its atomic self-loops.
Second, inspecting the construction, we see that nothing else produces either.
This yields~\ref{cactus:prop} as well as~\ref{cactus:loop}: the case $\sigma = \pP{?}$ is~\ref{cactus:prop} itself, while $\sigma = \pe$ and $\sigma = \inv{\pe}$ both reduce to the equivalence of $(\elv, \elv) \in \pe^{\strC}$ and $(\homof(\elv), \homof(\elv)) \in \pe^{\strA}$, already handled by the preservation of atomic self-loops.

\begin{claim}
$\strC$ is cactus-like.
\end{claim}
\begin{proof}
We first show that $\domC$ is a tree of arity $|\rmD|{+}1 = 2|Q|^2{+}1$, as Definition~\ref{def:phi-cactus} requires, \ie a prefix-closed subset of $\{0, 1, \ldots, |\rmD|\}^{*}$.
This holds at the outset, where $\domC = \{\varepsilon\}$.
It is preserved by every processing step: processing a vertex~$\elv$ adds, for each demand~$\delta$ handled there, the vertices $\elv \kappa(\delta) 0^{i-1}$ with $1 \leqslant i \leqslant m$.
These lie in $\{0, 1, \ldots, |\rmD|\}^{*}$, and their proper prefixes are either prefixes of~$\elv$, and so already present, or added alongside them.
We next handle the transition relations.
Every pair entering a relation~$\pe^{\strC}$ is a self-loop or joins $\elv_i$ and~$\elv_{i+1}$ for some $i < n$ inside a branch.
In the second case, $\elv_{i+1}$ is a child of~$\elv_i$ whenever $i < m$.
A tree-like structure may contain self-loops and pairs of neighbours.
Only the pair joining $\elv_{m} = \elv \kappa(\delta) 0^{m-1}$ and $\elv_{n} = \elv$ for a cycle demand thus remains, and we claim that it is a back edge.
Its lower endpoint is an anchored descendant of~$\elv$.
Indeed, the direction $\kappa(\delta)$ is non-zero, and $\elv \kappa(\delta) 0^{m-1}$ is maximal.
For maximality, recall that every vertex is created once.
A $0$-child of $\elv \kappa(\delta) 0^{m-1}$ could therefore only be the $(m{+}1)$-st vertex of the branch for~$\delta$ at~$\elv$, which stops at the $m$-th.
Finally, that branch contributes exactly one of $(\elv_{m}, \elv)$ and $(\elv, \elv_{m})$: the former if $\sigma_{n} = \pe$, the latter if $\sigma_{n} = \inv{\pe}$.
Moreover, $\kappa$ being injective, no other branch ends at~$\elv \kappa(\delta) 0^{m-1}$.
Hence, we are done.
\end{proof}

\begin{claim}
    $\strC$ satisfies Condition~\ref{phicactus:fini}.
\end{claim}
\begin{proof}
    The first two parts rest on one observation.
    Fresh vertices appear only when a vertex~$\elu$ is processed, which adds, for each demand~$\delta$ handled there, the vertices $\elu \kappa(\delta) 0^{i-1}$ with $1 \leqslant i \leqslant m$.
    Only the first of them is a child of~$\elu$; each further one is the $0$-child of its predecessor.
    Hence the children of a vertex~$\elv$ are the vertices $\elv \kappa(\delta)$ for the demands~$\delta$ handled at~$\elv$, together with the $0$-child of~$\elv$ inside the branch in which~$\elv$ was created; in particular, the root, created outside every branch, has no $0$-child.
    Every vertex is created, and hence processed, once, so $\elv$ handles each demand at most once and lies in at most one branch.
    This confines every $0$-path to a single branch, and branches are finite: $m$ is bounded by the length of the reduction of the fixed witness~$\xi^{\homof(\elu)}_{\delta}$.
    For the last part, let $\elv \neq \varepsilon$, put $\elu_0 \deff \elu$ and $\elu_i \deff \elu \kappa(\delta) 0^{i-1}$ for the vertex~$\elu$ and the demand~$\delta$ whose branch created~$\elv$, and let $\elv = \elu_{i+1}$ with $i < m$; the parent of~$\elv$ is~$\elu_i$.
    That branch puts $(\elu_i, \elu_{i+1}) \in \sem{\sigma_{i+1}}_{\strC}$, where $\sigma_{i+1}$ is an action or its converse labelling a transition of~$\Delta$, hence occurring in~$\varphi$; so $\inv{\sigma_{i+1}} \in \Sigma_{\varphi}$.
    Thus $(\elv, \elu_i) \in \sem{\inv{\sigma_{i+1}}}_{\strC}$, that is, $\elv$ has a $\Sigma_{\varphi}$-edge to its parent.
\end{proof}

\begin{claim}
$\homof$ is a homomorphism with $\homof(\varepsilon) = \ela$.
\end{claim}
\begin{proof}
The equality is part of the initialisation, and atomic propositions are preserved by~\ref{cactus:prop}.
For the transition relations, take any pair $(\elu, \elw) \in \pe^{\strC}$.
By the classification above, it is either a self-loop, handled by~\ref{cactus:loop}, or the unique pair realising $(\elv_i, \elv_{i+1}) \in \sem{\sigma_{i+1}}_{\strC}$ inside a branch.
In the second case, $\sigma_{i+1} = \pe$ by construction of $\strC$ yields $(\elu, \elw) = (\elv_i, \elv_{i+1})$ with $(\homof(\elv_i), \homof(\elv_{i+1})) \in \pe^{\strA}$.
The case of $\sigma_{i+1} = \inv{\pe}$ is symmetric by swapping the vertices.
\end{proof}

\begin{claim}
    $\strC$ satisfies Conditions \ref{phicactus:bool}, \ref{phicactus:walk}, and \ref{phicactus:cycle}. Thus $(\strC, \varepsilon) \models \varphi$.
\end{claim}
\begin{proof}
All conjuncts of~$\varphi$ except~$\pP_{\varphi}$ sit under~$\mlbox{\univprog}$, so every world of~$\strA$ satisfies the conjunction inside that box; we use this silently.
Fix $\elv \in \domC$ and put $\elb \deff \homof(\elv)$.
Condition~\ref{phicactus:bool} is immediate from~\ref{cactus:prop}: it gives $\varepsilon \in \pP_{\varphi}^{\strC}$ from $\homof(\varepsilon) = \ela \in \pP_{\varphi}^{\strA}$, and lets $\elv$ inherit $\bigwedge_{i} (\pP_i \iff \psi_i)$ from~$\elb$.

We next establish~\ref{phicactus:walk}. 
Let $\elv$ satisfy $\wdem_{q,q'}$, abbreviated~$\delta$.
Then so does~$\elb$ by~\ref{cactus:prop}, and hence $\elb$ satisfies $\mldiamond{\autA_{q,q'}}\top$.
The construction thus fixed a witnessing walk~$\xi^{\elb}_{\delta}$ for $\autA_{q,q'}$ and~$\elb$; let $\ela_0, \sigma_1, \ldots, \sigma_{n}, \ela_{n}$ be its reduction.
Let $\elv_0 \deff \elv$ and $\elv_i \deff \elv \kappa(\delta) 0^{i-1}$ be the vertices of the branch for~$\delta$ at~$\elv$, so that $\homof(\elv_i) = \ela_i$.
Let $\hat{\xi}$ arise from~$\xi^{\elb}_{\delta}$ by replacing every world of its $i$-th block by~$\elv_i$.
For $n = 0$ there is no such branch, and $\hat{\xi}$ is the stationary walk at~$\elv$.
In either case $\hat{\xi}$ is a walk in~$\strC$.
Its steps inside the $i$-th block are self-loops at~$\elv_i$ by~\ref{cactus:loop}.
Its step to the $(i{+}1)$-st block is~$\sigma_{i+1}$, realised between $\elv_i$ and~$\elv_{i+1}$ by construction.
As $\hat{\xi}$ carries the letters of~$\xi^{\elb}_{\delta}$, it is a witnessing walk for $\autA_{q,q'}$ and~$\elv$.
For $n \geqslant 1$, deleting consecutive duplicates from its vertices leaves the anchored path $\elv, \elv\kappa(\delta), \ldots, \elv\kappa(\delta)0^{n-1}$, whose last vertex is maximal as shown above.
Hence $\hat{\xi}$ is stationary for $n = 0$ and anchored otherwise, as required.
For the converse direction, suppose $\elv$ has a witnessing walk for $\autA_{q,q'}$.
Then its $\homof$-image witnesses $\mldiamond{\autA_{q,q'}}\top$ at~$\elb$.
Thus $\elb$ satisfies~$\wdem_{q,q'}$, and~\ref{cactus:prop} carries the satisfaction of $\wdem_{q,q'}$ back to~$\elv$.

Condition~\ref{phicactus:cycle} is the same argument, read with ``cycle'' for ``walk'', with $\ldem_{q,q'}$ for~$\wdem_{q,q'}$, and with $\lop(\autA_{q,q'})$ for~$\mldiamond{\autA_{q,q'}}\top$, up to one point: a non-stationary witnessing cycle has $n \geqslant 2$, so its branch is non-empty, and the last step of~$\hat{\xi}$, leading from~$\elv_{n-1} = \elv \kappa(\delta) 0^{m-1}$ to~$\elv_{n} = \elv$, is precisely the back edge that the branch contributes.
Thus $\hat{\xi}$ returns to~$\elv$. Deleting consecutive duplicates and its final vertex leaves the anchored path $\elv, \elv\kappa(\delta), \ldots, \elv\kappa(\delta)0^{n-2}$.
\end{proof}

Note that the proof above, in particular the part concerning Condition~\ref{phicactus:cycle}, remains valid for arbitrary automata, not merely bannable ones; hence the entire construction of $\varphi$-cacti transfers verbatim from our fragment $\ourLoopCPDL$ to the full logic $\LoopCPDL$. The reason is that the domain of a $\varphi$-cactus is infinite and that we add a fresh branch for every demand, which renders the interaction between distinct cycles almost irrelevant.

\begin{claim}
If $\strA$ is finite, then $\strC$ is regular.
\end{claim}
\begin{proof}
\newcommand{\vkind}{\mathrm{kd}} 
\newcommand{\vres}{\mathrm{rs}}   
The proof idea is simple: we assign to every vertex of~$\strC$ a kind from a finite set, and show that vertices of equal kind have isomorphic subtrees.
Formally, the \emph{kind} of $\elv \in \domC$ is $\vkind(\elv) \deff (\homof(\elv), \vres(\elv))$.
Here the \emph{residue}~$\vres(\elv)$ is~$\varepsilon$ if $\elv$ is the root, and the triple $(\homof(\elu), \delta, i)$ if $\elv$ was created as the $i$-th vertex of the branch attached to~$\elu$ for~$\delta$.
As every vertex is created once, this is well defined.
Intuitively, the first component of $\vkind(\elv)$ records the origin of~$\elv$, and the second its reason for existing.
As $\strA$ is finite, only finitely many kinds occur: the index~$i$ is bounded by the length of $\xi^{\homof(\elu)}_{\delta}$, one of finitely many words fixed in advance.

As a preparatory step, we show that vertices of equal kind have children in the same directions, and that corresponding children again have equal kinds.
Indeed, let $\vkind(\elu) = \vkind(\elu')$.
The children of~$\elu$ in a non-zero direction are the vertices $\elu\kappa(\delta')$, one for each demand~$\delta'$ satisfied by~$\homof(\elu)$ with a non-stationary fixed witness.
A $0$-child exists precisely when $\vres(\elu) = (\elb', \delta', i)$ with $i$ below the length of the branch for~$(\elb', \delta')$.
Both conditions depend only on~$\vkind(\elu)$, as do the residues of the children, which are read off the corresponding~fixed~witness.

It remains to show that vertices of equal kind have isomorphic subtrees.
Take $\elv, \elv' \in \domC$ with $\vkind(\elv) = \vkind(\elv')$, and let $\homoh$ map $\elv \elw$ to~$\elv' \elw$.
By the preparatory step and induction on~$|\elw|$, this is a kind-preserving bijection between the subtrees rooted at~$\elv$ and at~$\elv'$.
It is moreover an isomorphism.
Indeed, it leaves the first component of the kind unchanged, so atomic propositions and self-loops are preserved by~\ref{cactus:prop} and~\ref{cactus:loop}.
The transition to each child and the back edge contributed by each cycle demand are read off the corresponding fixed witness, and hence agree at every vertex~$\elu$ of the subtree and at~$\homoh(\elu)$.
Since an anchored descendant lies in the subtree of its anchor, this covers every back edge inside the subtree.
Hence $\strC$ has at most as many subtrees up to isomorphism as $\vkind$ has values.
\end{proof}

Lemma~\ref{lemma:model-to-cactus} now follows from the established claims, concluding the proof.

\subsection{The Encoding Function and the Encoding-Ready Trees}\label{appendix:encoding-ready-trees}

Recall that a $\Sigma$-labelled tree is a pair $\treeT \deff (\domT, \ell)$ 
consisting of a tree $\domT$ and a labelling $\ell \colon \domT \to \Sigma$.
Throughout this section we work with \emph{exact} $\varphi$-cacti, 
\ie those in which $\ldem_{q,q'}$ holds precisely at the vertices having a 
witnessing cycle for $\autA_{q,q'}$ (note that by Condition~\ref{phicactus:walk}, 
walk demands are exact anyway). This is without loss of generality: a cycle 
demand with non-bannable $\autA_{q,q'}$ occurs only positively in~$\varphi$ and 
no automaton tests it, so in any model we may raise it to its exact value 
$\sem{\lop(\autA_{q,q'})}$, again obtaining a model of~$\varphi$.
Lemma~\ref{lemma:model-to-cactus} then yields an exact $\varphi$-cactus, 
since $\elv$ has a witnessing cycle for $\autA_{q,q'}$, resp.\ satisfies 
$\ldem_{q,q'}$, exactly when $\homof(\elv)$ does.
A motivation for this is to establish 
Remark~\ref{remark:general-loop-pdl} almost for free.

\begin{definition}[Encoding]\label{def:encoding}
Let $\strC \deff (\domC, \varepsilon)$ be a $\varphi$-cactus.
Its \emph{encoding} $\enc(\strC) \deff (\domT, \ell)$ is the complete 
infinite $(2|Q|^2{+}1)$-ary $2^{\Gamma_{\varphi}}$-labelled tree in which, 
for every vertex $\elv \in \domT$, every proposition 
$\pP \in \Sigma_{\varphi} \cap \AP$, every action 
$\pe \in \Sigma_{\varphi} \cap \PR$, and every 
possibly inverted action $\sigma \in \Sigma_{\varphi}$:
\begin{itemize}
  \item $\NULL \in \ell(\elv)$ iff $\elv \notin \domC$;
  \item $\pD_{\varepsilon} \in \ell(\elv)$ iff $\elv = \varepsilon$, and, for every 
  direction $0 \leqslant i \leqslant 2|Q|^2$, we have $\pD_{i} \in \ell(\elv)$ iff 
  $\elv \in \domC \setminus \{\varepsilon\}$ is the $i$-th child of its parent;
  \item $\pP \in \ell(\elv)$ iff $\elv \in \domC$ and $\elv \in \pP^{\strC}$;
  \item $\sigma \in \ell(\elv)$ iff $\elv \in \domC \setminus \{\varepsilon\}$ 
  and $(\elw, \elv) \in \sem{\sigma}_{\strC}$ for the parent $\elw$ of~$\elv$;
  \item $\pred{slf}_{\pe} \in \ell(\elv)$ iff $\elv \in \domC$ and $(\elv, \elv) \in \pe^{\strC}$;
  \item $\pred{bck}_{\sigma} \in \ell(\elv)$ iff $\elv \in \domC$ 
  is an anchored descendant whose anchor~$\elu$ satisfies $(\elv, \elu) \in \sem{\sigma}_{\strC}$.
\end{itemize}
\end{definition}

\noindent Thus a vertex outside the domain~$\domC$ is labelled $\{\NULL\}$ and nothing else, direction letter included, while a vertex of~$\domC$ carries its atomic propositions, exactly one direction letter---$\pD_{\varepsilon}$ at the root and $\pD_{i}$ at an $i$-th child---the letters recording the actions from its parent to it, one letter $\pred{slf}_{\pe}$ per atomic self-loop, and, if it is an anchored descendant, the letter of its back edge---unambiguous, as anchors are unique. Consult Figure~\ref{fig:encoding} for a visualisation.

We now isolate the properties characterising the trees of the form $\enc(\strC)$; we call such trees \emph{encoding-ready}.
Conditions~\ref{encready:A0}--\ref{encready:A5} transcribe Definition~\ref{def:phi-cactus} and are read off the shape of a $\varphi$-cactus.
Conditions~\ref{encready:B0}--\ref{encready:C4} close the demand labels under the ways in which witnessing walks and cycles arise; they are what turns the ``otherwise'' clauses of~\ref{phicactus:walk} and~\ref{phicactus:cycle} into a property that a tree automaton can enforce.
Three auxiliary notions are needed first.

\vspace{0.25em}
\noindent \textbf{$\NULL$-vertices and lifted tree vocabulary.}\,
A vertex $\elv$ is a \emph{$\NULL$-vertex} if $\NULL \in \ell(\elv)$.
We lift the notions of \emph{anchor}, \emph{anchored descendant}, \emph{anchored path}, and \emph{stationary path} from Section~\ref{sec:preliminaries} to labelled trees, reading ``$\elw$ has no $0$-child'' as ``$\elw$ is a non-$\NULL$-vertex whose $0$-child is a $\NULL$-vertex'', and so on.

\vspace{0.25em}
\noindent \textbf{Witnessing paths.}\,
Put $\Gamma_{\varphi}^{0} \deff (\Sigma_{\varphi} \cap \AP) \cup \{ \pred{slf}_{\pe} \mid \pe \in \Sigma_{\varphi} \cap \PR \}$ and $\Gamma_{\varphi}^{1} \deff \{ \pe, \inv{\pe} \mid \pe \in \Sigma_{\varphi} \cap \PR \}$:
a letter of $\Gamma^{0}_{\varphi}$ present at a vertex records a move that stays there, one of $\Gamma^{1}_{\varphi}$ a move that crosses the tree edge to the parent.
We \emph{read} such letters as letters of~$\Sigma$: a proposition~$\pP$ is read as the test~$\pP{?}$, a letter $\pred{slf}_{\pe}$ as $\pe$ or as $\inv{\pe}$ (a self-loop offers both), and a letter of $\Gamma^{1}_{\varphi}$ as itself.
For neighbouring non-$\NULL$ vertices $\elx, \ely$ we let $\theta(\elx,\ely) \deff \Gamma^{1}_{\varphi} \cap \ell(\ely)$ if $\ely$ is a child of~$\elx$, and $\theta(\elx,\ely) \deff \{ \inv{\sigma} \mid \sigma \in \Gamma^{1}_{\varphi} \cap \ell(\elx) \}$ if $\elx$ is a child of~$\ely$; these are the letters realising a move from~$\elx$ to~$\ely$.
A path $\elv_0, \elv_1, \ldots, \elv_n$ of non-$\NULL$ vertices is a \emph{witnessing path} for $\autA_{q,q'}$ if $\autA_{q,q'}$ accepts a word $w_0 w_1 \cdots w_n$ in which $w_0$ is a reading of a word over $\Gamma^{0}_{\varphi} \cap \ell(\elv_0)$ and, for each $i \geqslant 1$, $w_i$ is a reading of a word from $\theta(\elv_{i-1},\elv_i) \cdot (\Gamma^{0}_{\varphi} \cap \ell(\elv_i))^{*}$.
Such a path is \emph{stationary} if $n = 0$, \emph{downward} (resp. \emph{upward}) if every $\elv_{i+1}$ is a child (resp. the parent) of~$\elv_i$, and \emph{anchored in direction~$j$} if it is the anchored path $\elv_0, \elv_0 j, \elv_0 j 0, \ldots, \elv_n$. 
The same notions for the extension~$\hat{\autA}_{q,q'}$ are obtained by reading, in addition, each cycle demand $\ldem_{p,p'} \in \Gamma^{0}_{\varphi}$~as~the~test~$\ldem_{p,p'}?$.

\vspace{0.25em}
\noindent \textbf{$\sigma$-successors.}\,
Let $\sigma$ be a possibly inverted action.
A non-$\NULL$ vertex~$\elv$ is a \emph{$\sigma$-successor} of a non-$\NULL$ vertex~$\elw$ if
(i)~$\elv$ is a child of~$\elw$ and $\sigma \in \ell(\elv)$;
(ii)~$\elw$ is a child of~$\elv$ and $\inv{\sigma} \in \ell(\elw)$;
(iii)~$\elv$ is an anchored descendant of~$\elw$ and $\pred{bck}_{\inv{\sigma}} \in \ell(\elv)$; or
(iv)~$\elw$ is an anchored descendant of~$\elv$~and~$\pred{bck}_{\sigma} \in \ell(\elw)$.

\noindent We now define encoding-ready trees. The definition is technical, but it is what lets a tree automaton check the properties of $\varphi$-cacti. For the reader's convenience we interleave the formal definitions with their intuitive explanations.

\begin{definition}[Encoding-Ready Trees]\label{def:encoding-ready-trees}
  A labelled tree $\treeT \deff (\domT, \ell)$ is \emph{encoding-ready} if it satisfies all the conditions below, for all states $q, q', q'' \in Q$ and all non-$\NULL$ vertices $\elv$.
    \begin{namedlist}[$(\mathrm{A}_0)$\hspace{0.4em}]
      \item[\desclabel{$(\mathrm{A}_0)$}{encready:A0}]
      $\treeT$ is a complete infinite $(2|Q|^2{+}1)$-ary $2^{\Gamma_{\varphi}}$-labelled tree; $\ell(\elu) = \{\NULL\}$ for every $\NULL$-vertex~$\elu$, $\ell(\elv)$ meets $\Gamma^{1}_{\varphi}$ if and only if $\elv \neq \varepsilon$, a letter $\pred{bck}_{\sigma}$ lies in $\ell(\elv)$ only if $\elv$ is an anchored descendant of some vertex, and the direction letters are \emph{faithful}: $\ell(\elv)$ meets $\Gamma^{\mathrm{dir}}_{\varphi}$ in exactly one letter, which is $\pD_{\varepsilon}$ for $\elv = \varepsilon$ and $\pD_{i}$ for $\elv$ the $i$-th child~of~its~parent.\\
      \why{The basic shape of our trees, plus the bookkeeping implicit in~$\enc$: outside the domain nothing but $\NULL$ is recorded, and the root has no parent, hence no incoming edge to record, whereas every other vertex records at least the $\Sigma_{\varphi}$-edge to its parent demanded by Condition~\ref{phicactus:fini}. The direction letters are redundant for a reader, who sees the address of~$\elv$; not for a 2APTA, which sees only labels. They are what lets a copy that has just moved up decide whether it came from a $0$-child, and so recognise anchored paths while climbing---as Condition~\ref{encready:C4} and, later, the foldability automaton both require.}

      \item[\desclabel{$(\mathrm{A}_1)$}{encready:A1}]
      all $0$-paths consisting of non-$\NULL$ vertices are finite, and the $0$-child of the root~$\varepsilon$ is a $\NULL$-vertex.
      \why{Condition~\ref{phicactus:fini}; its last part, like the arity fixed in Definition~\ref{def:phi-cactus}, is already part of~\ref{encready:A0}.}

      \item[\desclabel{$(\mathrm{A}_2)$}{encready:A2}]
      all descendants of a $\NULL$-vertex are $\NULL$-vertices.

      \item[\desclabel{$(\mathrm{A}_3)$}{encready:A3}]
      $\pP_{\varphi} \in \ell(\varepsilon)$, and the propositional valuation making true precisely the propositions of $\ell(\elv) \cap \AP$ satisfies $\textstyle\bigwedge_{i=1}^{N_\varphi} (\pP_i \iff \psi_i)$.
      \why{A direct transcription of Condition~\ref{phicactus:bool}.}

      \item[\desclabel{$(\mathrm{A}_4)$}{encready:A4}]
      if $\wdem_{q,q'} \in \ell(\elv)$, then $\elv$ has a stationary witnessing path for $\autA_{q,q'}$, or one anchored in direction $\kappa(\wdem_{q,q'})$.
      \why{A direct transcription of Condition~\ref{phicactus:walk}.}

      \item[\desclabel{$(\mathrm{A}_5)$}{encready:A5}]
      if $\ldem_{q,q'} \in \ell(\elv)$, then $\elv$ has a stationary witnessing path for $\autA_{q,q'}$, or there are states $q_1, q_2 \in Q$ and a possibly inverted action~$\sigma$ with $(q_1,\sigma,q_2) \in \Delta$ such that $\elv$ has a witnessing path for $\autA_{q,q_1}$ anchored in direction $\kappa(\ldem_{q,q'})$ and ending on a vertex~$\elw$ of which $\elv$ is a $\sigma$-successor, as well as a stationary witnessing path for $\autA_{q_2,q'}$.

      \why{A~less~direct~transcription of Condition~\ref{phicactus:cycle}. A non-stationary witnessing cycle at~$\elv$ is anchored, so it descends the branch reserved for the demand, closes with a single step from the last vertex~$\elw$ of that branch back to~$\elv$, and can afterwards only idle at~$\elv$; the three requirements state exactly this split into a downward part, the closing step, and a stationary tail.}

      \item[\desclabel{$(\mathrm{B}_0)$}{encready:B0}]
      $\ldem_{q,q} \in \ell(\elv)$ and $\wdem_{q,q} \in \ell(\elv)$.
      \why{Not moving at all is a witnessing cycle for~$\autA_{q,q}$.}

      \item[\desclabel{$(\mathrm{B}_1)$}{encready:B1}]
      if $\ldem_{q,q'} \in \ell(\elv)$, then $\wdem_{q,q'} \in \ell(\elv)$.
      \why{A witnessing cycle is in particular a witnessing walk.}

      \item[\desclabel{$(\mathrm{B}_2)$}{encready:B2}]
      if $(q, \pP{?}, q') \in \Delta$ and $\pP \in \ell(\elv)$, then $\ldem_{q,q'} \in \ell(\elv)$.
      \why{If $\elv$ is labelled with $\pP$, then $\elv, \pP{?}, \elv$ is a witnessing cycle for $\autA_{q,q'}$ at~$\elv$.}

      \item[\desclabel{$(\mathrm{B}_3)$}{encready:B3}]
      if $(q, \pe, q') \in \Delta$ or $(q, \inv{\pe}, q') \in \Delta$, and $\pred{slf}_{\pe} \in \ell(\elv)$, then $\ldem_{q,q'} \in \ell(\elv)$.
      \why{An $\pe$-self-loop at~$\elv$ makes both $\elv, \pe, \elv$ and $\elv, \inv{\pe}, \elv$ witnessing cycles at~$\elv$.}

      \item[\desclabel{$(\mathrm{B}_4)$}{encready:B4}]
      if $\ldem_{q,q'} \in \ell(\elv)$ and $\wdem_{q',q''} \in \ell(\elv)$, then $\wdem_{q,q''} \in \ell(\elv)$.
      \why{Prefixing a witnessing walk for $\autA_{q',q''}$ with a witnessing cycle for $\autA_{q,q'}$ at its first world yields a witnessing walk for~$\autA_{q,q''}$.}

      \item[\desclabel{$(\mathrm{B}_5)$}{encready:B5}]
      if $\ldem_{q,q'} \in \ell(\elv)$ and $\ldem_{q',q''} \in \ell(\elv)$, then $\ldem_{q,q''} \in \ell(\elv)$.
      \why{Concatenating two witnessing cycles at the same world yields a witnessing cycle.}

      \item[\desclabel{$(\mathrm{C}_1)$}{encready:C1}]
      let $q_1,q_2,q_3,q_4 \in Q$ and let $\sigma, \tau$ be possibly inverted actions with $(q_1,\sigma,q_2), (q_3,\tau,q_4) \in \Delta$.
      If $\elv$ is a $\sigma$-successor of a non-$\NULL$ vertex~$\elw$, $\elw$ is a $\tau$-successor of~$\elv$, and $\ldem_{q_2,q_3} \in \ell(\elv)$, then $\ldem_{q_1,q_4} \in \ell(\elw)$.
      \why{This captures the scenario of stepping from~$\elw$ to~$\elv$ along~$\sigma$, cycling at~$\elv$, and returning along~$\tau$ is a witnessing cycle for $\autA_{q_1,q_4}$ at~$\elw$.}

      \item[\desclabel{$(\mathrm{C}_2)$}{encready:C2}]
      let $q_1,q_2,q_3 \in Q$ and let $\sigma$ be a possibly inverted action with $(q_1,\sigma,q_2) \in \Delta$.
      If $\elv$ is a $\sigma$-successor of a non-$\NULL$ vertex~$\elw$ and $\wdem_{q_2,q_3} \in \ell(\elv)$, then $\wdem_{q_1,q_3} \in \ell(\elw)$.
      \why{Prefixing a witnessing walk starting at~$\elv$ with the step from~$\elw$ to~$\elv$.}

      \item[\desclabel{$(\mathrm{C}_3)$}{encready:C3}]
      let $q_1,q_2,q_3,q_4 \in Q$, let $\sigma$ be a possibly inverted action with $(q_2,\sigma,q_3) \in \Delta$, let $\elw$ be a non-$\NULL$ anchored descendant of its anchor~$\elu$ with $\pred{bck}_{\sigma} \in \ell(\elw)$, and let $\elv$ lie on the anchored path from~$\elu$ to~$\elw$.
      If there is a downward witnessing path for $\hat{\autA}_{q_1,q_2}$ from~$\elv$ to~$\elw$ and one for $\hat{\autA}_{q_3,q_4}$ from~$\elu$ to~$\elv$, then $\ldem_{q_1,q_4} \in \ell(\elv)$.

      \item[\desclabel{$(\mathrm{C}_4)$}{encready:C4}]
      as in~\ref{encready:C3}, but with $\pred{bck}_{\inv{\sigma}} \in \ell(\elw)$ and with the two paths reversed: if there is an upward witnessing path for $\hat{\autA}_{q_1,q_2}$ from~$\elv$ to~$\elu$ and one for $\hat{\autA}_{q_3,q_4}$ from~$\elw$ to~$\elv$, then $\ldem_{q_1,q_4} \in \ell(\elv)$.
      \why{Conditions~\ref{encready:C3} and~\ref{encready:C4} are the two ways of travelling once around the cycle that a back edge closes with its anchored path, in each case letting the vertex~$\elv$ at which the travel starts record the resulting cycle; note that $\elv$ may coincide with~$\elu$ or with~$\elw$. Using $\hat{\autA}$ rather than $\autA$ is essential: it absorbs, into a single test, every detour that the travel makes off the anchored path, such a detour being itself a cycle already recorded at the vertex where it starts. See the bottom row of Figure~\ref{fig:encready}.}
    \end{namedlist}
\end{definition}

\noindent Observe that \ref{encready:A4} and~\ref{encready:A5} provide witnesses for the demands that a vertex carries, whereas \ref{encready:B0}--\ref{encready:C4} do the converse: they force a vertex to carry every demand it can witness.
Moreover, no condition treats a self-loop as a $\sigma$-successor step, and none is needed: the case $\elv = \elw$ of~\ref{encready:C1} follows from \ref{encready:B3} and \ref{encready:B5}, and that of~\ref{encready:C2} from \ref{encready:B3} and~\ref{encready:B4}. 

\begin{figure}[htb]
  \centering
  \begin{tikzpicture}[
      x=1cm, y=1cm, every node/.style={transform shape},
      wld/.style   = {circle, draw=black!75, fill=white, line width=0.7pt,
                      inner sep=0pt, minimum size=6.6pt},
      zedge/.style = {->, >=stealth, draw=blue!55!black, line width=0.75pt,
                      dash pattern=on 0.6pt off 1.6pt},
      bedge/.style = {->, >=stealth, draw=red!75!black, line width=0.85pt,
                      dash pattern=on 3pt off 2pt},
      sedge/.style = {->, >=stealth, draw=orange!85!black, line width=0.95pt},
      wit/.style   = {->, >=stealth, draw=orange!85!black, line width=0.9pt,
                      dash pattern=on 0.7pt off 1.9pt},
      vname/.style = {font=\footnotesize, text=black!75, inner sep=1.5pt},
      hyp/.style   = {font=\scriptsize, inner sep=2.2pt, rounded corners=1.6pt,
                      draw=orange!75!black, fill=orange!10, text=orange!55!black},
      con/.style   = {font=\scriptsize, inner sep=2.2pt, rounded corners=1.6pt,
                      draw=vert, fill=vert!8, text=vert},
      olab/.style  = {font=\scriptsize, text=orange!55!black, inner sep=1.6pt},
      rlab/.style  = {font=\scriptsize, text=red!65!black,    inner sep=1.6pt},
      ptag/.style  = {font=\footnotesize, text=black!85, inner sep=1.5pt},
      leg/.style   = {font=\scriptsize, text=black!75, anchor=west, inner sep=1.5pt},
    ]

  \begin{scope}[shift={(1.55,0)}]
    \node[wld] (w1) at (-0.95,0) {};
    \node[wld] (v1) at ( 0.95,0) {};
    \draw[sedge] (w1) to[bend left=26]  node[olab, above=0pt] {$\sigma$} (v1);
    \draw[sedge] (v1) to[bend left=26]  node[olab, below=0pt] {$\tau$}   (w1);
    \draw[wit] (v1) .. controls +(42:1.15) and +(138:1.15) .. (v1);
    \node[vname, anchor=east]  at (-1.16, 0)    {$\elw$};
    \node[vname, anchor=north] at ( 0.95,-0.14) {$\elv$};
    \node[hyp, anchor=south]   at ( 0.95, 0.70) {$\ldem_{q_2,q_3}$};
    \node[con, anchor=north]   at (-0.95,-0.50) {$\ldem_{q_1,q_4}$};
    \node[ptag] at (0,-1.85) {$(\mathrm{C}_1)$};
  \end{scope}

  \begin{scope}[shift={(4.95,0)}]
    \node[wld] (w2) at (-0.85,0) {};
    \node[wld] (v2) at ( 0.85,0) {};
    \node[wld] (x2) at ( 1.70,0) {};
    \draw[sedge] (w2) -- node[olab, above=0pt] {$\sigma$} (v2);
    \draw[wit, -] (v2) -- (x2);
    \draw[wit]    (x2) -- (2.32,0);
    \node[vname, anchor=east]  at (-1.06, 0)    {$\elw$};
    \node[vname, anchor=north] at ( 0.85,-0.14) {$\elv$};
    \node[hyp, anchor=south]   at ( 0.85, 0.32) {$\wdem_{q_2,q_3}$};
    \node[con, anchor=north]   at (-0.85,-0.42) {$\wdem_{q_1,q_3}$};
    \node[ptag] at (0.50,-1.85) {$(\mathrm{C}_2)$};
  \end{scope}

  \begin{scope}[shift={(8.80,0)}]
    \node[wld] (u3) at (0, 0.95) {};
    \node[wld] (v3) at (0, 0.00) {};
    \node[wld] (w3) at (0,-0.95) {};
    \draw[zedge] (u3) -- (v3);  \draw[zedge] (v3) -- (w3);
    \draw[bedge] (w3) to[out=195,in=165,looseness=1.35] (u3);
    \node[rlab, anchor=east] at (-0.9,0) {$\sigma$};
    \node[rlab, anchor=west] at (0.14, 0.95) {$q_3$};
    \node[rlab, anchor=west] at (0.14,-0.95) {$q_2$};
    \draw[sedge] (0.38, 0.85) -- node[olab, right=-1pt] {$\hat{\autA}_{q_3,q_4}$} (0.38, 0.10);
    \draw[sedge] (0.38,-0.10) -- node[olab, right=-1pt] {$\hat{\autA}_{q_1,q_2}$} (0.38,-0.85);
    \node[vname, anchor=south] at (0, 1.10) {$\elu$};
    \node[vname, anchor=north] at (0,-1.10) {$\elw$};
    \node[vname, anchor=east]  at (-0.20,0) {$\elv$};
    \node[con, anchor=west]    at ( 0.55,0) {$\ldem_{q_1,q_4}$};
    \node[ptag] at (0.35,-1.85) {$(\mathrm{C}_3)$};
  \end{scope}

  \begin{scope}[shift={(11.90,0)}]
    \node[wld] (u4) at (0, 0.95) {};
    \node[wld] (v4) at (0, 0.00) {};
    \node[wld] (w4) at (0,-0.95) {};
    \draw[zedge] (u4) -- (v4);  \draw[zedge] (v4) -- (w4);
    \draw[bedge] (u4) to[out=165,in=195,looseness=1.35] (w4);
    \node[rlab, anchor=east] at (-0.9,0) {$\sigma$};
    \node[rlab, anchor=west] at (0.14, 0.95) {$q_2$};
    \node[rlab, anchor=west] at (0.14,-0.95) {$q_3$};
    \draw[sedge] (0.38, 0.10) -- node[olab, right=-1pt] {$\hat{\autA}_{q_1,q_2}$} (0.38, 0.85);
    \draw[sedge] (0.38,-0.85) -- node[olab, right=-1pt] {$\hat{\autA}_{q_3,q_4}$} (0.38,-0.10);
    \node[vname, anchor=south] at (0, 1.10) {$\elu$};
    \node[vname, anchor=north] at (0,-1.10) {$\elw$};
    \node[vname, anchor=east]  at (-0.20,0) {$\elv$};
    \node[con, anchor=west]    at ( 0.55,0) {$\ldem_{q_1,q_4}$};
    \node[ptag] at (0.35,-1.85) {$(\mathrm{C}_4)$};
  \end{scope}

  \begin{scope}[shift={(0,-2.65)}]
    \draw[zedge] (0.70,0) -- (1.17,0);  \node[leg] at (1.29,0) {anchor path};
    \draw[bedge] (3.29,0) -- (3.76,0);  \node[leg] at (3.88,0) {back edge};
    \draw[sedge] (5.68,0) -- (6.15,0);  \node[leg] at (6.27,0) {$\sigma$-successor step};
    \draw[wit]   (9.27,0) -- (9.74,0);  \node[leg] at (9.86,0) {witnessing walk or cycle};
    \node[leg, anchor=center] at (6.75,-0.55)
         {a boxed demand at a vertex lies in its label:\;
          \tikz[baseline=-0.55ex]{\node[hyp,font=\scriptsize]{orange};} if assumed,\;
          \tikz[baseline=-0.55ex]{\node[con,font=\scriptsize]{green};} if forced};
  \end{scope}
  \end{tikzpicture}
  \caption{The closure Conditions \ref{encready:C1}--\ref{encready:C4}.
    In $(\mathrm{C}_1)$ and $(\mathrm{C}_2)$ each orange edge between $\elw$
    and~$\elv$ stands for any of the four cases~(i)--(iv) of the definition
    of a $\sigma$-successor; in $(\mathrm{C}_3)$ and $(\mathrm{C}_4)$ the
    vertex $\elv$ may coincide with~$\elu$ or with~$\elw$, and the red arrow
    shows the direction in which the back edge is traversed.
    The transitions $(q_1,\sigma,q_2), (q_3,\tau,q_4) \in \Delta$ required
    by the conditions are not depicted.}
  \label{fig:encready}
\end{figure}
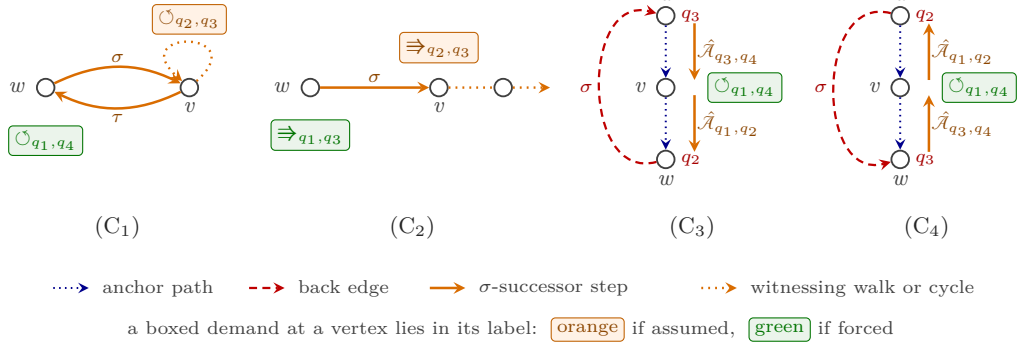

As the next step towards the proof of Lemma~\ref{lemma:cactus-automaton},
we show that the encoding of a $\varphi$-cactus is encoding-ready and, 
conversely, that every encoding-ready tree is the encoding of some 
$\varphi$-cactus (actually such a cactus will be unique). This is done in the 
next two sections. To keep the proofs modular, we first record three simple 
observations, used throughout.

\begin{observation}
Let $\strC$ be a $\varphi$-cactus and let $\treeT \deff (\domT,\ell) \deff \enc(\strC)$.
\begin{namedlist}[$(O_3)$\hspace{0.4em}]
  \item[\desclabel{$(O_1)$}{dict:dom}] The set $\domC$ is precisely the set of non-$\NULL$ vertices of~$\treeT$, and the lifted notions of anchor, anchored descendant, and anchored path in~$\treeT$ agree with the original ones in~$\strC$.
  \item[\desclabel{$(O_2)$}{dict:edge}] Let $\elu, \elv \in \domC$ and let $\sigma$ be a possibly inverted action with underlying action~$\pe$. Then $(\elv,\elv) \in \sem{\sigma}_{\strC}$ if and only if $\pred{slf}_{\pe} \in \ell(\elv)$; and, for $\elu \neq \elv$, we have $(\elu,\elv) \in \sem{\sigma}_{\strC}$ if and only if $\elv$ is a $\sigma$-successor of~$\elu$ in~$\treeT$.
  \item[\desclabel{$(O_3)$}{dict:path}] A path $\elv_0, \ldots, \elv_n$ of non-$\NULL$ vertices is a witnessing path for $\autA_{q,q'}$ (respectively $\hat{\autA}_{q,q'}$) in~$\treeT$ if and only if $\strC$ has a witnessing walk for $\autA_{q,q'}$ (respectively $\hat{\autA}_{q,q'}$) from~$\elv_0$ to~$\elv_n$ that uses no back edge and whose sequence of worlds, after deleting consecutive duplicates, is $\elv_0, \ldots, \elv_n$. 
\end{namedlist}
\end{observation}
\begin{proof}
  Indeed, \ref{dict:dom} is immediate from the first item of Definition~\ref{def:encoding} and from the reading of ``has no $0$-child''.
  For~\ref{dict:edge}, note that, as $\strC$ is cactus-like, a pair $(\elu,\elv) \in \pe^{\strC}$ with $\elu \neq \elv$ either joins two neighbours or joins an anchor with one of its anchored descendants. Cases~(i)--(iv) of the definition of a $\sigma$-successor are precisely these four possibilities, and each is covered by the corresponding item of Definition~\ref{def:encoding}.

For~\ref{dict:path}, note first that, by Definition~\ref{def:encoding} and~\ref{dict:edge}, the letters looping at a non-$\NULL$ vertex~$\elv$, that is those $\gamma \in \Gamma_{\varphi}$ with $(\elv,\elv) \in \sem{\gamma}_{\strC}$, are precisely the elements of $\Gamma^{0}_{\varphi} \cap \ell(\elv)$: the tests $\pP{?}$ with $\pP \in \ell(\elv)$, the possibly inverted actions~$\sigma$ with $\pred{slf}_{\pe} \in \ell(\elv)$ for the underlying~$\pe$, and, for~$\hat{\autA}$, the tests $\ldem_{p,p'}{?}$ with $\ldem_{p,p'} \in \ell(\elv)$. Similarly, the letters~$\gamma$ with $(\elu,\elv) \in \sem{\gamma}_{\strC}$ for a tree edge $\elu \neq \elv$ are precisely those of $\theta(\elu,\elv)$. Both notions of a witness are therefore an accepted word together with a sequence of worlds realising its letters in the transition relations of~$\strC$; they differ only in that a walk records one world per letter, whereas a path records one vertex per edge traversed. The conversion between them therefore amounts just to inserting and deleting consecutive duplicates, possibly after changing the letters of $\Gamma^{0}_{\varphi}$ into their corresponding readings and vice versa.
\end{proof}

\subsection{From Phi-Cacti to Encoding-Ready Trees}\label{appendix:from-phi-cacti-to-encoding-ready-trees}

We now show that the encoding of a $\varphi$-cactus yields an encoding-ready 
tree. The proof is routine: the definition of a $\varphi$-cactus was tailored 
precisely to make this so. It thus suffices to go through the conditions of 
Definition~\ref{def:encoding-ready-trees} one by one and verify that the encoding 
satisfies each of them; no hidden difficulties arise.

\begin{lemma}\label{lemma:from-cacti-to-encoding-ready-trees}
  Let $\strC$ be a $\varphi$-cactus. Then $\enc(\strC)$ is an encoding-ready tree.
  Moreover, if $\strC$ is regular, then $\enc(\strC)$ is regular as well.
\end{lemma}
\begin{proof}
  Write $\treeT \deff (\domT, \ell) \deff \enc(\strC)$ and fix a non-$\NULL$ vertex~$\elv$.
  
  Conditions~\ref{encready:A0}--\ref{encready:A3} are read off Definition~\ref{def:encoding} and Definition~\ref{def:phi-cactus}, with~\ref{dict:dom} matching the lifted tree vocabulary to the original one: \ref{encready:A0} collects the arity fixed in Definition~\ref{def:phi-cactus}, the ``only if'' halves of the items of Definition~\ref{def:encoding}, as discussed after that definition (its second item places at every vertex of~$\domC$ exactly the direction letter that~\ref{encready:A0} demands, and none outside~$\domC$), and the last part of~\ref{phicactus:fini} (every $\elv \neq \varepsilon$ has $(\elv, \elw) \in \sem{\sigma}_{\strC}$ for its parent~$\elw$ and some possibly inverted action $\sigma \in \Sigma_{\varphi}$, so the fourth item of Definition~\ref{def:encoding} places $\inv{\sigma} \in \Gamma^{1}_{\varphi}$ in~$\ell(\elv)$); \ref{encready:A1} is the rest of~\ref{phicactus:fini}; \ref{encready:A2} holds as $\domC$ is prefix-closed; and \ref{encready:A3} is~\ref{phicactus:bool}, since $\ell(\elv) \cap \AP$ is the set of propositions satisfied by~$\elv$ in~$\strC$.

  For~\ref{encready:A4} and~\ref{encready:A5}, let $\delta \in \ell(\elv)$ be a demand and let $\xi^{\elv}_{\delta}$ be the witness for it granted by~\ref{phicactus:walk} or~\ref{phicactus:cycle}.
  If $\xi^{\elv}_{\delta}$ is stationary, then so is the witnessing path that~\ref{dict:path} returns, and we are done; so let $\xi^{\elv}_{\delta}$ be anchored.
  Deleting consecutive duplicates from~$\xi^{\elv}_{\delta}$, and its final vertex if $\delta$ is a cycle demand, leaves the anchored path $\elv, \elv\kappa(\delta), \ldots, \elw$, whose consecutive vertices are a parent and a child.
  For $\delta$ equal to  $\wdem_{q,q'}$ this shows that $\xi^{\elv}_{\delta}$ uses no back edge, so~\ref{dict:path} turns it into the required witnessing path.
  For $\delta$ equal to $\ldem_{q,q'}$ it splits $\xi^{\elv}_{\delta}$ into a back-edge-free walk from~$\elv$ down to~$\elw$, a single step back to~$\elv$, and a stationary walk at~$\elv$; by~\ref{dict:path} the outer two parts are witnessing paths for $\autA_{q,q_1}$ and $\autA_{q_2,q'}$, where $q_1, q_2$ are the states of the run before and after the closing step, and that step realises some $(q_1,\sigma,q_2) \in \Delta$ with $(\elw,\elv) \in \sem{\sigma}_{\strC}$, \ie with $\elv$ a $\sigma$-successor of~$\elw$ by~\ref{dict:edge}.

  Conditions~\ref{encready:B0}--\ref{encready:C4} all follow the same pattern: their hypotheses exhibit, through~\ref{dict:edge} and~\ref{dict:path}, a witnessing walk or cycle in~$\strC$, and their conclusions are then forced by the exactness of the demands.
  For~\ref{encready:B0}--\ref{encready:B5} the witnesses are the empty walk, a witnessing cycle read as a walk, the cycles $\elv,\pP{?},\elv$ and $\elv,\sigma,\elv$ for $\sigma \in \{\pe,\inv{\pe}\}$, and the two concatenations of witnesses at~$\elv$, the tests and self-loops involved being available by~\ref{dict:edge} and the second item of Definition~\ref{def:encoding}.
  For~\ref{encready:C1} we splice the two steps supplied by~\ref{dict:edge} around a witnessing cycle for $\autA_{q_2,q_3}$ at~$\elv$; \ref{encready:C2} is the same with a single step and a walk.
  For~\ref{encready:C3} and~\ref{encready:C4}, we observe that \ref{dict:path} turns the two witnessing paths into witnessing walks for~$\hat{\autA}$, and the letter $\pred{bck}_{\sigma}$, respectively $\pred{bck}_{\inv{\sigma}}$, realises the step $(q_2,\sigma,q_3)$ across the back edge by Definition~\ref{def:encoding}; splicing the three pieces gives a witnessing walk for $\hat{\autA}_{q_1,q_4}$ from~$\elv$ to~$\elv$ and hence, as $(\strC,\varepsilon) \models \varphi$, a witnessing cycle for $\autA_{q_1,q_4}$ at~$\elv$.

  Finally, let $\strC$ be regular and assign to every $\elv \in \domC$ the quadruple consisting of the isomorphism type of the subcactus rooted at~$\elv$, the set $\ell(\elv) \cap \Gamma^{1}_{\varphi}$, the direction letter of~$\elv$, and the letter $\pred{bck}_{\sigma}$ carried by the last vertex of the $0$-path issuing from~$\elv$, if any.
  Only finitely many quadruples occur, and by Definition~\ref{def:encoding} a quadruple determines the subtree of~$\treeT$ rooted at~$\elv$: the subcactus at~$\elv$ fixes the label of every vertex below~$\elv$, except for the letters of $\Gamma^{1}_{\varphi}$ and the direction letter at~$\elv$ itself, and for the $\pred{bck}$-letter of the last vertex of the $0$-path from~$\elv$, the only vertex below~$\elv$ whose anchor may lie outside that subcactus.
  The direction letters of the proper descendants of~$\elv$ need no bookkeeping: each of them is determined by the address of that descendant inside the subtree.
  Hence $\treeT$ has finitely many subtrees up to isomorphism.
\end{proof}

\subsection{From Encoding-Ready Trees to Phi-Cacti}\label{appendix:from-encoding-ready-trees-to-phi-cacti}

To prove the converse of Lemma~\ref{lemma:from-cacti-to-encoding-ready-trees}, we introduce a \emph{decoding} map~$\dec$, which intuitively turns encoding-ready trees back into $\varphi$-cacti.
Note that $\dec$ discards the direction letters. It may: by~\ref{encready:A0} they merely spell out the addresses of the vertices, so they are~recoverable~from~$\domT$.

\begin{definition}[Decoding]\label{def:decoding}
  Let $\treeT \deff (\domT, \ell)$ be an encoding-ready tree.
  Its \emph{decoding} $\dec(\treeT)$ is the structure~$\strC$ whose domain $\domC$ consists of the non-$\NULL$ vertices of~$\treeT$.
  For every $\pP \in \Sigma_{\varphi} \cap \AP$ we put $\pP^{\strC} \deff \{ \elv \in \domC \mid \pP \in \ell(\elv) \}$, and for every action $\pe \in \Sigma_{\varphi} \cap \PR$ we let $\pe^{\strC}$ consist of:
  all pairs $(\elv,\elv)$ with $\pred{slf}_{\pe} \in \ell(\elv)$;
  all pairs $(\elw,\elv)$ with $\pe \in \ell(\elv)$ and all pairs $(\elv,\elw)$ with $\inv{\pe} \in \ell(\elv)$, where $\elw$ is the parent of~$\elv$;
  and all pairs $(\elv,\elu)$ with $\pred{bck}_{\pe} \in \ell(\elv)$ and all pairs $(\elu,\elv)$ with $\pred{bck}_{\inv{\pe}} \in \ell(\elv)$, where $\elu$ is the anchor of~$\elv$.
\end{definition}

We next establish several properties of~$\dec$ that will eventually 
lead to the converse of Lemma~\ref{lemma:from-cacti-to-encoding-ready-trees}. 
We first show that the decoding of an encoding-ready tree is cactus-like.

\begin{claim}\label{claim:decoding-gives-unique-cactus-like-structures}
  Let $\treeT$ be an encoding-ready tree. Then $\dec(\treeT)$ is the unique cactus-like structure $\strC$ satisfying $\enc(\strC) = \treeT$ and \ref{phicactus:fini}.
\end{claim}
\begin{proof}
  Let $\strC$ be the decoding of~$\treeT$.
  By~\ref{encready:A2} and~\ref{encready:A3}, the domain $\domC$ is a non-empty prefix-closed subset of~$\domT$, and hence a tree of the arity fixed in~\ref{encready:A0}. By~\ref{encready:A1}, all its $0$-paths are finite and $\varepsilon$ has no $0$-child. By~\ref{encready:A0}, every $\elv \neq \varepsilon$ carries a letter $\sigma \in \Gamma^{1}_{\varphi}$, for which Definition~\ref{def:decoding} puts $(\elw, \elv) \in \sem{\sigma}_{\strC}$, where $\elw$ is the parent of~$\elv$; so $(\elv, \elw) \in \sem{\inv{\sigma}}_{\strC}$ is a $\Sigma_{\varphi}$-edge to the parent. Hence~\ref{phicactus:fini} holds.
  Every pair added by Definition~\ref{def:decoding} is a self-loop, joins two neighbours, or joins an anchor to its anchored descendant. Hence $\strC$ is cactus-like.
  Comparing Definition~\ref{def:decoding} and  Definition~\ref{def:encoding} item by item yields $\enc(\strC) = \treeT$; for the direction letters, which $\dec$ discards, this uses the last part of~\ref{encready:A0}, pinning them in~$\treeT$ to exactly the values that~$\enc$ assigns.
  The same comparison shows that $\dec(\enc(\strC')) = \strC'$ holds for every cactus-like~$\strC'$, whence $\strC$ is unique.
\end{proof}

We now check that $\dec$ preserves regularity.
\begin{claim}\label{claim:decoding-preserves-regularity}
  Let $\treeT$ and $\strC$ be as in 
  Claim~\ref{claim:decoding-gives-unique-cactus-like-structures}, 
  and suppose that $\treeT$ is regular.~Then~so~is~$\strC$.
\end{claim}
\begin{proof}
  By Definition~\ref{def:decoding}, the subcactus of~$\strC$ rooted at a 
  vertex~$\elv$ is determined by the subtree of~$\treeT$ rooted at~$\elv$. 
  Hence $\strC$ has, up to isomorphism, at most as many subcacti as $\treeT$ 
  has subtrees. Since $\treeT$ is regular, the latter number is finite, 
  and hence so is the former.
\end{proof}

We next verify the easier conditions in the definition of $\varphi$-cacti.
\begin{claim}\label{claim:decoding-preserves-easy-properties}
  Let $\treeT$ and $\strC$ be as in Claim~\ref{claim:decoding-gives-unique-cactus-like-structures}.
  Then $\strC$ satisfies~\ref{phicactus:bool}, as well as the ``if'' 
  parts of~\ref{phicactus:walk} and~\ref{phicactus:cycle}.
\end{claim}
\begin{proof}
  Observe that~\ref{dict:dom}--\ref{dict:path} apply to $\strC$ and~$\treeT$.
  Condition~\ref{phicactus:bool} is then precisely~\ref{encready:A3}.
  The ``if'' parts of~\ref{phicactus:walk} and~\ref{phicactus:cycle} follow, 
  by~\ref{dict:path}, from~\ref{encready:A4} and~\ref{encready:A5} respectively: 
  a stationary or anchored witnessing path translates into a witnessing walk 
  of the same shape. For~\ref{encready:A5}, the closing step to the 
  $\sigma$-successor~$\elv$ of~$\elw$, together with the stationary tail, 
  turns that walk into an anchored witnessing cycle through 
  $\elv \cdot \kappa(\ldem_{q,q'})$.
\end{proof}

It remains to establish the ``only if'' directions of~\ref{phicactus:walk}
and~\ref{phicactus:cycle}: a vertex enjoying a witnessing cycle or a path 
must also carry the corresponding demand. We first treat the witnesses that 
move only between equal or neighbouring worlds. Such a witness decomposes 
into ``basic'' pieces that the closure 
Conditions~\ref{encready:B0}--\ref{encready:C2} detect, and an
induction on its length reassembles the resulting labels.
This is where we make essential use of the extended automaton~$\hat{\autA}$,
which, intuitively, can absorb any detour off the path into a single test.
Call a witnessing walk or cycle in a cactus-like structure \emph{tree-like} if
any two consecutive worlds on it are equal or are neighbours.~We~prove:

\begin{claim}\label{claim:tree-like-witnesses-are-labelled}
  Let $\treeT$ and $\strC$ be as in
  Claim~\ref{claim:decoding-gives-unique-cactus-like-structures}, let
  $q, q' \in Q$, and let $\elv \in \domC$.
  If $\elv$ has a tree-like witnessing walk for~$\hat{\autA}_{q,q'}$, then
  $\wdem_{q,q'} \in \ell(\elv)$; if it has a tree-like witnessing cycle
  for~$\hat{\autA}_{q,q'}$, then $\ldem_{q,q'} \in \ell(\elv)$.
\end{claim}
\begin{proof}
  As already discussed, Observations
  \ref{dict:dom}--\ref{dict:path} apply to $\strC$ and~$\treeT$; we use
  \ref{dict:dom} and~\ref{dict:edge} silently, the former also to know that
  every world occurring on a witness is a non-$\NULL$ vertex.
  We argue by induction on the number~$n$ of letters of the witness~$\xi$,
  proving both statements simultaneously for all vertices and all pairs of
  states. One observation recurs: a step of~$\xi$ between two \emph{distinct}
  worlds carries a possibly inverted action, whereas all transitions of
  $\hat{\Delta} \setminus \Delta$ are tests; such a step therefore realises a
  transition of~$\Delta$ itself, as required by
  Conditions~\ref{encready:C1}~and~\ref{encready:C2}.

  \smallskip\noindent\emph{Base case $n = 0$.}
  Then $\xi = \elv$ and $q = q'$, so both statements are~\ref{encready:B0}.

  \smallskip\noindent\emph{Base case $n = 1$.}
  Then $\xi = \elv, \sigma, \elw$ with $(q,\sigma,q') \in \hat{\Delta}$ and
  $(\elv, \elw) \in \sem{\sigma}_{\strC}$.
  \begin{itemize}
  \item Suppose first that $\elv = \elw$, so that $\xi$ is a cycle.
  If $\sigma$ is a test $\pP{?}$ with $(q,\pP{?},q') \in \Delta$, then
  $\pP \in \ell(\elv)$ by Definition~\ref{def:decoding},
  and~\ref{encready:B2} gives $\ldem_{q,q'} \in \ell(\elv)$.
  If $\sigma$ is one of the fresh transitions of~$\hat{\Delta}$, that is
  $\sigma = \ldem_{q,q'}?$, then $\ldem_{q,q'} \in \ell(\elv)$ already by
  Definition~\ref{def:decoding}.
  If $\sigma$ is a possibly inverted action with underlying action~$\pe$, then
  $\pred{slf}_{\pe} \in \ell(\elv)$ by~\ref{dict:edge},
  and~\ref{encready:B3} gives $\ldem_{q,q'} \in \ell(\elv)$.
  In each case \ref{encready:B1} upgrades $\ldem_{q,q'}$ to $\wdem_{q,q'}$,
  which settles~the~walk~statement~as~well.

  \item Suppose now that $\elv \neq \elw$. Then $\xi$ is not a cycle and only the walk
  statement is at stake. Here $\sigma$ is a possibly inverted action with
  $(q,\sigma,q') \in \Delta$, and $\elw$ is a $\sigma$-successor of~$\elv$
  by~\ref{dict:edge}. As $\wdem_{q',q'} \in \ell(\elw)$
  by~\ref{encready:B0}, Condition~\ref{encready:C2} yields
  $\wdem_{q,q'} \in \ell(\elv)$.
  \end{itemize}

  \smallskip\noindent\emph{Inductive step $n \geqslant 2$.}
  Write $\xi = \elv_0, \sigma_1, \elv_1, \ldots, \sigma_n, \elv_n$ with
  $\elv_0 = \elv$, and fix an accepting run
  $(\elv_0, q_0), \ldots, (\elv_n, q_n)$ of $\hat{\autA}_{q,q'}$ along~$\xi$,
  so that $q_0 = q$ and $q_n = q'$. Every infix of~$\xi$ is tree-like,
  and is witnessed by the corresponding infix of that run; 
  we~use~this~without~further~mention.

  \smallskip\noindent\emph{The case of a walk.}
  If $\elv_1 = \elv_0$, then $\elv_0, \sigma_1, \elv_0$ is a witnessing cycle
  for $\hat{\autA}_{q_0,q_1}$ and~$\elv_0$, and
  $\elv_1, \sigma_2, \ldots, \sigma_n, \elv_n$ is a witnessing walk for
  $\hat{\autA}_{q_1,q_n}$ and~$\elv_1 = \elv_0$. Their lengths are
   both below~$n$, so the induction hypothesis gives
  $\ldem_{q_0,q_1} \in \ell(\elv)$ and $\wdem_{q_1,q_n} \in \ell(\elv)$,
  and~\ref{encready:B4} gives $\wdem_{q,q'} \in \ell(\elv)$.
  If $\elv_1 \neq \elv_0$, then $\sigma_1$ is a possibly inverted action with
  $(q_0,\sigma_1,q_1) \in \Delta$, and $\elv_1$ is a $\sigma_1$-successor
  of~$\elv_0$ by~\ref{dict:edge}. The same suffix now gives
  $\wdem_{q_1,q_n} \in \ell(\elv_1)$, and~\ref{encready:C2} gives
  $\wdem_{q,q'} \in \ell(\elv)$.

  \smallskip\noindent\emph{The case of a cycle.} Here $\elv_n = \elv_0 = \elv$.
  If $\elv_1 = \elv_0$, split $\xi$ after its first letter: both
  $\elv_0, \sigma_1, \elv_0$ and $\elv_1, \sigma_2, \ldots, \sigma_n, \elv_n$
  are witnessing cycles at~$\elv$, of lengths below $n$. The induction
  hypothesis gives $\ldem_{q_0,q_1}, \ldem_{q_1,q_n} \in \ell(\elv)$,
  and~\ref{encready:B5} concludes.
  Otherwise $\elu \deff \elv_1$ is a neighbour of~$\elv$. Let $k$ be least with
  $k > 1$ and $\elv_k = \elv$; it exists and satisfies $k \leqslant n$, as
  $\elv_n = \elv$. Observe that $\elv_{k-1} = \elu$, which is a consequence of 
  $\strC$ being cactus-like and $\xi$ being tree-like.
  Consequently $\elv_1, \sigma_2, \ldots, \sigma_{k-1}, \elv_{k-1}$ is a
  witnessing cycle for $\hat{\autA}_{q_1,q_{k-1}}$ and~$\elu$, of length
  smaller than $n$, so $\ldem_{q_1,q_{k-1}} \in \ell(\elu)$ by the induction
  hypothesis.
  The steps $\sigma_1$ and $\sigma_k$ join the distinct worlds $\elv$
  and~$\elu$, hence are possibly inverted actions with
  $(q_0, \sigma_1, q_1), (q_{k-1}, \sigma_k, q_k) \in \Delta$; and
  by~\ref{dict:edge} the vertex $\elu$ is a $\sigma_1$-successor of~$\elv$,
  while $\elv$ is a $\sigma_k$-successor of~$\elu$. Condition~\ref{encready:C1}
  therefore yields $\ldem_{q_0,q_k} \in \ell(\elv)$.
  Finally, $\elv_k, \sigma_{k+1}, \ldots, \sigma_n, \elv_n$ is a witnessing
  cycle at~$\elv$ of length $n{-}k < n$, so $\ldem_{q_k,q_n} \in \ell(\elv)$ by
  the induction hypothesis, and~\ref{encready:B5} gives
  $\ldem_{q,q'} \in \ell(\elv)$.
\end{proof}

\noindent Two remarks. First, since $\Delta \subseteq \hat{\Delta}$, the claim
applies in particular to witnesses for~$\autA_{q,q'}$. Second, tree-likeness is
used only in the case of a cycle: the argument for walks peels off a single
step and never inspects the shape of the remaining witness, so the walk
statement holds for arbitrary witnessing walks. 
Together with Claim~\ref{claim:decoding-preserves-easy-properties} 
we can now conclude:

\begin{corollary}\label{cor:witnessing-walks-works-well-in-decoding}
  Let $\treeT$ and $\strC$ be as in Claim~\ref{claim:decoding-gives-unique-cactus-like-structures}.
  Then $\strC$ satisfies~\ref{phicactus:walk}.
\end{corollary}

Two further kinds of basic cycles are needed. The first travels once around
the cycle that a back edge closes with its anchor path, in either direction,
and interleaves its steps with single tests of cycle demands (these tests
will absorb the detours that an arbitrary cycle makes off the anchor path).
Call a witnessing cycle
$
  \xi \deff \elv_0, \ldem_{p_0,p_0'}?, \elv_0, \sigma_1, \elv_1,
  \ldem_{p_1,p_1'}?, \elv_1, \sigma_2, \elv_2, \ldots,
  \sigma_n, \elv_n, \ldem_{p_n,p_n'}?, \elv_n
$
for some~$\hat{\autA}_{q,q'}$ \emph{cactus-like} if
$\elv_0, \elv_1, \ldots, \elv_{n-1}$ or its reverse 
is a cyclic permutation of an anchored path. 
Relying on Conditions~\ref{encready:C3} and~\ref{encready:C4}, we~prove:

\begin{claim}\label{claim:cactus-like-witnesses-are-labelled}
  Let $\treeT$ and $\strC$ be as in
  Claim~\ref{claim:decoding-gives-unique-cactus-like-structures}, let
  $q, q' \in Q$, and let $\elv \in \domC$.
  If $\elv$ has a cactus-like witnessing cycle for~$\hat{\autA}_{q,q'}$, then
  $\ldem_{q,q'} \in \ell(\elv)$.
\end{claim}
\begin{proof}
  As in the proof of Claim~\ref{claim:tree-like-witnesses-are-labelled},
  Observations~\ref{dict:dom}--\ref{dict:path} apply to $\strC$ and~$\treeT$,
  and we use~\ref{dict:dom} silently.
  Let $\xi$ be as in the definition, assuming $\elv_0 = \elv_n = \elv$, and take
  $\zeta \deff \elu, \elu j, \elu j0, \ldots, \elu j0^{k}$ to be an anchored path
  such that $\elv_0, \ldots, \elv_{n-1}$ is a cyclic permutation of~$\zeta$ or
  of its reverse. Put $\elw \deff \elu j0^{k}$; thus $\elw$ is an anchored
  descendant of its anchor~$\elu$, and $\elv$ lies on the anchored path~$\zeta$.
  Fix an accepting run of $\hat{\autA}_{q,q'}$ along~$\xi$; again, every
  infix of~$\xi$ is a witnessing walk, witnessed by the corresponding infix
  of that run.
  We can assume~$k \geqslant 1$. Indeed, 
  otherwise $\xi$ is tree-like and the claim follows from
  Claim~\ref{claim:tree-like-witnesses-are-labelled}.
  Then $\elu$ and~$\elw$ are not neighbours. Since $\elv_n = \elv_0$, the
  unordered pairs $\{\elv_i, \elv_{i+1}\}$ with $i < n$ are precisely the $k{+}1$
  parent-child pairs of~$\zeta$ together with $\{\elu, \elw\}$, each occurring
  once. Let $i < n$ be the index with $\{\elv_i, \elv_{i+1}\} = \{\elu, \elw\}$,
  write $\xi = \xi_1 \eta \xi_2$ with $\eta \deff \elv_i, \sigma_{i+1}, \elv_{i+1}$,
  and let $(s, \sigma_{i+1}, t) \in \hat{\Delta}$ be the transition that the
  run takes along~$\eta$. As $\elv_i \neq \elv_{i+1}$, the letter
  $\sigma_{i+1}$ is not a test, so it is a possibly inverted action and
  $(s, \sigma_{i+1}, t) \in \Delta$.
  Moreover, $\xi_1$ and~$\xi_2$ are witnessing walks for $\hat{\autA}_{q,s}$
  and~$\hat{\autA}_{t,q'}$ whose non-test steps join a parent with its child.
  Hence neither uses a back edge, and after deleting consecutive duplicates
  their worlds are $\elv_0, \ldots, \elv_i$ and $\elv_{i+1}, \ldots, \elv_n$,
  respectively. By~\ref{dict:path}, these two paths are witnessing paths
  in~$\treeT$ for $\hat{\autA}_{q,s}$ and~$\hat{\autA}_{t,q'}$.
  \begin{itemize}
    \item Suppose $(\elv_i, \elv_{i+1}) = (\elw, \elu)$. By~\ref{dict:edge},
  $\elu$ is a $\sigma_{i+1}$-successor of~$\elw$. As $\elu$ is an ancestor
  of~$\elw$ but not a neighbour of it, only case~(iv) of the definition of
  a $\sigma$-successor applies, whence $\pred{bck}_{\sigma_{i+1}} \in \ell(\elw)$.
  Further, $\elv_0, \ldots, \elv_{n-1}$ is a cyclic permutation of~$\zeta$
  itself, as only there is $\elw$ cyclically followed by~$\elu$. Hence
  $\elv_0, \ldots, \elv_i$ descends~$\zeta$ from~$\elv$ to~$\elw$, and
  $\elv_{i+1}, \ldots, \elv_n$ descends it from~$\elu$ to~$\elv$.
  Condition~\ref{encready:C3}, applied to~$\sigma_{i+1}$ and the 
  states~$q, s, t, q'$, yields~$\ldem_{q,q'} \in \ell(\elv)$.

  \item For $(\elv_i, \elv_{i+1}) = (\elu, \elw)$ the argument is
  symmetric. Now $\elw$ is a $\sigma_{i+1}$-successor of~$\elu$ only by
  case~(iii), so $\pred{bck}_{\inv{\sigma_{i+1}}} \in \ell(\elw)$. The worlds
  of~$\xi$ follow the reverse of~$\zeta$, so the two paths climb from~$\elv$
  to~$\elu$ and from~$\elw$ to~$\elv$. Condition~\ref{encready:C4} then
  yields $\ldem_{q,q'} \in \ell(\elv)$.\qedhere
  \end{itemize}
\end{proof}

It remains to treat the last kind of basic cycles: those that cross a back edge,
make a single stationary step, and return across the same back edge.
Call a witnessing cycle in~$\strC$ a \emph{back-edge} cycle if it is of the form
$\elw, \sigma_1, \elv, \sigma_2, \elv, \sigma_3, \elw$, where $\elw$ is an
anchored descendant of its anchor~$\elv$, or vice versa.
Such cycles are detected by Condition~\ref{encready:C1}, as we now show.

\begin{claim}\label{claim:back-edge-witnesses-are-labelled}
  Let $\treeT$ and $\strC$ be as in
  Claim~\ref{claim:decoding-gives-unique-cactus-like-structures}, let
  $q, q' \in Q$, and let $\elw \in \domC$.
  If $\elw$ has a back-edge witnessing cycle for~$\hat{\autA}_{q,q'}$, then
  $\ldem_{q,q'} \in \ell(\elw)$.
\end{claim}
\begin{proof}
  Without any surprise now, 
  let $\xi = \elw, \sigma_1, \elv, \sigma_2, \elv, \sigma_3, \elw$ be such a
  cycle, and fix an accepting run
  $(\elw, q), (\elv, q_1), (\elv, q_2), (\elw, q')$ of $\hat{\autA}_{q,q'}$
  along~$\xi$.
  Its infix $\elv, \sigma_2, \elv$ is a witnessing cycle for
  $\hat{\autA}_{q_1,q_2}$ at~$\elv$, and it is trivially tree-like; hence
  $\ldem_{q_1,q_2} \in \ell(\elv)$ by
  Claim~\ref{claim:tree-like-witnesses-are-labelled}.
  The steps $\sigma_1$ and~$\sigma_3$ join the distinct worlds $\elw$
  and~$\elv$. As in the proof of that claim, they are therefore possibly
  inverted actions with $(q, \sigma_1, q_1), (q_2, \sigma_3, q') \in \Delta$.
  Moreover, by~\ref{dict:edge}, the vertex $\elv$ is a $\sigma_1$-successor
  of~$\elw$, while $\elw$ is a $\sigma_3$-successor of~$\elv$.
  Condition~\ref{encready:C1}~yields~$\ldem_{q,q'} \in \ell(\elw)$.
\end{proof}

\noindent Call a cycle \emph{basic} if it is tree-like, cactus-like, or a
back-edge cycle, and call a cycle an \emph{$\ldem$-cycle} if it is of the form
$\elv, \ldem_{q,q'}?, \elv$ for some $q, q' \in Q$. The following claim shows
that every witnessing cycle can be replaced by a witnessing $\ldem$-cycle.
The three preceding claims already establish this for basic cycles; we extend
it to all witnessing cycles by induction, through a careful analysis of their
structure that allows us to decompose them into simpler ones.

The following observation seems to be a standard property of cacti, but we were unable to find a reference, so we include a proof. It will be useful for the inductive step.

\begin{observation}\label{obs:simple-cycles-are-anchored}
  Let $\strC$ be a cactus-like structure, let $k \geqslant 3$, and let
  $\elv_0, \ldots, \elv_{k-1} \in \domC$ be pairwise distinct vertices such
  that, for every $\ell < k$, the vertices $\elv_\ell$ and
  $\elv_{(\ell+1) \bmod k}$ are neighbours or joined by a back edge. Then
  $\elv_0, \ldots, \elv_{k-1}$ or its reverse is a cyclic permutation of an
  anchored path.
\end{observation}
\begin{proof}
  Let $G$ be the Gaifman graph of~$\strC$, \ie the undirected simple graph
  with vertex set~$\domC$ whose edges join neighbours (\emph{tree edges}) or
  the two ends of a back edge (again \emph{back edges}). The tree edges span a
  tree.
  Put $\elv_k \deff \elv_0$. The edges $(\elv_\ell, \elv_{\ell+1})$ with
  $\ell < k$ form a cycle~$D$ in~$G$. No $e \in D$ is simultaneously a tree
  edge and a back edge: such an $e$ is the unique edge of~$G$ crossing the cut
  it induces, hence lies on no cycle. Let $B \subseteq D$ be its back edges;
  $B \neq \emptyset$, as trees are acyclic. For $e \in B$, let $P_e$ be the
  edges of the anchored path that~$e$ closes into a cycle; thus
  $e \notin P_e$. The sets~$P_e$ are pairwise disjoint: the lower end of such
  an edge has the form $\elu j 0^i$ with $j \neq 0$, which determines the
  anchor~$\elu$ and the direction~$j$, hence~$e$, as $\elu$ has at most one
  anchored descendant in direction~$j$. Replacing each step of
  $\elv_0, \ldots, \elv_k$ across some $e \in B$ by the corresponding anchored
  path yields a closed walk at~$\elv_0$ in the tree. It traverses each tree
  edge an even number of times, as removing one disconnects the tree. By
  disjointness, each $f \in P_e$ is traversed once inside the inserted paths,
  and once more exactly when $f \in D$; hence $f \in D$. Thus $D$ contains,
  and so equals, the cycle $\{e\} \cup P_e$ for each $e \in B$. Consequently,
  $\elv_0, \ldots, \elv_{k-1}$ or its reverse is a cyclic permutation of that
  anchored path.
\end{proof}

We finally prove that every witnessing cycle is properly detected by $\ldem$ labels.

\begin{claim}\label{claim:all-witnesses-are-labelled}
  Let $\treeT$ and $\strC$ be as in
  Claim~\ref{claim:decoding-gives-unique-cactus-like-structures}, let
  $q, q' \in Q$, and let $\elw \in \domC$.
  If $\elw$ has a witnessing cycle for~$\hat{\autA}_{q,q'}$,
  then $\ldem_{q,q'} \in \ell(\elw)$.
\end{claim}
\begin{proof}
  We proceed by strong induction on the number~$n$ of letters of the witnessing
  cycle, simultaneously for all vertices and all pairs of states.
  Let $\xi = \elv_0, \sigma_1, \elv_1, \ldots, \sigma_n, \elv_n$ be a witnessing
  cycle for~$\hat{\autA}_{q,q'}$ with $\elv_0 = \elv_n = \elw$, and fix an
  accepting run $(\elv_0, q_0), \ldots, (\elv_n, q_n)$ of~$\hat{\autA}_{q,q'}$
  along~$\xi$. If $\xi$ is tree-like, as is the case whenever $n \leqslant 1$,
  Claim~\ref{claim:tree-like-witnesses-are-labelled} applies. Otherwise, some
  step of~$\xi$ joins two non-neighbours; in particular, $\xi$ is not stationary.

  \smallskip\noindent\emph{Shortcuts.}
  For $i \leqslant j$ with $\elv_i = \elv_j$, call the replacement of the infix
  $\elv_i, \sigma_{i+1}, \ldots, \sigma_j, \elv_j$ by
  $\elv_i, \ldem_{q_i,q_j}?, \elv_i$ \emph{safe} if
  $\ldem_{q_i,q_j} \in \ell(\elv_i)$. A safe replacement again yields a
  witnessing cycle for~$\hat{\autA}_{q,q'}$ at~$\elw$, as
  $(q_i, \ldem_{q_i,q_j}?, q_j) \in \hat{\Delta}$ and $\elv_i$ satisfies
  $\ldem_{q_i,q_j}$ by Definition~\ref{def:decoding}.
  If $j - i \geqslant 2$ and $(i,j) \neq (0,n)$, then $j - i < n$, so the
  replacement is safe by the induction hypothesis and produces a shorter
  witnessing cycle, to which the induction hypothesis applies once more.
  We may thus assume that no such $i$ and~$j$ exist.
  Let $\elu_0, \ldots, \elu_k$ be the reduction of~$\xi$, \ie its sequence of
  worlds without consecutive duplicates. Then $\elu_0 = \elu_k = \elw$,
  $k \geqslant 2$, and $\elu_0, \ldots, \elu_{k-1}$ are pairwise distinct, since
  $\elu_a = \elu_b$ with $a < b < k$ would yield $i < j < n$ with
  $\elv_i = \elv_j$ and $j - i \geqslant 2$. Every maximal block of equal
  consecutive worlds of~$\xi$ is a stationary, hence tree-like, witnessing
  cycle; replacing all of them, which is safe by
  Claim~\ref{claim:tree-like-witnesses-are-labelled}, yields~a~witnessing~cycle
  \[
    \hat{\xi} \deff \elu_0, \ldem_{p_0,p_0'}?, \elu_0, \tau_1, \elu_1,
    \ldem_{p_1,p_1'}?, \elu_1, \ldots, \tau_k, \elu_k, \ldem_{p_k,p_k'}?, \elu_k
  \]
  for~$\hat{\autA}_{q,q'}$ at~$\elw$, where $p_0 = q$, $p_k' = q'$, and
  $\ldem_{p_a,p_a'} \in \ell(\elu_a)$ for all~$a \leqslant k$. As $\strC$ is
  cactus-like, any two consecutive worlds $\elu_a, \elu_{a+1}$ are neighbours or
  an anchor and its anchored descendant.

  \smallskip\noindent\emph{Case $k = 2$.}
  Every non-test step of~$\xi$ joins $\elw$ and~$\elu_1$, so these are not
  neighbours, and one is an anchored descendant of the other. Hence
  $\elw, \tau_1, \elu_1, \ldem_{p_1,p_1'}?, \elu_1, \tau_2, \elw$ is a back-edge
  witnessing cycle for~$\hat{\autA}_{p_0',p_2}$, and
  Claim~\ref{claim:back-edge-witnesses-are-labelled} yields
  $\ldem_{p_0',p_2} \in \ell(\elw)$. As also
  $\ldem_{q,p_0'}, \ldem_{p_2,q'} \in \ell(\elw)$, two applications
  of~\ref{encready:B5} give $\ldem_{q,q'} \in \ell(\elw)$.

  \smallskip\noindent\emph{Case $k \geqslant 3$.}
  By Observation~\ref{obs:simple-cycles-are-anchored}, applied to the tree~$\domC$
  and $\elu_0, \ldots, \elu_{k-1}$, the cycle~$\hat{\xi}$ is cactus-like, and
  Claim~\ref{claim:cactus-like-witnesses-are-labelled} concludes.
\end{proof}

Combining the properties established so far, we obtain the following lemma.

\begin{lemma}\label{lemma:from-encoding-ready-trees-to-phi-cacti}
  Let $\treeT$ be an encoding-ready tree.
  Then $\dec(\treeT)$ is an exact $\varphi$-cactus.
  Moreover, if $\treeT$ is regular, then so is $\dec(\treeT)$.
\end{lemma}
\begin{proof}
  That $\dec(\treeT)$ is a $\varphi$-cactus follows directly from
  Claims~\ref{claim:decoding-gives-unique-cactus-like-structures},
  \ref{claim:decoding-preserves-easy-properties},
  and~\ref{claim:all-witnesses-are-labelled}, together with
  Corollary~\ref{cor:witnessing-walks-works-well-in-decoding}.
  Regularity is preserved by
  Claim~\ref{claim:decoding-preserves-regularity}.
\end{proof}

\subsection{Two-Way Alternating Tree Automata}\label{appendix:twoapta}

For the reader's convenience, we collect here the automata-theoretic machinery 
that will be employed in the proof of Lemma~\ref{lemma:cactus-automaton}.
For the presentation we closely follow 
Calvanese, Eiter and Ortiz~\cite[Section~2]{CalvaneseEO14}, with two cosmetic deviations.
First, we run automata over trees as defined in Section~\ref{sec:preliminaries}, 
rather than over the trees of Calvanese et al.~\cite[Definition~2.6]{CalvaneseEO14}. 
Since our trees have their children numbered from~$0$ (and not from~$1$), 
we write $\stay$ and $\up$ for the two special directions 
denoted there by $0$ and~$-1$. Second, we speak of \emph{vertices} rather 
than of nodes, and reserve $Q$ for the states of the semiautomaton~$\autA$ 
fixed in the Proviso, using~$S$ for the states of tree automata.

\noindent \textbf{Positive Boolean formul{\ae}.}\,
Given a finite set~$W$, we write $\posbool{W}$ for the set of 
\emph{positive Boolean formul{\ae}} over~$W$, \ie the formul{\ae} built from 
the elements of~$W$ and the constants $\top, \bot$ using $\land$ and~$\lor$ only.
A~set $V \subseteq W$ \emph{satisfies} a formula $\beta \in \posbool{W}$ if 
assigning $\top$ to every element of~$V$ and $\bot$ to every element 
of $W \setminus V$ makes~$\beta$ evaluate to true.

\noindent \textbf{Directions.}\,
Fix $\rmN \in \N$ and put 
$\dirs \; \deff \; \{\, \up, \, \stay \,\} \; \cup \; \{\, 0, 1, \ldots, \rmN \,\}$,
the set of \emph{directions}; note that $|\dirs| = \rmN{+}3$. 
Let $\domT \subseteq \{0,1,\dots,\rmN\}^{*}$ be a tree, let $\elv \in \domT$ 
be a vertex, and let $d \in \dirs$. The vertex $\elv{\cdot}d$ 
\emph{reached from~$\elv$ in direction~$d$} is defined as follows:
\begin{itemize}\itemsep0em
  \item $\elv{\cdot}\stay \deff \elv$ (the automaton idles at~$\elv$);
  \item $\elv{\cdot}\up \deff \elw$, where $\elw$ is the parent of~$\elv$; 
  this is undefined for $\elv = \varepsilon$;
  \item $\elv{\cdot}i \deff \elv i$ for $0 \leqslant i \leqslant \rmN$, 
  the $i$-th child of~$\elv$; this is undefined if $\elv i \notin \domT$. 
  This corner case will not trouble us: by~\ref{run:R2}, a transition that 
  demands a copy in a direction undefined at the current vertex cannot be 
  satisfied, and the automata we construct use such transitions to force 
  their input trees to be complete (see Observation~\ref{obs:broadcast}).
\end{itemize}
Thus $\elv{\cdot}d$ is a vertex of~$\domT$ whenever it is defined, and it is either $\elv$ itself or a neighbour~of~$\elv$.

\noindent \textbf{The automaton model.}\,
The definition below is almost a verbatim copy of Definition~2.7 of Calvanese 
et al.~\cite{CalvaneseEO14}, modulo minor, already mentioned, cosmetic changes.

\begin{definition}[2APTA]\label{def:2apta}
  A~\emph{two-way alternating parity tree automaton} (2APTA) is a tuple 
  $\autT \deff (\Sigma, \rmN, S, \delta, s_0, F)$, where $\Sigma$ 
  is a finite \emph{input alphabet}, $\rmN \in \N$ fixes the direction 
  set~$\dirs$, $S$~is a finite set of \emph{states}, 
  $\delta \colon S \times \Sigma \to \posbool{\dirs \times S}$ is the 
  \emph{transition function}, $s_0 \in S$ is the \emph{initial state}, 
  and $F \deff (G_1, G_2, \ldots, G_n)$ with 
  $G_1 \subseteq G_2 \subseteq \cdots \subseteq G_n = S$ is the 
  \emph{parity acceptance condition}. The number~$n$ is the \emph{index} 
  of~$\autT$, denoted~$\ind(\autT)$. We~write $S(\autT), \Sigma(\autT)$, 
  etc., for the components of~$\autT$ if the automaton is not clear 
  from the context, and put $G_0 \deff \emptyset$.
\end{definition}

The automaton runs over $\Sigma$-labelled trees $(\domT, \ell)$ 
with $\domT \subseteq \{0,1,\dots,\rmN\}^{*}$, and it is \emph{two-way} 
and \emph{alternating}: a copy of~$\autT$ that sits at a vertex~$\elv$ in a 
state~$s$ reads the label $\ell(\elv)$ and spawns finitely many new copies, 
each of which is sent to a vertex reachable from~$\elv$ in a single move.
This is what the transition function encodes.
An atom $(d, s') \in \dirs \times S$ occurring in $\delta(s, \sigma)$ is to be 
read as the instruction 
\emph{``send a copy of~$\autT$ in state~$s'$ to the vertex $\elv{\cdot}d$''};
for $d = \stay$ the copy idles and re-reads the current vertex, for $d = \up$ 
it moves to the parent of~$\elv$ (which is what makes the automaton two-way), 
and for $d = i$ with $0 \leqslant i \leqslant \rmN$ it moves to the $i$-th 
child of~$\elv$. Disjunctions in $\delta(s,\sigma)$ are resolved existentially 
and conjunctions universally, the constant~$\top$ makes the current copy accept 
immediately (it may spawn no successors at all), and~$\bot$ makes it reject.
Note that a copy may only be sent in a direction~$d$ with $\elv{\cdot}d$ 
defined; in particular, a copy sitting at the root may not move up.

\noindent \textbf{Runs.}\,
A run records, in a tree-shaped fashion, the whole history of the copies 
spawned by~$\autT$: each of its vertices is labelled by a \emph{configuration}, 
\ie by a pair consisting of the vertex of the input tree that the corresponding 
copy reads and of the state it is in.
The definition below is an almost verbatim reproduction of Definition~2.8 
of Calvanese et al.~\cite{CalvaneseEO14}.

\begin{definition}[Run]\label{def:2apta-run}
  Let $\autT \deff (\Sigma, \rmN, S, \delta, s_0, F)$ be a 2APTA, 
  let $(\domT, \ell)$ be a $\Sigma$-labelled tree with 
  $\domT \subseteq \{0,1,\dots,\rmN\}^{*}$, let $\elv \in \domT$, 
  and let $s \in S$. An \emph{$(\elv, s)$-run} of~$\autT$ over $(\domT, \ell)$ 
  is a $(\domT {\times} S)$-labelled tree $(\domR, r)$ such that:
  \begin{namedlist}[$(\mathrm{R}_1)$\hspace{0.4em}]
    \item[\desclabel{$(\mathrm{R}_1)$}{run:R1}]
    $\varepsilon \in \domR$ and $r(\varepsilon) = (\elv, s)$;

    \item[\desclabel{$(\mathrm{R}_2)$}{run:R2}]
    every $\elw \in \domR$ \emph{satisfies}~$\delta$, that is: 
    if $r(\elw) = (\ely, s')$ and $\delta(s', \ell(\ely)) = \beta$, then there 
    are $m \in \N$ and a set 
    $W \deff \{ (d_0, t_0), \ldots, (d_{m-1}, t_{m-1}) \} \subseteq \dirs \times S$, 
    empty in case $m = 0$, such that $W$ satisfies~$\beta$ and, 
    for every $i < m$, the vertex $\ely{\cdot}d_i$ is defined, 
    $\elw i \in \domR$, and $r(\elw i) = (\ely{\cdot}d_i, t_i)$.
  \end{namedlist}
  For an infinite downward path~$\xi$ in~$\domR$, we let $\Inf(\xi)$ be the set 
  of those states $s' \in S$ for which infinitely many vertices~$\elw$ on~$\xi$ 
  satisfy $r(\elw) \in \domT {\times} \{s'\}$, \ie the states that occur 
  infinitely often along~$\xi$. Such a path~$\xi$ \emph{satisfies}~$F$ if there 
  is an even $i$ with $\Inf(\xi) \cap G_{i} \neq \emptyset$ 
  and $\Inf(\xi) \cap G_{i-1} = \emptyset$. The run $(\domR, r)$ is 
  \emph{accepting} if every infinite downward path of~$\domR$~satisfies~$F$.

  \noindent We say that $\autT$ \emph{accepts} a $\Sigma$-labelled tree 
  $(\domT, \ell)$ if it has an accepting $(\varepsilon, s_0)$-run over 
  $(\domT, \ell)$, and we write $\Lang(\autT)$ for the set of all 
  $\Sigma$-labelled trees accepted by~$\autT$, the \emph{language} of~$\autT$. 
  The \emph{non-emptiness problem} asks, given a 2APTA~$\autT$, whether 
  $\Lang(\autT) \neq \emptyset$. As an important consequence of Rabin's basis 
  theorem~\cite[Thm.~3.7]{PerrinPin04InfiniteTrees}, 
  if $\Lang(\autT) \neq \emptyset$, then $\autT$ accepts a~\emph{regular}~tree.
\end{definition}

\noindent A~few comments on Definition~\ref{def:2apta-run} are in order.
Note that the run $(\domR, r)$ is a $(|\dirs| \cdot |S|)$-ary tree.
Its vertices are indexed independently of those of~$\domT$; the map~$r$ is not 
required to be injective, and several copies of~$\autT$ may well read the same 
vertex of~$\domT$ in the same state. Note also that $\domR$ need not be 
complete, and that a copy whose transition is satisfied by 
$W \deff \emptyset$ (which happens precisely when $\emptyset$ satisfies 
$\delta(s', \ell(\ely))$, \eg when the latter equals~$\top$) contributes 
a leaf; finite branches of~$\domR$ are therefore unconstrained by~$F$, 
and only the infinite ones must satisfy the parity condition.
Finally, since $G_1 \subseteq \cdots \subseteq G_n = S$, we may equivalently 
assign to each state~$s'$ its \emph{priority}, namely the least~$i$ with 
$s' \in G_i$; as $S$ is finite and $\xi$ is infinite, $\Inf(\xi)$ is non-empty, 
and $\xi$ satisfies~$F$ if and only if the least priority occurring 
infinitely~often~along~$\xi$~is~even.

\noindent \textbf{Complexity and closure properties.}\,
The two facts about 2APTA that we rely upon are the following.
The first one, going back to Vardi~\cite{Vardi98}, is the cornerstone of the 
automata-theoretic approach to satisfiability: it is what turns our 
characterisation of satisfiability by means of $\varphi$-cacti 
into a decision procedure.

\begin{proposition}[Calvanese et al.~{\cite[Theorem~2.9]{CalvaneseEO14}}, 
  after Vardi~\cite{Vardi98}]\label{prop:2apta-nonemptiness}
  Non-emptiness of a given 2APTA~$\autT$ is decidable in time 
  single exponential in~$|S(\autT)|$ and polynomial in~$|\Sigma(\autT)|$.
\end{proposition}

\noindent This is precisely the shape of the bound that we need in 
Lemma~\ref{lemma:cactus-automaton}: the automaton~$\autT_{\varphi}$ that is 
going to be constructed there has only polynomially many states in~$|\varphi|$,
while its alphabet $2^{\Gamma_{\varphi}}$ 
(and hence also its transition function) is of exponential size. 
Both contributions are exponential in~$|\varphi|$, 
and so is the total running time.
The second fact lets us assemble~$\autT_{\varphi}$ from a family of small
automata, one for each condition of Definition~\ref{def:encoding-ready-trees},
rather than build a single monolithic automaton. This way, we can treat each
condition in isolation and then combine the resulting automata into one.

\begin{proposition}[Calvanese 
  et al.~{\cite[Proposition~2.10]{CalvaneseEO14}}]\label{prop:2apta-intersection}
  Let $\autT_1, \ldots, \autT_n$ be 2APTA sharing the same input 
  alphabet~$\Sigma$ and the same~$\rmN$. Then one can construct, in time 
  polynomial in $\sum_{i=1}^{n} |\autT_i|$, a 2APTA~$\autT$ with 
  $|S(\autT)| = \sum_{i=1}^{n} |S(\autT_i)| + 1$ and 
  $\ind(\autT) = \max_{i=1}^{n} \ind(\autT_i)$ such that 
  $\Lang(\autT) = \bigcap_{i=1}^{n} \Lang(\autT_i)$.
\end{proposition}

\newcommand{\AlphabetAPTA}{\Gamma^{\mathrm{corr}}_{\varphi}}
\newcommand{\everywhere}[1]{\square\,#1}
\newcommand{\gate}[2]{\big[#1\big]_{#2}}
\newcommand{\nlab}[1]{\overline{s}_{#1}}
\newcommand{\dtest}[2]{s^{\Downarrow}_{#1,#2}}
\newcommand{\clmb}[1]{s^{\Uparrow}_{#1}}
\newcommand{\thst}[1]{s^{\Theta}_{#1}}
\newcommand{\suc}[2]{\mathrm{succ}_{#1}(#2)}
\newcommand{\dwnst}[3]{s^{\downarrow}_{#1,#2,#3}}
\newcommand{\stpst}[3]{s^{\dashv}_{#1,#2,#3}}
\newcommand{\upst}[2]{s^{\uparrow}_{#1,#2}}
\newcommand{\statst}[2]{s^{\circlearrowleft}_{#1,#2}}
\newcommand{\jump}[1]{\mathrm{jump}^{#1}}
\newcommand{\move}{\mathrm{move}}
\newcommand{\viol}{\mathrm{viol}}
\newcommand{\stat}[1]{\rightsquigarrow^{#1}}

\subsection{A Toolbox for Two-Way Tree Automata}\label{app:toolbox}

We now construct the 2APTA promised in Lemmas~\ref{lemma:cactus-automaton} and~\ref{lemma:foldable-cactus-automaton}.
Writing tree automata is, in spirit, programming in Assembler; this is why almost no paper on extensions of $\PDL$ contains an actual construction, usually stating merely that ``it can be done'' or that ``the construction is routine''.
We are more pedantic, and buy back readability by programming with \emph{subroutines}: this section fixes an input alphabet and assembles a small library of gadgets, from which the next three sections build the automata we need.
The library has three shelves: generic constructions turning automata into automata (Observations~\ref{obs:broadcast} and~\ref{obs:duality}); \emph{navigation} gadgets, which walk a copy to a prescribed relative of a vertex; and \emph{simulation} gadgets, which move a copy along a route while simulating a semiautomaton on the labels~it~reads.

\subsubsection{The Alphabet and Our Conventions}

Several properties of encoding-ready trees merely constrain how a single vertex is labelled. We implement such constraints at the level of the alphabet, without constructing any automaton at all. This concerns the conditions \ref{encready:A3}, \ref{encready:B0}, \ref{encready:B1}, \ref{encready:B2}, \ref{encready:B3}, \ref{encready:B4}, \ref{encready:B5}, and, from \ref{encready:A0}, the facts that the only subset containing $\NULL$ is $\{\NULL\}$, that every subset without $\NULL$ contains exactly one direction letter, and that a subset contains a letter from $\Gamma^{1}_{\varphi}$ if and only if it contains a direction letter other than~$\pD_{\varepsilon}$.
Iterating naively over all subsets of $\Gamma_{\varphi}$ and brute-force checking the constraints, we obtain the alphabet $\AlphabetAPTA \subseteq 2^{\Gamma_{\varphi}}$ by filtering out from $2^{\Gamma_{\varphi}}$ every symbol that violates one of the listed conditions.
This takes time exponential in $|\varphi|$, and $|\AlphabetAPTA|$, being a subset of $2^{\Gamma_{\varphi}}$, is also at most exponential in~$|\varphi|$.
From now on, \emph{all} 2APTA under construction work on $\AlphabetAPTA$-labelled trees of arity $\rmN{+}1$ for $\rmN \deff 2|Q|^2$.

\noindent \textbf{Conventions.}\,
Recall that we put $\Gamma_{\varphi}^{0} \deff (\Sigma_{\varphi} \cap \AP) \cup \{ \pred{slf}_{\pe} \mid \pe \in \Sigma_{\varphi} \cap \PR \}$ and $\Gamma_{\varphi}^{1} \deff \{ \pe, \inv{\pe} \mid \pe \in \Sigma_{\varphi} \cap \PR \}$.
Throughout, the letter~$i$ ranges over $\{0, 1, \ldots, \rmN\}$, the letter~$j$ over $\{1, \ldots, \rmN\}$, and the letter~$a$ over $\AlphabetAPTA$.
The set $\Gamma^{\mathrm{bck}}_{\varphi} \deff \{\, \pred{bck}_{\sigma} \mid \sigma \in \Sigma_{\varphi} \,\}$ collects the back-edge letters.
We abbreviate $\everywhere{s} \deff \bigwedge_{i=0}^{\rmN} (i, s)$, \emph{``send a copy in state~$s$ to every child''}.
We write $\gate{\beta}{P}$ for~$\beta$ when the side condition~$P$ holds, and~$\bot$ otherwise.
Empty disjunctions are~$\bot$ and empty conjunctions~$\top$.
By~\ref{run:R2}, every direction occurring in a transition formula must be defined at the current vertex.
This is how the gadget of Observation~\ref{obs:broadcast} forces its input trees to be complete, at no cost.
Finally, we name the property in terms of which~all~our~gadgets~are~specified.

\begin{definition}[Realisability]\label{def:realisable}
  Let $\autT$ be a 2APTA with state set~$S$, let $(\domT, \ell)$ be a labelled tree, let $\elv \in \domT$, and let $\beta \in \posbool{\dirs \times S}$.
  We say that $\beta$ is \emph{realisable at~$\elv$} if some set $W \subseteq \dirs {\times} S$ satisfying~$\beta$ is such that, for every $(d, t) \in W$, the vertex $\elv {\cdot} d$ is defined and $\autT$ has an accepting $(\elv {\cdot} d, t)$-run.
\end{definition}

\noindent By~\ref{run:R1} and~\ref{run:R2}, $\autT$ has an accepting $(\elv, s)$-run if and only if $\delta(s, \ell(\elv))$ is realisable at~$\elv$; we use this silently.
We say that $s$ \emph{succeeds at~$\elv$} if $\autT$ has an accepting $(\elv, s)$-run.

\subsubsection{Generic Constructions: Universality and Complementation}

Besides the closure of 2APTA under intersection obtained via 
Proposition~\ref{prop:2apta-intersection}, we employ two standard transformations.
The first spares us the bookkeeping for conditions required of
\emph{every} vertex: it runs a given automaton at every vertex at once.

\begin{observation}[Broadcast]\label{obs:broadcast}
  Let $\autT' \deff (\AlphabetAPTA, \rmN, S', \delta', s', (G'_1, \ldots, G'_n))$ be a 2APTA, and let $\autT \deff (\AlphabetAPTA, \rmN, S' \cupdot \{s_{\forall}\}, \delta, s_{\forall}, F)$, where $\delta$ extends~$\delta'$ by $\delta(s_{\forall}, a) \deff (\stay, s') \land \everywhere{s_{\forall}}$ and $F$ arises from $(G'_1, \ldots, G'_n)$ by putting~$s_{\forall}$ at the least \emph{even} priority not below~$n$, that is, $F \deff (G'_1, \ldots, G'_{n-1}, S' \cupdot \{s_{\forall}\})$ for even~$n$ and $F \deff (G'_1, \ldots, G'_{n}, S' \cupdot \{s_{\forall}\})$ for odd~$n$.
  Then $\autT$ accepts a tree $(\domT, \ell)$ if and only if $\domT$ is the complete $(\rmN{+}1)$-ary tree and $s'$ succeeds at every $\elv \in \domT$.
  Moreover, $\ind(\autT) \leqslant \ind(\autT') + 1$.
\end{observation}
\begin{proof}
  Since $\delta(s_{\forall}, a)$ demands a copy in every direction $i \leqslant \rmN$, \ref{run:R2} forces every vertex read in state~$s_{\forall}$ to have all $\rmN{+}1$ children; such copies reach every vertex, so $\domT$ is complete.
  By~\ref{run:R1} and~\ref{run:R2}, an $(\varepsilon, s_{\forall})$-run of~$\autT$ consists of a skeleton of $s_{\forall}$-labelled vertices covering~$\domT$, carrying at the vertex reading~$\elv$ an $(\elv, s')$-run of~$\autT'$ attached by the $\stay$-move; conversely, any such family of runs combines into an $(\varepsilon, s_{\forall})$-run.
  An infinite path~$\xi$ of the glued run either enters an attached run, and is then a suffix of an infinite path of it, or stays in the skeleton forever, where $\Inf(\xi) = \{s_{\forall}\}$; the latter satisfies~$F$ as $s_{\forall}$ has an even priority, and for the former $F$ and $(G'_1, \ldots, G'_n)$ agree, both conditions having the same restriction to~$S'$.
  Hence the constructed run is accepting precisely when all the attached ones are.
\end{proof}

\noindent The second transformation complements an automaton.
It lets us state every condition below in its \emph{positive} form---as the existence of a violation, of a successor, of a path---and let the automaton guess; complementation then turns the guesses into the universal quantifiers the conditions really carry.
Its only proviso is that a copy must never be asked to move up from the root: call a 2APTA \emph{root-safe} if no atom $(\up, \cdot)$ occurs in $\delta(s, a)$ for a letter~$a$ with $\pD_{\varepsilon} \in a$.

\begin{observation}[Duality]\label{obs:duality}
  Let $\autT \deff (\AlphabetAPTA, \rmN, S, \delta, s_0, F)$ be a root-safe 2APTA, let $\overline{\delta}(s, a)$ be the De Morgan dual of $\delta(s, a)$, that is, the formula obtained by exchanging $\land$ with~$\lor$ and $\top$ with~$\bot$ while keeping all atoms, and let $\overline{F}$ be a parity condition over~$S$ accepting exactly the paths that $F$ rejects.
  Then, over complete $(\rmN{+}1)$-ary $\AlphabetAPTA$-labelled trees, $\overline{\autT} \deff (\AlphabetAPTA, \rmN, S, \overline{\delta}, s_0, \overline{F})$ has an accepting $(\elv, s)$-run if and only if $\autT$ has none.
  Such an $\overline{F}$ always exists with $\ind(\overline{F}) \leqslant \ind(F) + 1$, namely $\overline{F} \deff (\emptyset, G_1, \ldots, G_n)$.
\end{observation}
\begin{proof}
  Fix a complete tree $(\domT, \ell)$ and consider the \emph{acceptance game} $\mathcal{G}(\autT)$~\cite[Chapter 9]{GraedelTW02} played by two players: Pathfinder and Automaton. Its positions are the pairs $(\elv, \beta)$ with $\elv \in \domT$ and $\beta$ a subformula of some $\delta(s, \ell(\elv))$; at $(\elv, \beta_1 \lor \beta_2)$ Automaton moves to one of $(\elv, \beta_k)$ and at $(\elv, \beta_1 \land \beta_2)$ Pathfinder does; at $(\elv, (d, t))$ the play continues at $(\elv {\cdot} d, \delta(t, \ell(\elv {\cdot} d)))$; a play ending at $\top$, or in which Pathfinder cannot move, is won by Automaton, and dually for~$\bot$.
  Every direction ever met is defined---the directions $0, \ldots, \rmN$ by completeness and $\up$ by root-safety---so the play is never stuck at an atom.
  Unfolding~\ref{run:R1} and~\ref{run:R2}, the strategies of Automaton from $(\elv, \delta(s, \ell(\elv)))$ are exactly the $(\elv, s)$-runs of~$\autT$, and a play conforming to one traces a downward path of that run; hence Automaton wins from that position if and only if $s$ succeeds at~$\elv$, the winning condition being $F$ read on the $S$-components.
  The game $\mathcal{G}(\overline{\autT})$ has the same positions and moves, with the two players exchanged and the winning condition complemented.
  Parity games are determined~\cite[Chapter~6]{GraedelTW02}, so exactly one of the two players wins from a given position, which is the claim.
  Finally, a path satisfies $(\emptyset, G_1, \ldots, G_n)$ if and only if it violates $(G_1, \ldots, G_n)$, as the priorities of all states are shifted by one. 
\end{proof}

\subsubsection{Basic Shape Trees}\label{app:basic-shape}

All gadgets below are specified relative to trees of an already known shape, which we identify first.
Call an $\AlphabetAPTA$-labelled tree a \emph{basic shape tree} if it satisfies every condition of Definition~\ref{def:encoding-ready-trees} except \ref{encready:A4}, \ref{encready:A5} and \ref{encready:C1}--\ref{encready:C4}.
Filtering the alphabet (\ie switching from $2^{\Gamma_{\varphi}}$ to $\AlphabetAPTA$) has secured \ref{encready:A3}, \ref{encready:B0}--\ref{encready:B5} and every label-local part of~\ref{encready:A0}.
It remains to require that the accepted tree $(\domT, \ell)$ satisfy the following, for every vertex~$\elv$ and every direction~$i$:
\begin{namedlist}[$(\mathrm{S}_7)$\hspace{0.4em}]
  \item[\desclabel{$(\mathrm{S}_1)$}{bst:S1}]
  $\domT$ is the complete $(\rmN{+}1)$-ary tree;
  \item[\desclabel{$(\mathrm{S}_2)$}{bst:S2}]
  $\pD_{\varepsilon} \in \ell(\varepsilon)$;
  \item[\desclabel{$(\mathrm{S}_3)$}{bst:S3}]
  $\elv i$ is a $\NULL$-vertex or $\pD_{i} \in \ell(\elv i)$;
  \item[\desclabel{$(\mathrm{S}_4)$}{bst:S4}]
  every $0$-path of non-$\NULL$ vertices is finite;
  \item[\desclabel{$(\mathrm{S}_5)$}{bst:S5}]
  the $0$-child of~$\varepsilon$ is a $\NULL$-vertex;
  \item[\desclabel{$(\mathrm{S}_6)$}{bst:S6}]
  if $\elv$ is a $\NULL$-vertex, then so is~$\elv i$;
  \item[\desclabel{$(\mathrm{S}_7)$}{bst:S7}]
  if $\ell(\elv) \cap \Gamma^{\mathrm{bck}}_{\varphi} \neq \emptyset$, then $\elv \neq \varepsilon$ and the $0$-child of~$\elv$ is a $\NULL$-vertex.
\end{namedlist}

\noindent Indeed, \ref{bst:S1} is the shape required by~\ref{encready:A0}, and \ref{bst:S2} with~\ref{bst:S3} upgrades the filtered ``exactly one direction letter'' to the faithfulness demanded there.
Since $\pD_{\varepsilon}$ then marks exactly the root, this also validates the root clause of~\ref{encready:A3} and the clause of~\ref{encready:A0} that places letters of~$\Gamma^{1}_{\varphi}$ exactly at the non-root non-$\NULL$ vertices.
Moreover, \ref{bst:S2} forces~$\varepsilon$ to be non-$\NULL$, as $\{\NULL\}$ is the only symbol of $\AlphabetAPTA$ containing~$\NULL$.
Conditions \ref{bst:S4} and~\ref{bst:S5} are \ref{encready:A1}, and \ref{bst:S6} is~\ref{encready:A2}.
Finally, \ref{bst:S7} is the back-edge clause of~\ref{encready:A0}: given \ref{bst:S5} and~\ref{bst:S6}, a non-$\NULL$ vertex is an anchored descendant if and only if it is not the~root~and~its~$0$-child~is~$\NULL$.

In what follows, the states of the intended automata will be ``\emph{label tests}'': for $\Lambda \subseteq \Gamma_{\varphi}$ we let $s_{\Lambda}$ and $\nlab{\Lambda}$ be the states with $\delta(s_{\Lambda}, a) \deff \top$ if $\Lambda \subseteq a$ and $\bot$ otherwise, and $\delta(\nlab{\Lambda}, a) \deff \top$ if $\Lambda \cap a = \emptyset$ and $\bot$ otherwise; for singletons we write $s_{\gamma}$ and~$\nlab{\gamma}$.
Thus, in particular, $s_{\NULL}$ and $\nlab{\NULL}$ test whether the current vertex is and is not a $\NULL$-vertex, respectively.

\noindent \textbf{2APTA.} Each of the following has index~$2$ and $G_1 \deff \emptyset$, so that all of its runs are accepting.
\begin{itemize}\itemsep0.15em
  \item $\autT_{\mathrm{root}}$ verifies \ref{bst:S2} and~\ref{bst:S5}. Its states are $\{s_{\varepsilon}, s_{\NULL}\}$, the initial one being $s_{\varepsilon}$, and $\delta(s_{\varepsilon}, a) \deff (0, s_{\NULL})$ if $\pD_{\varepsilon} \in a$, and $\bot$ otherwise.
  \item $\autT_{\mathrm{dir}}$ checks \ref{bst:S3}. It is obtained by applying Observation~\ref{obs:broadcast} to the 2APTA with states $\{s_{r}\} \cupdot \{\, s'_{i} \,\}$, initial state~$s_{r}$, $\delta(s_{r}, a) \deff \bigwedge_{i=0}^{\rmN} (i, s'_{i})$, and $\delta(s'_{i}, a) \deff \top$ if either $\NULL \in a$ or $\pD_i \in a$, and $\bot$ otherwise.
  \item $\autT_{0\text{-}\mathrm{fin}}$ checks \ref{bst:S4}. Apply Observation~\ref{obs:broadcast} to the automaton with the single state~$s_{u}$, $G_1 \deff \{s_{u}\}$, and $\delta(s_{u}, a) \deff \top$ if $\NULL \in a$, and $(0, s_{u})$ otherwise.
  \item $\autT_{\NULL}$ checks \ref{bst:S6}. Apply Observation~\ref{obs:broadcast} to the automaton with states $\{s_{r}, s_{\NULL}\}$, initial state~$s_{r}$, and $\delta(s_{r}, a) \deff \everywhere{s_{\NULL}}$ if $\NULL \in a$, and $\top$ otherwise.
  \item $\autT_{\mathrm{bck}}$ tests for \ref{bst:S7}. Apply Observation~\ref{obs:broadcast} to the 2APTA with states $\{s_{r}, s_{\NULL}\}$, initial state~$s_{r}$, and $\delta(s_{r}, a) \deff \top$ if $a \cap \Gamma^{\mathrm{bck}}_{\varphi} = \emptyset$; $\delta(s_{r}, a) \deff \bot$ if $a \cap \Gamma^{\mathrm{bck}}_{\varphi} \neq \emptyset$ and $\pD_{\varepsilon} \in a$; and $\delta(s_{r}, a) \deff (0, s_{\NULL})$ otherwise.
\end{itemize}

Let $\autT_{\mathrm{bst}}$ be the intersection of the five 2APTA above, constructed as in Proposition~\ref{prop:2apta-intersection}.
\begin{lemma}\label{lemma:bst-automaton}
  The 2APTA $\autT_{\mathrm{bst}}$ accepts precisely the basic shape $\AlphabetAPTA$-labelled trees.
  It has at most $2|Q|^2{+}14$ states and index~$2$, and is constructed in time polynomial~in~$|\AlphabetAPTA| + |Q|$.
\end{lemma}
\begin{proof}
  By Proposition~\ref{prop:2apta-intersection} and the discussion above, it suffices to see that each 2APTA recognises its condition; the bounds follow, as they have $2$, $\rmN{+}3$, $2$, $3$, and~$3$~states~and~index~$2$.
  For $\autT_{\mathrm{root}}$, an accepting $(\varepsilon, s_{\varepsilon})$-run exists precisely when $\pD_{\varepsilon} \in \ell(\varepsilon)$ and $\NULL \in \ell(0)$.
  For the remaining four we invoke Observation~\ref{obs:broadcast}, which also yields~\ref{bst:S1}, and check the local automaton at an arbitrary vertex~$\elv$.
  In~$\autT_{\mathrm{dir}}$, the copy in state~$s_{r}$ sends to each child~$\elv i$ a copy in state~$s'_{i}$, which accepts if and only if $\elv i$ is a $\NULL$-vertex or carries~$\pD_i$: this is~\ref{bst:S3} at~$\elv$.
  In $\autT_{\NULL}$ and $\autT_{\mathrm{bck}}$ the argument is the same, the runs again having depth at most one, and the case distinction in $\delta(s_{r}, \cdot)$ transcribing~\ref{bst:S6}, respectively~\ref{bst:S7}, verbatim.
  For the latter, recall that $\pD_{\varepsilon} \in \ell(\elv)$ holds exactly for $\elv = \varepsilon$, by~\ref{bst:S2},~\ref{bst:S3}, and the uniqueness of direction letters, which the intersection provides.
  Finally, every $(\elv, s_{u})$-run of the local automaton of $\autT_{0\text{-}\mathrm{fin}}$ contains the path $\elv, \elv 0, \elv 0 0, \ldots$, followed until the first $\NULL$-vertex.
  If the $0$-path of non-$\NULL$ vertices starting at~$\elv$ is finite, so is that path, and the run consisting of it alone is accepting; otherwise every run has an infinite path with $\Inf = \{s_{u}\} = G_1$, and is rejecting.
  Quantifying over all vertices~$\elv$ yields~\ref{bst:S4}.
\end{proof}

\noindent \textbf{Anchors and ends.}\, Fix a basic shape tree and a non-$\NULL$ vertex~$\elv$.
By~\ref{bst:S4} the maximal $0$-path starting at~$\elv$ is finite; write $\mathrm{end}(\elv)$ for its last vertex, which is maximal.
Unless $\elv = \varepsilon$, the maximal $0$-path \emph{ending} at~$\elv$ avoids the root by~\ref{bst:S5}, hence starts at a vertex carrying $\pD_{j}$ for some $j \geqslant 1$; we write~$\mathrm{anc}(\elv)$ for that vertex's parent, and leave $\mathrm{anc}(\varepsilon)$ undefined.
Thus every non-root non-$\NULL$ vertex~$\elv$ lies on exactly one anchor path, that from~$\mathrm{anc}(\elv)$ to~$\mathrm{end}(\elv)$.
The following observation is immediate but useful: 
\begin{observation}\label{obs:anchor-paths-end-anc-functions}
  In a basic shape tree, for all vertices $\elu$ and~$\elw$ the following holds. The vertex~$\elv$ lies on the anchor path from~$\elu$ to~$\elw$ if and only if $\elv \neq \elu$ and $(\elu, \elw) = (\mathrm{anc}(\elv), \mathrm{end}(\elv))$, or $\elv = \elu$ and $\elw$ is an anchored descendant (possibly a child) of~$\elv$.
\end{observation}
\begin{proof}
  This is just a rephrasing of an anchor path in terms of $\mathrm{anc}(\cdot)$, $\mathrm{end}(\cdot)$.
\end{proof}

\subsubsection{Navigation Gadgets}\label{app:navigation}
The first shelf of the library moves a copy of the automaton from the currently scanned vertex~$\elv$ to a relative~$\elw$ of it, and resumes there a task given as a parameter~$s$; the parameter is what makes the gadgets composable.
Below, the target vertex~$\elw$ is either unique (the first two cases) or existentially quantified (the last one). Due to this, the first two gadgets are states, while the third one is a macro.

\begin{itemize}\itemsep0.6em
  \item \textbf{Descend and test.}
  In the state $\dtest{\Lambda}{s}$ the automaton should \emph{``descend the $0$-path to $\mathrm{end}(\elv)$, then check $\Lambda \subseteq \ell(\mathrm{end}(\elv))$, and finally resume~$s$ there''}:
  \[
    \delta(\dtest{\Lambda}{s}, a) \deff \big( (0, s_{\NULL}) \land (\stay, s_{\Lambda}) \land (\stay, s) \big) \;\lor\; \big( (0, \nlab{\NULL}) \land (0, \dtest{\Lambda}{s}) \big).
  \]
  The transition selects one of two cases: either we already reached the target, and then we verify that it is alright, or we continue the search from the $0$-child of the current vertex.

  \item \textbf{Climb.}
  In the state $\clmb{s}$ the automaton should \emph{``climb to $\mathrm{anc}(\elv)$ and resume the task~$s$~there''}:
  \[
    \delta(\clmb{s}, a) \deff (\up, \clmb{s}) \, \text{ if } \pD_{0} \in a; \quad
    \delta(\clmb{s}, a) \deff (\up, s) \,\text{if}\, \pD_{j} \in a, j \geqslant 1; \quad
    \delta(\clmb{s}, a) \deff \bot \, \text{otherwise.}
  \]
  The automaton moves up as long as it stays among $0$-children; on reaching a non-zero child it makes the final move. Otherwise its run fails.

  \item \textbf{Successor.}
  For a non-empty set~$\Lambda$ of possibly inverted actions, the macro $\suc{\Lambda}{a, s}$ means \emph{``move to a guessed vertex that is a $\sigma$-successor of $\elv$ for all $\sigma \in \Lambda$; resume~$s$~there''}:
  \[
    \suc{\Lambda}{a, s} \deff \bigvee_{i=0}^{\rmN} \big( (i, \thst{\Lambda}) \land (i, s) \big)
    \;\lor\; \bigvee_{j=1}^{\rmN} \big( j, \dtest{\pred{bck}_{\inv{\Lambda}}}{s} \big)
    \;\lor\; \gate{(\up, s)}{\inv{\Lambda} \subseteq \Theta(a)}
    \;\lor\; \gate{(\stay, \clmb{s})}{\pred{bck}_{\Lambda} \subseteq a}.
  \]
  The four disjuncts guess, in this order, that the target is a child, an anchored descendant, the parent, or the anchor of~$\elv$; these are exactly the clauses (i)--(iv) defining a $\sigma$-successor.
  Here $\inv{\Lambda} \deff \{\, \inv{\sigma} \mid \sigma \in \Lambda \,\}$ and $\pred{bck}_{\Lambda} \deff \{\, \pred{bck}_{\sigma} \mid \sigma \in \Lambda \,\}$, the state $\thst{\Lambda}$ accepts a letter~$a$ if and only if $\Lambda \subseteq \Theta(a)$, and
  $
    \Theta(a) \; \deff \; (\Gamma^{1}_{\varphi} \cap a) \;\cup\; \{\, \inv{\rho} \mid \pred{bck}_{\rho} \in a \text{ and } \pD_{j} \in a \text{ for some } j \geqslant 1 \,\}
  $
  collects the letters that realise a move from the parent of a vertex labelled~$a$~to~that~vertex.
\end{itemize}

We close the section with two remarks and an almost trivial observation.
First, one pair of vertices may fall under two clauses---a maximal $j$-child is also an anchored descendant---so a set~$\Lambda$ with more than one element may have to be realised partly by letters of $\Gamma^{1}_{\varphi}$ and partly by $\pred{bck}$-letters; this is why $\Theta(a)$ is defined as a \emph{union} of these two cases.
Second, the $\sigma$-successor relation is symmetric under inversion, so a single gadget serves both directions:

\begin{observation}[Inversion]\label{obs:inversion}
  Let $\elv, \elw$ be non-$\NULL$ vertices and let $\sigma$ be a possibly inverted action.
  Then $\elv$ is a $\sigma$-successor of~$\elw$ if and only if $\elw$ is an $\inv{\sigma}$-successor of~$\elv$.
\end{observation}
\begin{proof}
  Swapping $\elv$ and~$\elw$ and replacing $\sigma$ by~$\inv{\sigma}$ turns clause~(i) of the definition of a $\sigma$-successor into clause~(ii) and clause~(iii) into clause~(iv), and conversely.
\end{proof}

\subsubsection{Simulation Gadgets}\label{app:simulation}

The second shelf consists of automata that, from a given vertex, simulate a run of a (word) semiautomaton on the labels read along the way.
All gadgets in this section are parameterised by a semiautomaton $\autE \deff (Q, \Delta_{\autE})$ over~$\Sigma_{\varphi}$: either $\autA$, its extension $\hat{\autA}$ with $\ldem_{q,q'}?$-tests, or the \emph{reversal} $\inv{\hat{\autA}}$ of $\hat{\autA}$.
All three share the state set $Q$; the transition relation of $\inv{\hat{\autA}}$ consists of all triples $(q', \inv{\sigma}, q)$ with $(q, \sigma, q') \in \Delta_{\hat{\autA}}$.
Reversal spares us half of this shelf: an \emph{upward} path for~$\autE$ (required \eg in Condition~\ref{encready:C4}) is a \emph{downward} path for~$\inv{\autE}$, read backwards.~It~is~easy~to~see:

\begin{observation}[Reversal]\label{obs:reversal}
  A walk $\xi$ in a structure is a witnessing walk for $\autE_{q,q'}$ if and only if its reversal $\inv{\xi}$ is one for~$\inv{\autE}_{q',q}$, and a path $\elv_0, \ldots, \elv_n$ in a labelled tree is a witnessing path for $\autE_{q,q'}$ if and only if $\elv_n, \ldots, \elv_0$ is one for~$\inv{\autE}_{q',q}$.
\end{observation}
\begin{proof}
  The first claim is immediate from $\sem{\inv{\sigma}}_{\strA} = \{ (\elv, \elu) \mid (\elu,\elv) \in \sem{\sigma}_{\strA} \}$.
  For the second, note that the readings of the letters of $\Gamma^{0}_{\varphi} \cap \ell(\elx)$ are closed under $\inv{\cdot}$, a test being its own inverse and $\pred{slf}_{\pe}$ offering both $\pe$ and~$\inv{\pe}$; and that $\theta(\elv, \elu) = \{\, \inv{\sigma} \mid \sigma \in \theta(\elu, \elv) \,\}$ for neighbours $\elu, \elv$, by the two clauses defining~$\theta$.
  Reversing an accepted word and inverting each of its letters therefore turns a witness for one side into a witness for the~other.
\end{proof}

\noindent \textbf{Idling.}\, For $a \in \AlphabetAPTA$ let $\stat{\autE}_{a} \subseteq Q \times Q$ collect the pairs $(q,q')$ for which a vertex labelled~$a$ has a \emph{stationary} witnessing path for~$\autE_{q,q'}$, and $\stat{\autE,+}_{a}$ its restriction to paths reading at least one letter.
Both are well defined, the existence of such a path depending only on~$a$, and computable in time polynomial in $|Q| {+} |\Delta_{\autE}| {+} |a|$: build the relation containing $(q, q')$ whenever $(q, \pP{?}, q') \in \Delta_{\autE}$ for some $\pP \in a$, or $(q, \sigma, q') \in \Delta_{\autE}$ for a possibly inverted action~$\sigma$ whose underlying action~$\pe$ satisfies $\pred{slf}_{\pe} \in a$, and take its transitive, respectively reflexive-transitive, closure.
By Observation~\ref{obs:reversal}, the relation $\stat{\inv{\autE}}_{a}$ is the converse of $\stat{\autE}_{a}$.
For $\autE \deff \hat{\autA}$ the fresh transitions contribute exactly the tests $\ldem_{q,q'}?$ with $\ldem_{q,q'} \in a$.

\noindent \textbf{Crossing.}\, Observe that an action $\sigma$ from a vertex to its $\sigma$-successor is recorded at the \emph{lower} of its two endpoints, and this asymmetry shapes the design of our gadgets.
Going down, the automaton must therefore inspect the vertex it moves to; going up, it reads that letter before moving, which makes the ascent the simpler of the two.
Below, $\gamma$ ranges over $\Gamma^{\mathrm{bck}}_{\varphi} \cup \Gamma^{1}_{\varphi} \cup \{\varnothing\}$, the dummy symbol~$\varnothing$ reading \emph{``no letter required''}; we write $\gamma \sqsubseteq a$ for \emph{``$\gamma = \varnothing$ or $\gamma \in a$''}.
All four families of gadgets below come in one copy for each semiautomaton~$\autE$ from our list; we display it in brackets, as in $\statst{q}{q'}[\autE]$, and omit it whenever it is crystal clear from the context.
\begin{itemize}\itemsep0.6em
  \item \textbf{Idle.}
  In the state $\statst{q}{q'}$, the automaton's task is to show that \emph{``the current vertex has a stationary witnessing path for~$\autE_{q,q'}$''}. Here we just exploit the precomputed relation $\stat{\autE}_{a}$:
  \[
    \delta(\statst{q}{q'}, a) \; \deff \;
    \begin{cases}
      \top & \text{if } q \stat{\autE}_{a} q', \\
      \bot & \text{otherwise.}
    \end{cases}
  \]

  \item \textbf{Idle'.}
  In the state $\stpst{q}{q'}{\gamma}$, the automaton's task is to show that \emph{``the current vertex is maximal, carries~$\gamma$, and has a stationary witnessing path for~$\autE_{q,q'}$''}:
  \[
    \delta(\stpst{q}{q'}{\gamma}, a) \; \deff \;
    \begin{cases}
      (0, s_{\NULL}) & \text{if } \gamma \sqsubseteq a \text{ and } q \stat{\autE}_{a} q', \\
      \bot & \text{otherwise.}
    \end{cases}
  \]
  This case differs from the previous one only in two further checks: that the letter~$\gamma$ holds at the scanned vertex, and that the current vertex is maximal.
  The automaton performs the second check by spawning a copy in~$s_{\NULL}$, verifying that the $0$-child is~indeed~$\NULL$.

  \item \textbf{Descend.}
  In the state $\dwnst{q}{q'}{\gamma}$, the automaton's task is to show that \emph{``the $0$-path from the current vertex to $\mathrm{end}(\elv)$ is a witnessing path for~$\autE_{q,q'}$, and $\mathrm{end}(\elv)$ carries~$\gamma$''}:
  \[
    \delta(\dwnst{q}{q'}{\gamma}, a) \; \deff \; (\stay, \stpst{q}{q'}{\gamma}) \,\lor\, \jump{\downarrow}_{0}(a, q, q', \gamma).
  \]
    Intuitively, the transition above offers two options.
    Either the current vertex is already maximal, and we invoke the simulation described in the previous case, or it is not, and we move one step down to the $0$-child and continue the simulation there.
    We~define~$\jump{\downarrow}_{0}$~later.

  \item \textbf{Ascend.}
  In the state $\upst{q}{q'}$, the automaton's task is to show that
  \emph{``the upward path from the current vertex to $\mathrm{anc}(\elv)$ is a witnessing path for~$\autE_{q,q'}$''}:
  \[
    \delta(\upst{q}{q'}, a) \; \deff \;
    \begin{cases}
      \bigvee_{(\tau, q''')} (\up, \upst{q'''}{q'}) & \text{if } \pD_{0} \in a, \\
      \bigvee_{(\tau, q''')} (\up, \statst{q'''}{q'}) & \text{if } \pD_{j} \in a \text{ for some } j \geqslant 1,\\
      \bot & \text{otherwise,}
    \end{cases}
  \]
  both disjunctions ranging over the pairs $(\tau, q''') \in \Gamma^{1}_{\varphi} \times Q$ with $\tau \in a$ and $(q'', \inv{\tau}, q''') \in \Delta_{\autE}$ for some~$q''$ with $q \stat{\autE}_{a} q''$.
  Intuitively, the cases above describe the options for moving up the tree, each updating the state of the simulated semiautomaton (we can idle a bit and then we need to move up).
  If we move to the parent that is also a $0$-child, we continue the process.
  If we move to the parent that is not a $0$-child, we must stay there, updating the state only along stationary paths.
  Otherwise, the move~is~illegal~and~the~automaton~rejects.
\end{itemize}

Above, and repeatedly below, we use two abbreviations, the first for a direction~$j$ and a ``task'' $X \in \{\downarrow, \dashv\}$, the second for a map $\mathbf{s} \colon Q \to S$:
\begin{align*}
  \jump{X}_{j}(a, q, q', \gamma) &\; \deff \!\!\!\bigvee_{\substack{\tau \in \Gamma^{1}_{\varphi},\; q'', q''' \in Q \\ q \,\stat{\autE}_{a}\, q'', \; (q'', \tau, q''') \in \Delta_{\autE}}}\!\!\! (j, s_{\tau}) \,\land\, (j, s^{X}_{q''', q', \gamma}),
  \\
  \move(a, q, \mathbf{s}) &\; \deff \!\!\!\bigvee_{\substack{\sigma \in \Gamma^{1}_{\varphi},\; q'', q''' \in Q \\ q \,\stat{\autE}_{a}\, q'', \; (q'', \sigma, q''') \in \Delta_{\autE}}}\!\!\! \suc{\{\sigma\}}{a, \mathbf{s}(q''')}.
\end{align*}
Both read: \emph{``simulating $\autE$ from~$q$, idle at the current vertex until you reach~$q''$, then cross one edge''}.
In $\jump{X}_{j}$ the edge goes down to the $j$-child, which must carry~$\tau$, and the task resumes in~$s^{X}_{q''', q', \gamma}$.
For $X = \downarrow$ the copy then continues down the $0$-path; for $X = \dashv$ it idles.
In $\move$ the edge goes to a guessed $\sigma$-successor, and the task resumes in~$\mathbf{s}(q''')$.
Thus $\jump{\downarrow}_{j}$ and $\jump{\dashv}_{j}$ walk along the tree, whereas $\move$ walks along the \emph{structure}~encoded~by~it.

We can now state what the whole library does in detail, as summarised in Lemma~\ref{lemma:toolbox} below.
It is one of those long lemmas that collect many properties of the construction, but whose proofs are routine, resting on the semantics of 2APTA and on unfolding the macros.

\begin{lemma}[Toolbox]\label{lemma:toolbox}
  Let $\autT$ be a 2APTA over $\AlphabetAPTA$ containing the state families above, let $(\domT, \ell)$ be a basic shape tree, let $\elv \in \domT$, and put $a \deff \ell(\elv)$. Then:
  \begin{namedlist}[$(\mathrm{T}_8)$\hspace{0.4em}]
    \item[\desclabel{$(\mathrm{T}_1)$}{tb:tests}]
    $s_{\Lambda}$ succeeds at~$\elv$ if and only if $\Lambda \subseteq a$, and $\nlab{\Lambda}$ does if and only if $\Lambda \cap a = \emptyset$.
    \item[\desclabel{$(\mathrm{T}_2)$}{tb:down}]
    for non-empty~$\Lambda$, the state $\dtest{\Lambda}{s}$ succeeds at~$\elv$ if and only if $\elv$ is non-$\NULL$, $\Lambda \subseteq \ell(\mathrm{end}(\elv))$, and $s$ succeeds at~$\mathrm{end}(\elv)$; and $\clmb{s}$ succeeds at~$\elv$ if and only if $\mathrm{anc}(\elv)$ is defined and $s$ succeeds there.
    \item[\desclabel{$(\mathrm{T}_3)$}{tb:succ}]
    for non-empty~$\Lambda$, the formula $\suc{\Lambda}{a, s}$ is realisable at~$\elv$ if and only if $s$ succeeds at some vertex~$\elw$ such that, for every $\sigma \in \Lambda$, $\elw$ is a $\sigma$-successor of~$\elv$.
    \item[\desclabel{$(\mathrm{T}_4)$}{tb:stat}]
    $\statst{q}{q'}[\autE]$ succeeds at~$\elv$ if and only if $\elv$ has a stationary witnessing path for~$\autE_{q,q'}$, and $\stpst{q}{q'}{\gamma}[\autE]$ does if and only if, in addition, $\elv$ is maximal and $\gamma \sqsubseteq a$.
    \item[\desclabel{$(\mathrm{T}_5)$}{tb:dwn}]
    for non-$\NULL$~$\elv$, the state $\dwnst{q}{q'}{\gamma}[\autE]$ succeeds at~$\elv$ if and only if the $0$-path from~$\elv$ to $\mathrm{end}(\elv)$ is a witnessing path for~$\autE_{q,q'}$ and $\gamma \sqsubseteq \ell(\mathrm{end}(\elv))$. Consequently, for every direction~$j$, the formula $\jump{\downarrow}_{j}[\autE](a, q, q', \gamma)$ is realisable at~$\elv$ if and only if some witnessing path for~$\autE_{q,q'}$ anchored at~$\elv$ in direction~$j$ ends in a vertex~$\elw$ with $\gamma \sqsubseteq \ell(\elw)$; the same holds for $\jump{\dashv}_{j}[\autE]$, restricted to paths of length one.
    \item[\desclabel{$(\mathrm{T}_6)$}{tb:up}]
    $\upst{q}{q'}[\autE]$ succeeds at~$\elv$ if and only if $\mathrm{anc}(\elv)$ is defined and the upward path from~$\elv$ to $\mathrm{anc}(\elv)$ is a witnessing path for~$\autE_{q,q'}$.
    \item[\desclabel{$(\mathrm{T}_7)$}{tb:move}]
    $\move[\autE](a, q, \mathbf{s})$ is realisable at~$\elv$ if and only if there are $q'', q''' \in Q$ and $\sigma \in \Gamma^{1}_{\varphi}$ with $q \stat{\autE}_{a} q''$ and $(q'', \sigma, q''') \in \Delta_{\autE}$, and a $\sigma$-successor~$\elw$ of~$\elv$ at which $\mathbf{s}(q''')$ succeeds.
    \item[\desclabel{$(\mathrm{T}_8)$}{tb:hygiene}]
    all gadgets above are root-safe, and every branch that a gadget spawns, other than those of its parameters, is finite. More precisely, on any path of a run, only finitely many configurations in the states above lie between two consecutive configurations in a parameter state ($s$ in \ref{tb:down} and~\ref{tb:succ}, and $\mathbf{s}(\cdot)$ in~\ref{tb:move}).
  \end{namedlist}
  Clauses \ref{tb:down}, \ref{tb:succ}, and~\ref{tb:move} hold for arbitrary parameters, whether their runs are finite or not: by~\ref{tb:hygiene}, a gadget only prefixes a run of its parameter with finitely many configurations, which affects neither $\Inf$ nor the parity condition.
\end{lemma}
\begin{proof}
  \ref{tb:tests} and~\ref{tb:stat} hold by definition, the latter since a non-$\NULL$ vertex is maximal if and only if its $0$-child is a $\NULL$-vertex.
  \ref{tb:move} follows from~\ref{tb:succ}, disjunct by disjunct.

  \ref{tb:down}. The two disjuncts of $\delta(\dtest{\Lambda}{s}, a)$ exclude each other, as the first requires the $0$-child to be a $\NULL$-vertex and the second forbids it.
  Hence a copy in $\dtest{\Lambda}{s}$ descends the $0$-path from~$\elv$, finite by~\ref{bst:S4}, and spawns $s_{\Lambda}$ and~$s$ exactly at $\mathrm{end}(\elv)$.
  At a $\NULL$-vertex it takes the first disjunct by~\ref{bst:S6}, and $s_{\Lambda}$ fails, as $\Lambda \neq \emptyset$ and $\{\NULL\}$ is the only letter of~$\AlphabetAPTA$ containing~$\NULL$.
  Dually, a copy in $\clmb{s}$ climbs while reading~$\pD_{0}$, then takes one more step, to $\mathrm{anc}(\elv)$; it rejects at the root and at $\NULL$-vertices, whose labels carry no direction letter, and terminates as the depth decreases.

  \ref{tb:succ}. Let $\Theta(\elx, \ely) \deff \{\, \sigma \mid \ely \text{ is a } \sigma\text{-successor of } \elx \,\}$; we show that the disjuncts of $\suc{\Lambda}{a, s}$ reach exactly the vertices~$\elw$ with $\Lambda \subseteq \Theta(\elv, \elw)$.
  As $\Lambda \neq \emptyset$, each disjunct forces $\elv$ and its target to be non-$\NULL$.
  By \ref{bst:S5}--\ref{bst:S7}, only anchored descendants carry letters $\pred{bck}_{\rho}$, and such a vertex is the anchored descendant of its parent if and only if it also carries some~$\pD_{j}$ with $j \geqslant 1$.
  Reading off clauses (i)--(iv) from the definition of a successor, $\Theta(\ely, \elx) = \Theta(\ell(\elx))$ for every child~$\elx$ of~$\ely$, which settles the first disjunct and, with Observation~\ref{obs:inversion}, the third.
  By Observation~\ref{obs:anchor-paths-end-anc-functions}, the anchored descendants of~$\elv$ are the vertices $\mathrm{end}(\elv j)$ with $j \geqslant 1$ and $\elv j$ non-$\NULL$, so by~\ref{tb:down} the second disjunct covers clause~(iii); the fourth covers clause~(iv).

  \ref{tb:dwn}. The proof goes via induction on the height of runs, which are finite by~\ref{bst:S4}.
  The disjunct $(\stay, \stpst{q}{q'}{\gamma})$ is handled by~\ref{tb:stat}; for the other, the disjuncts of $\jump{\downarrow}_{0}$ match one-to-one the blocks from $(\Gamma^{0}_{\varphi} \cap \ell(\elv))^{*} \cdot \theta(\elv, \elv 0)$ in the definition of a witnessing path, by the definition of~$\stat{\autE}_{a}$ and as $\theta(\elv, \elv 0) = \Gamma^{1}_{\varphi} \cap \ell(\elv 0)$.
  Prefixing the step to the $j$-child yields the claims on $\jump{X}_{j}$.

  \ref{tb:up} is dual, but reads the letter crossing an edge at the current vertex~$\elx$, before moving to its parent~$\ely$: a disjunct of $\delta(\upst{q}{q'}, a)$ consumes a block from $(\Gamma^{0}_{\varphi} \cap \ell(\elx))^{*}$ via $q \stat{\autE}_{a} q''$, then the letter $\inv{\tau} \in \theta(\elx, \ely)$ with $\tau \in \Gamma^{1}_{\varphi} \cap \ell(\elx)$.
  The copy reaches $\mathrm{anc}(\elv)$ as in~\ref{tb:down}, and~\ref{tb:stat} handles the last vertex.

  \ref{tb:hygiene}. Atoms $(\up, \cdot)$ occur only in $\delta(\clmb{\cdot}, a)$, in $\delta(\upst{\cdot}{\cdot}, a)$, and in the third disjunct of $\suc{\Lambda}{a, \cdot}$.
  The first two are~$\bot$ unless $a$ contains $\pD_{0}$ or some~$\pD_{j}$ with $j \geqslant 1$.
  The third requires $\inv{\Lambda} \subseteq \Theta(a)$, which fails if $\pD_{\varepsilon} \in a$: by the filtering of~$\AlphabetAPTA$, such an~$a$ has no other direction letter and no letter of~$\Gamma^{1}_{\varphi}$, so $\Theta(a) = \emptyset$.
  For finiteness, every move from a state above enters a parameter state, is a $\stay$-move (these are acyclic on the states above), descends a $0$-path, finite by~\ref{bst:S4}, or strictly decreases the depth.
  Hence no path of a run stays among these states forever; in particular, a run whose parameters are all label tests is finite.
\end{proof}

\subsection{A 2APTA for Encoding-Ready Trees}\label{app:encready-automaton}

With the library in place, writing 2APTA for the remaining conditions is almost straightforward.
Throughout this section we work over basic shape trees, so that Lemma~\ref{lemma:toolbox} applies, and  we combine the resulting automata only at the end. For the reader's convenience, we remind that $\kappa$ maps the demands bijectively onto the non-zero directions $\{1, \ldots, \rmN\}$.

\subsubsection{Providing Witnesses: \ref{encready:A4} and \ref{encready:A5}}\label{app:twoapta-witnesses}

Both conditions constrain a single vertex, hence it suffices to verify them at 
an arbitrary vertex and let Observation~\ref{obs:broadcast} do the bookkeeping.
A vertex satisfying $\wdem_{q,q'}$ must exhibit a witnessing path for $\autA_{q,q'}$ that is stationary or anchored in the direction reserved for that demand; this is literally $\stat{\autA}_{a}$ or $\jump{\downarrow}$.
A vertex satisfying $\ldem_{q,q'}$ must in addition close the anchored path by a single step back, which by~\ref{encready:A5} leads from the last vertex~$\elw$ of that path to a $\sigma$-successor (for some possibly inverted action $\sigma$) of~$\elw$, namely~$\elv$ itself. As $\elw$ is then a strict descendant of~$\elv$, only clauses (ii) and~(iv) of the definition of a $\sigma$-successor can apply: $\elw$ is a child of~$\elv$ carrying~$\inv{\sigma}$, or an anchored descendant of~$\elv$ carrying~$\pred{bck}_{\sigma}$. These are matched by $\jump{\dashv}$ and $\jump{\downarrow}$, respectively.
We therefore obtain $\autT_{\mathrm{wtn}}$ by applying Observation~\ref{obs:broadcast} to the 2APTA with initial state~$s_{\mathrm{wtn}}$ and acceptance condition $(\emptyset, S)$, where
\[
  \delta(s_{\mathrm{wtn}}, a) \; \deff \; \bigwedge_{\wdem_{q,q'} \in a} \alpha^{\wdem}_{q,q'}(a) \;\land\; \bigwedge_{\ldem_{q,q'} \in a} \alpha^{\ldem}_{q,q'}(a)
\]
and where $\alpha^{\wdem}_{q,q'}(a) \deff \top$ and $\alpha^{\ldem}_{q,q'}(a) \deff \top$ whenever $q \stat{\autA}_{a} q'$, while otherwise
\begin{align*}
  \alpha^{\wdem}_{q,q'}(a) &\deff \jump{\downarrow}_{\kappa(\wdem_{q,q'})}[\autA](a, q, q', \varnothing),
  \\
  \alpha^{\ldem}_{q,q'}(a) &\deff \!\!\!\bigvee_{\substack{(q_1, \sigma, q_2) \in \Delta \\ q_2 \,\stat{\autA}_{a}\, q'}}\!\!\! \Big( \jump{\downarrow}_{\kappa(\ldem_{q,q'})}[\autA](a, q, q_1, \pred{bck}_{\sigma}) \;\lor\; \jump{\dashv}_{\kappa(\ldem_{q,q'})}[\autA](a, q, q_1, \inv{\sigma}) \Big).
\end{align*}

\noindent The next lemma is immediate from the semantics of our gadgets and the explanations above.
\begin{lemma}\label{lemma:wtn-automaton}
  A basic shape tree is accepted by $\autT_{\mathrm{wtn}}$ if and only if it satisfies \ref{encready:A4} and~\ref{encready:A5}. It has polynomially many states in $|Q| {+} |\Sigma_{\varphi}|$ and index~$2$, and is constructed in time polynomial in $|\AlphabetAPTA| {+} |Q| {+} |\Delta|$.
\end{lemma}
\begin{proof}
  By Observation~\ref{obs:broadcast} it remains to show that $s_{\mathrm{wtn}}$ succeeds at~$\elv$ precisely when $\elv$ satisfies~\ref{encready:A4} and~\ref{encready:A5}.
  By~\ref{tb:dwn}, $\alpha^{\wdem}_{q,q'}(\ell(\elv))$ is realisable at~$\elv$ exactly when $\elv$ has a stationary witnessing path for $\autA_{q,q'}$ or one anchored in direction $\kappa(\wdem_{q,q'})$, which is~\ref{encready:A4}; and $\alpha^{\ldem}_{q,q'}(\ell(\elv))$ exactly when \ref{encready:A5} holds at~$\elv$, by the discussion above and because $q_2 \stat{\autA}_{\ell(\elv)} q'$ is the stationary tail required there. A $\NULL$-vertex carries no demand, so both conjunctions are empty at~$\{\NULL\}$, matching that \ref{encready:A4} and~\ref{encready:A5} speak about non-$\NULL$ vertices only. The bounds are read off the gadgets of Lemma~\ref{lemma:toolbox} instantiated with $\autE \deff \autA$, plus the state added by Observation~\ref{obs:broadcast}; the time is dominated by computing $\stat{\autA}_{a}$ per each~of~$|\AlphabetAPTA|$~letters.
\end{proof}

\subsubsection{Closing the Demands: \ref{encready:C1}--\ref{encready:C4}}\label{app:twoapta-closure}

All four remaining conditions have the same shape: an \emph{existential} hypothesis---the existence of a $\sigma$-successor, or of a witnessing path along an anchor path---and a label-local conclusion.
We therefore build a single 2APTA $\autT_{\viol}$ that, started at a vertex, \emph{guesses a violation} of one of them there, and let Observation~\ref{obs:duality} do the rest.
Its initial state is $s_{\viol}$ with $\delta(s_{\viol}, a) \deff \bigvee_{k=1}^{4} (\stay, s_{\mathrm{C}_k})$, its acceptance condition is $(\emptyset, S)$, and its four main states are as follows.
For \ref{encready:C1} and~\ref{encready:C2}, Observation~\ref{obs:inversion} turns the clause ``$\elw$ is a $\tau$-successor of~$\elv$'' of~\ref{encready:C1} into ``$\elv$ is an $\inv{\tau}$-successor of~$\elw$'', so both requirements on the guessed vertex are discharged by a single call to $\suc{\{\sigma, \inv{\tau}\}}{\cdot}$. The conditions below are almost verbatim translations of \ref{encready:C1}--\ref{encready:C2}, modulo rewriting implications of conjunctions as disjunctions.

\begin{align*}
  \delta(s_{\mathrm{C}_1}, a)
    &\deff \!\!\bigvee_{\substack{(q_1, \sigma, q_2),\, (q_3, \tau, q_4) \in \Delta \\ \ldem_{q_1,q_4} \notin a}}\!\!\!
      \suc{\{\sigma, \inv{\tau}\}}{a,\, s_{\ldem_{q_2,q_3}}},
  \\[0.4em]
  \delta(s_{\mathrm{C}_2}, a)
    &\deff \!\!\bigvee_{\substack{(q_1, \sigma, q_2) \in \Delta,\; q_3 \in Q \\ \wdem_{q_1,q_3} \notin a}}\!\!\!
      \suc{\{\sigma\}}{a,\, s_{\wdem_{q_2,q_3}}}.
\end{align*}
Conditions \ref{encready:C3} and~\ref{encready:C4} speak about a vertex~$\elv$ on the anchored path from an anchor~$\elu$ to an anchored descendant~$\elw$.
By Observation~\ref{obs:anchor-paths-end-anc-functions} the pair $(\elu, \elw)$ is \emph{determined} by~$\elv$, unless $\elv = \elu$, in which case the segment from~$\elu$ to~$\elv$ is stationary and the one from~$\elv$ to~$\elw$ is anchored in a guessed direction.
This dichotomy is the two disjuncts below; each segment is simulated in the direction in which the automaton walks, which by Observation~\ref{obs:reversal} means exchanging $\autE \deff \hat{\autA}$ with~$\inv{\autE}$ and swapping the corresponding pair of states:
\begin{align*}
  \delta(s_{\mathrm{C}_3}, a) \deff
    \!\!\bigvee_{\substack{(q_2, \sigma, q_3) \in \Delta \\ q_1, q_4 \in Q,\; \ldem_{q_1,q_4} \notin a}}\!\!\! \Big(
    &\gate{\textstyle\bigvee_{j=1}^{\rmN} \jump{\downarrow}_{j}[\autE](a, q_1, q_2, \pred{bck}_{\sigma})}{q_3 \,\stat{\autE}_{a}\, q_4}
  \\[0.15em]
  {}\lor\;\;
    &\big( (\stay, \dwnst{q_1}{q_2}{\pred{bck}_{\sigma}}[\autE]) \land (\stay, \upst{q_4}{q_3}[\inv{\autE}]) \big) \Big),
  \\[0.5em]
  \delta(s_{\mathrm{C}_4}, a) \deff
    \!\!\bigvee_{\substack{(q_2, \sigma, q_3) \in \Delta \\ q_1, q_4 \in Q,\; \ldem_{q_1,q_4} \notin a}}\!\!\! \Big(
    &\gate{\textstyle\bigvee_{j=1}^{\rmN} \jump{\downarrow}_{j}[\inv{\autE}](a, q_4, q_3, \pred{bck}_{\inv{\sigma}})}{q_1 \,\stat{\autE}_{a}\, q_2}
  \\[0.15em]
  {}\lor\;\;
    &\big( (\stay, \upst{q_1}{q_2}[\autE]) \land (\stay, \dwnst{q_4}{q_3}{\pred{bck}_{\inv{\sigma}}}[\inv{\autE}]) \big) \Big).
\end{align*}

\begin{lemma}\label{lemma:closure-automaton}
  Let $\autT_{\mathrm{C}}$ be the result of applying Observation~\ref{obs:broadcast} to $\overline{\autT_{\viol}}$.
  A basic shape tree is accepted by $\autT_{\mathrm{C}}$ if and only if it satisfies \ref{encready:C1}--\ref{encready:C4}.
  The automaton $\autT_{\mathrm{C}}$ has polynomially many states in $|Q| {+} |\Sigma_{\varphi}|$ and index~$2$, and is constructed in time polynomial in $|\AlphabetAPTA| {+} |Q| {+} |\Delta|$.
\end{lemma}
\begin{proof}
  The automaton $\autT_{\viol}$ is root-safe by~\ref{tb:hygiene}, and its acceptance condition $(\emptyset, S)$ is complemented by $(S, S)$, again of index~$2$; so by Observations \ref{obs:duality} and~\ref{obs:broadcast}, $\autT_{\mathrm{C}}$ accepts a basic shape tree if and only if $s_{\viol}$ succeeds at no vertex of it. It suffices to show that $s_{\viol}$ succeeds at~$\elv$ precisely when one of \ref{encready:C1}--\ref{encready:C4} fails there.

  For \ref{encready:C1}, clause~\ref{tb:succ} says that $\suc{\{\sigma, \inv{\tau}\}}{\ell(\elw), s_{\ldem_{q_2,q_3}}}$ is realisable at~$\elw$ exactly when some vertex is both a $\sigma$- and an $\inv{\tau}$-successor of~$\elw$ and carries~$\ldem_{q_2,q_3}$; with Observation~\ref{obs:inversion} and the side condition $\ldem_{q_1,q_4} \notin \ell(\elw)$ this is the negation of~\ref{encready:C1} there. Condition~\ref{encready:C2} is the same, with a single letter and the demand~$\wdem$.

  For~\ref{encready:C3}, fix $\elv$ and a tuple with $\ldem_{q_1,q_4} \notin \ell(\elv)$, and recall Observation~\ref{obs:anchor-paths-end-anc-functions}.
  If $\elv = \elu$, the downward witnessing path for $\hat{\autA}_{q_3,q_4}$ from~$\elu$ to~$\elv$ is stationary, that is $q_3 \,\stat{\autE}_{\ell(\elv)}\, q_4$, while the one for $\hat{\autA}_{q_1,q_2}$ from~$\elv$ to~$\elw$ is anchored in some direction $j \geqslant 1$ and ends on a vertex carrying $\pred{bck}_{\sigma}$: by~\ref{tb:dwn}, the first disjunct.
  If $\elv \neq \elu$, the path for $\hat{\autA}_{q_1,q_2}$ instead descends the $0$-path from~$\elv$ to $\mathrm{end}(\elv)$, which $\dwnst{q_1}{q_2}{\pred{bck}_{\sigma}}[\autE]$ tests by~\ref{tb:dwn}, while the path from~$\elu$ to~$\elv$ climbs back to $\mathrm{anc}(\elv)$ and is tested by $\upst{q_4}{q_3}[\inv{\autE}]$, by Observation~\ref{obs:reversal} and~\ref{tb:up}.
  Conversely, $\upst{\cdot}{\cdot}$ fails at the root and its ascent is non-trivial, so the second disjunct is realisable only when $\elv \neq \elu$.
  Hence $\delta(s_{\mathrm{C}_3}, \ell(\elv))$ is realisable exactly when \ref{encready:C3} fails at~$\elv$.
  Condition~\ref{encready:C4} is the mirror image: both of its paths are upward, so Observation~\ref{obs:reversal} exchanges $\autE$ with~$\inv{\autE}$, swaps the two pairs of states, and turns $\pred{bck}_{\sigma}$ into $\pred{bck}_{\inv{\sigma}}$.

  For the bounds, $\autT_{\viol}$ consists of the five states above together with the gadgets of Lemma~\ref{lemma:toolbox}, instantiated with $\autE \in \{\hat{\autA}, \inv{\hat{\autA}}\}$ and with the $O((|Q| {+} |\Sigma_{\varphi}|)^{2})$ parameters $\Lambda$ and~$s$ that occur; dualisation adds no state, Observation~\ref{obs:broadcast} adds one, and neither changes the index. The transition function is tabulated letter by letter, each entry computable in polynomial time once $\stat{\autE}_{a}$ is.
\end{proof}

\subsubsection{Putting It Together}

Harvesting the results of the previous subsections, we obtain a 2APTA that accepts precisely the encoding-ready trees for~$\varphi$.

\begin{lemma}\label{lemma:encoding-ready-automaton}
  The 2APTA $\autT_{\varphi}$, the intersection of $\autT_{\mathrm{bst}}$, $\autT_{\mathrm{wtn}}$ and $\autT_{\mathrm{C}}$, accepts precisely the encoding-ready trees for~$\varphi$. It has index~$2$, polynomially many states in~$|\varphi|$, and is computed in time exponential in~$|\varphi|$.
\end{lemma}
\begin{proof}
  A $2^{\Gamma_{\varphi}}$-labelled tree using a letter outside~$\AlphabetAPTA$ violates some condition of Section~\ref{app:toolbox}, hence is not encoding-ready; so we restrict attention to $\AlphabetAPTA$-labelled trees, where the alphabet itself secures \ref{encready:A3}, \ref{encready:B0}--\ref{encready:B5} and the label-local part of~\ref{encready:A0}.
  By Lemma~\ref{lemma:bst-automaton}, the conditions other than \ref{encready:A4}, \ref{encready:A5} and \ref{encready:C1}--\ref{encready:C4} cut out exactly the basic shape trees; on those, Lemmas~\ref{lemma:wtn-automaton} and~\ref{lemma:closure-automaton} add the six missing ones.
  Proposition~\ref{prop:2apta-intersection} assembles the three automata, preserving the index and summing the state counts, which stay polynomial in~$|\varphi|$ since $|Q|$ and~$|\Delta|$ are.
  Tabulating the transition functions over~$\AlphabetAPTA$ dominates the running time, and $|\AlphabetAPTA|$ is exponential in~$|\varphi|$.
\end{proof}

\subsection{Proof of Lemma~\ref{lemma:cactus-automaton}}
\begin{proof}
  Definition~\ref{def:encoding} defines~$\enc$, and Lemma~\ref{lemma:encoding-ready-automaton} constructs $\autT_{\varphi}$ in exponential time. Its non-emptiness can be tested in exponential time by Proposition~\ref{prop:2apta-nonemptiness}. Every (regular) tree~$\treeT$ accepted by $\autT_{\varphi}$ is of the form $\enc(\strC)$ for the $\varphi$-cactus $\strC \deff \dec(\treeT)$, by Lemma~\ref{lemma:from-encoding-ready-trees-to-phi-cacti} and Claim~\ref{claim:decoding-gives-unique-cactus-like-structures}.
\end{proof}

\subsection{Proof of Remark~\ref{remark:general-loop-pdl}}

We now reprove the following theorem, proved without the converse by Danecki~\cite{DaneckiLoopPDL84} in 1984.
The method mirrors the one of Section~\ref{sec:cactus-models}; only minor modifications are needed.

\begin{procedurebox}{Deciding general satisfiability of $\LoopPDL$ in $\ExpTime$}
\begin{namedlist}[3.]
  \item[\textbf{1.}] Given an input $\LoopPDL$-formula~$\varphi'$, compute a semiautomaton~$\autA$ and a formula~$\varphi$ in $\autA$-normal form (we treat all automata as bannable). \why{A generalisation of Lemma~\ref{lemma:normal-form} assuming that all automata are bannable; polynomial time, and $|\varphi|$ is polynomial in~$|\varphi'|$.}
  \item[\textbf{2.}] Construct a 2APTA~$\autT_{\varphi}$ recognising the encodings of exact $\varphi$-cacti. \why{Lemmas~\ref{lemma:encoding-ready-automaton}, \ref{lemma:from-encoding-ready-trees-to-phi-cacti}.}
  \item[\textbf{3.}] Return \textsc{Yes} if $\autT_{\varphi}$ has non-empty language, and \textsc{No} otherwise. \why{Lemma~\ref{lemma:encoding-ready-automaton}; exp.~time~in~$|\varphi|$.}
\end{namedlist}
\end{procedurebox}

We compute the normal form $\varphi$ by following the proof of Lemma~\ref{lemma:normal-form} and treating all automata as bannable; its correctness proof is verbatim the one given there.
The automaton $\autT_{\varphi}$ can be constructed (Lemma~\ref{lemma:encoding-ready-automaton}) and tested for non-emptiness (Proposition~\ref{prop:2apta-nonemptiness}) in exponential time in~$|\varphi|$.
It remains to show that $\autT_{\varphi}$ has a non-empty language if and only if $\varphi$ has a (possibly infinite) model.
Suppose $\varphi$ has a model.
By the L\"owenheim--Skolem theorem we may assume it is countable, and hence, by Lemma~\ref{lemma:model-to-cactus}, that $\varphi$ has a $\varphi$-cactus model $\strC$.
Then $\enc(\strC)$ is an encoding-ready tree (Lemma~\ref{lemma:from-cacti-to-encoding-ready-trees}), and is thus accepted by $\autT_{\varphi}$.
Conversely, let $\autT_{\varphi}$ accept some encoding-ready tree.
Decoding it yields an exact $\varphi$-cactus $\strC$ (Lemma~\ref{lemma:from-encoding-ready-trees-to-phi-cacti}), which is by design a model of~$\varphi$.
This concludes the proof.

\subsection{Proof of Lemma~\ref{lemma:if-finite-then-cactus-foldable}}
\begin{proof}
Regularity was already established in Lemma~\ref{lemma:model-to-cactus}, and 
thus we focus on foldability only.
Let $\homof \colon \strC \to \strA$ be as in Lemma~\ref{lemma:model-to-cactus}, and suppose that $\strC$ has a bad $\omega$-walk $\xi_0 \xi_1 \ldots$ for a concatenation-closed $\hat{\autA}_{q,q'}$ with bannable $\autA_{q,q'}$.
Let $\elv_k$ be the first vertex of~$\xi_k$.
As $\strA$ is finite, $\homof(\elv_i) = \homof(\elv_j) \eqqcolon \elb$ for some $i < j$.
By concatenation-closedness, $\zeta \deff \xi_i \cdots \xi_{j-1}$ or its reversal is a witnessing walk for $\hat{\autA}_{q,q'}$, so its $\homof$-image yields a witnessing cycle for $\hat{\autA}_{q,q'}$ at~$\elb$.
As $\autA_{q,q'}$ is bannable, $\elb$ satisfies $\ldem_{q,q'}$, whereas $\elv_i$ does not; this contradicts~\ref{cactus:prop}.
\end{proof}

\subsection{Foldable Encoding-Ready Trees and Their 2APTA}\label{app:foldable-cactus-automaton}

\newcommand{\bad}{\mathrm{bad}}
\newcommand{\seg}{\mathrm{seg}}
\newcommand{\wlk}[1]{\mathrm{wlk}_{#1}}

Recall from Definition~\ref{def:foldability} that a \emph{bad $(q,q')$-walk} is a witnessing walk for $\hat{\autA}_{q,q'}$ whose ends both fail~$\ldem_{q,q'}$.
A structure is foldable if it admits no $\omega$-walk factorising into infinitely many bad $(q,q')$-walks, or into the reversals of such, for a bannable concatenation-closed~$\autA_{q,q'}$.
Let $\rmB \subseteq Q \times Q$ collect the pairs $(q,q')$ for which $\autA_{q,q'}$ is bannable and concatenation-closed.
By Lemma~\ref{lemma:normal-form}, this set is produced together with~$\autA$ and~$\varphi$: the bannable NFAs are exactly those built for the negated loops of the input, and the concatenation-closed ones among them are exactly those built for a~$\neg\lop(\pi^{+})$.
An encoding-ready tree~$\treeT$ is \emph{foldable} if the decoding $\dec(\treeT)$ is.
We first bring the definition into a shape that an automaton can follow more easily.

\begin{observation}\label{obs:bad-normalised}
  A $\varphi$-cactus $\strC$ has a bad $\omega$-walk if and only if there are
  \begin{itemize}\itemsep0em
    \item a pair $(q,q') \in \rmB$ and a triple $(\autE, p_{\mathrm{in}}, p_{\mathrm{fin}})$ equal to either $(\hat{\autA},\, q,\, q')$ or $(\inv{\hat{\autA}},\, q',\, q)$,
    \item an infinite sequence $\elv_0, \elv_1, \ldots$ of vertices of~$\strC$, none of which satisfies~$\ldem_{q,q'}$, and
    \item for every~$i$, a witnessing walk~$\xi_i$ for $\autE_{p_{\mathrm{in}}, p_{\mathrm{fin}}}$ from~$\elv_i$ to~$\elv_{i+1}$ that makes at least one step.
  \end{itemize}
  We call the walks~$\xi_i$ the \emph{segments} of this factorisation.
\end{observation}
\begin{proof}
  Observation~\ref{obs:reversal} identifies the case in which the~$\inv{\xi_i}$ are the bad $(q,q')$-walks with the triple $(\inv{\hat{\autA}}, q', q)$, so both directions follow from the case $\autE = \hat{\autA}$; only the steps need bookkeeping.
  For~($\Leftarrow$), each~$\xi_i$ makes a step, so $\xi_0 \xi_1 \cdots$ is an $\omega$-walk, bad~by~definition.
  For~($\Rightarrow$), factorise a bad $\omega$-walk as $\xi_0 \xi_1 \cdots$, let $\elv_i$ be the first vertex of~$\xi_i$, and discard the segments making no step: each has $\elv_i = \elv_{i+1}$, so the survivors still factorise the same $\omega$-walk, and infinitely many survive, an~$\omega$-walk~being~infinite.
\end{proof}

\noindent \textbf{The automaton $\autT_{\bad}$.}\, Observation~\ref{obs:bad-normalised} has the shape of a B\"uchi condition: walk forever, and pass infinitely often through a \emph{segment boundary}, that is, through one of the vertices~$\elv_i$, where a segment ends, the next one begins, and the simulated run of~$\autE$ restarts at~$p_{\mathrm{in}}$.
The walk itself is simulated by~$\move$, whose two ingredients---idling at the current vertex and stepping to a guessed $\sigma$-successor---are, by~\ref{dict:edge}, exactly the two ways in which a walk of $\dec(\treeT)$ can proceed.
Accordingly, for every $(q,q') \in \rmB$ and every triple as above, $\autT_{\bad}$ has the states $\seg$ and $\wlk{p}$ for $p \in Q$ (all decorated with $q$, $q'$, and~$\autE$, which we suppress).
Writing $\wlk{}$ for the map $p \mapsto \wlk{p}$, their transitions are:
\begin{align*}
  \delta(\seg, a) &\deff \gate{\;\gate{(\stay, \seg)}{p_{\mathrm{in}} \stat{\autE,+}_{a} p_{\mathrm{fin}}} \;\lor\; \move[\autE](a, p_{\mathrm{in}}, \wlk{})\;}{\ldem_{q,q'} \notin a},
  \\
  \delta(\wlk{p}, a) &\deff \gate{(\stay, \seg)}{p \stat{\autE}_{a} p_{\mathrm{fin}}} \;\lor\; \move[\autE](a, p, \wlk{}).
\end{align*}
Both lines are read alike.
The state~$\wlk{p}$ says \emph{``I am inside a segment~$\xi_i$, the simulated run of~$\autE$ has got as far as~$p$, and I still owe a witnessing walk for $\autE_{p, p_{\mathrm{fin}}}$ from here''}; the state~$\seg$ says the same with $p = p_{\mathrm{in}}$, and in addition that the segment begins at the current world.
From either state, the automaton picks one of two moves.
It may \emph{close} the segment, idling on the spot until the simulated run reaches~$p_{\mathrm{fin}}$ and returning to~$\seg$ at the very same vertex, which is where the next segment then starts; or it may \emph{continue} the segment with~$\move$, idling for a while and then crossing one edge to a guessed $\sigma$-successor, where it resumes in whichever state~$\wlk{p_2}$ the crossing has advanced the simulated run to.
The two lines differ only in~$\delta(\seg, \cdot)$, and there in two places: the gate $\gate{\;\cdot\;}{\ldem_{q,q'} \notin a}$ around the whole of it lets a segment open only at a world that fails the cycle demand, and the~$+$ in~$\stat{\autE,+}_{a}$ makes a segment that is closed straight from~$\seg$, without any~$\move$, read at least one letter.
So, read along a run, these moves spell out a factorisation as in Observation~\ref{obs:bad-normalised}---the vertices at which the run is in state~$\seg$ are the~$\elv_i$, and what happens between two consecutive visits to~$\seg$ is~a~$\xi_i$.
The initial state is~$s_{\bad}$, with $\delta(s_{\bad}, a) \deff \bigvee (\stay, \seg^{q,q',\autE})$, the disjunction ranging over all the above choices.
The acceptance condition is $(\emptyset, \rmB^{\ast}, S)$, where $\rmB^{\ast}$ is the set of all the states~$\seg^{q,q',\autE}$: a run is accepting if and only if each of its infinite paths meets~$\rmB^{\ast}$ infinitely~often. Unfolding the macros now tells us what~$\seg$ does.

\begin{observation}\label{obs:seg-succeeds}
  Let $\treeT$ be an encoding-ready tree and put $\strC \deff \dec(\treeT)$.
  Then, for every vertex~$\elv$ of~$\treeT$ and every choice of $(q,q')$ and $(\autE, p_{\mathrm{in}}, p_{\mathrm{fin}})$ as above, the state $\seg$ succeeds at~$\elv$ if and only if $\strC$ has worlds $\elc_0 = \elv, \elc_1, \ldots$ and walks $\xi_0, \xi_1, \ldots$ as in Observation~\ref{obs:bad-normalised}.
\end{observation}
\begin{proof}
  By Claim~\ref{claim:decoding-gives-unique-cactus-like-structures}, $\enc(\strC) = \treeT$, so \ref{dict:dom}--\ref{dict:path} apply to~$\strC$ and~$\treeT$; in particular $\treeT$ is a basic shape tree, so Lemma~\ref{lemma:toolbox} applies too.
  Consider the unique infinite path of a run, and on it a configuration in a state $\wlk{p}$ or~$\seg$ (the latter read as~$\wlk{p_{\mathrm{in}}}$), with vertex~$\elu$ and index~$p$.
  By~\ref{tb:move}, a $\move[\autE](\ell(\elu), p, \wlk{})$-move leads to a configuration $(\elw, \wlk{p_2})$ with
  \[
    p \stat{\autE}_{\ell(\elu)} p_1, \qquad (p_1, \sigma, p_2) \in \Delta_{\autE}, \qquad \text{and} \qquad \text{$\elw$ a $\sigma$-successor of~$\elu$}
  \]
  for some $p_1$ and~$\sigma$.
  By~\ref{dict:path}, the first says that $\strC$ has a stationary witnessing walk for~$\autE_{p,p_1}$ at~$\elu$; by~\ref{dict:edge}, the last says $(\elu, \elw) \in \sem{\sigma}_{\strC}$ with $\elu \neq \elw$.
  The move is thus a stationary block followed by one non-stationary step of a walk in~$\strC$, advancing the simulated run of~$\autE$ from~$p$ to~$p_2$; both clauses being equivalences, every such step arises this way.
  Cutting the infinite path at its visits to~$\seg$ therefore yields precisely the factorisations of Observation~\ref{obs:bad-normalised}:
  \begin{itemize}\itemsep0em
    \item between consecutive visits, the walk read off the path witnesses~$\autE_{p_{\mathrm{in}}, p_{\mathrm{fin}}}$, the closing idle coming from the first disjunct of~$\delta(\wlk{p}, \cdot)$, respectively of~$\delta(\seg, \cdot)$;
    \item it makes at least one step, by the placement of~$\stat{\autE,+}$ above;
    \item the vertex visited in~$\seg$ fails~$\ldem_{q,q'}$, hence so does the corresponding world of~$\strC$ (by $\dec$).
  \end{itemize}
  Finally, the run is accepting if and only if its infinite path meets~$\rmB^{\ast}$ infinitely often, that is, if and only if there are infinitely~many~segments.
\end{proof}

\noindent Let $\autT_{\mathrm{fold}}$ be the result of applying Observation~\ref{obs:broadcast} to $\overline{\autT_{\bad}}$. Complementing and broadcasting now deliver the automaton we~promised.

\begin{lemma}\label{lemma:fold-automaton}
  An encoding-ready tree $\treeT$ is accepted by $\autT_{\mathrm{fold}}$ if and only if it is foldable.
  The automaton $\autT_{\mathrm{fold}}$ has polynomially many states in $|Q| {+} |\Sigma_{\varphi}|$ and index~$2$, and is constructed in time polynomial in~$|\AlphabetAPTA| {+} |Q| {+} |\Delta|$.
\end{lemma}
\begin{proof}
  \emph{Correctness.}\;
  By~\ref{tb:hygiene}, the automaton $\autT_{\bad}$ is root-safe, and its B\"uchi condition is complemented by the co-B\"uchi condition $(\rmB^{\ast}, S)$ of index~$2$.
  Observations~\ref{obs:duality} and~\ref{obs:broadcast} therefore make $\autT_{\mathrm{fold}}$ accept exactly those trees at whose every vertex every state~$\seg$ fails.
  For an encoding-ready~$\treeT$, Observation~\ref{obs:seg-succeeds} reads this as $\dec(\treeT)$ admitting none of the sequences of Observation~\ref{obs:bad-normalised}, and hence, by that observation, no bad $\omega$-walk~at~all.

  \smallskip\noindent\emph{Size.}\;
  Besides~$s_{\bad}$, the automaton $\autT_{\bad}$ has $2|Q|^2 \cdot 2 \cdot (|Q|{+}1)$ states $\seg, \wlk{p}$, together with the navigation gadgets of Lemma~\ref{lemma:toolbox} used by~$\move$.
  Dualisation adds no state, and Observation~\ref{obs:broadcast} adds one.
  The index stays~$2$, since the complemented condition is already of index~$2$ and $n = 2$ is even.
  Finally, the transition function is tabulated over~$\AlphabetAPTA$, each entry being computable in polynomial time once $\stat{\autE}_{a}$ and~$\stat{\autE,+}_{a}$ are.
\end{proof}

\subsection{Proof of Lemma~\ref{lemma:foldable-cactus-automaton}}
\begin{proof}
  As the map we keep the encoding~$\enc$ of Definition~\ref{def:encoding}, already used for Lemma~\ref{lemma:cactus-automaton}; only the automaton is new.
  Let $\autT^{f}_{\varphi}$ be the intersection (Proposition~\ref{prop:2apta-intersection}, applicable as both components share the alphabet~$\AlphabetAPTA$ and the arity $\rmN = 2|Q|^2{+}1$) of the automaton~$\autT_{\varphi}$ of Lemma~\ref{lemma:encoding-ready-automaton} with the automaton~$\autT_{\mathrm{fold}}$ of Lemma~\ref{lemma:fold-automaton}.
  As $\autT_{\varphi}$ admits only encoding-ready trees, which is where Lemma~\ref{lemma:fold-automaton} pins down the language of~$\autT_{\mathrm{fold}}$, the intersection accepts precisely the foldable encoding-ready trees for~$\varphi$.
  These are exactly the trees $\enc(\strC)$ for foldable $\varphi$-cacti~$\strC$: the maps $\enc$ and~$\dec$ are mutually inverse on the two sides (Lemmas~\ref{lemma:from-cacti-to-encoding-ready-trees} and~\ref{lemma:from-encoding-ready-trees-to-phi-cacti}, Claim~\ref{claim:decoding-gives-unique-cactus-like-structures}) and preserve regularity (Claim~\ref{claim:decoding-preserves-regularity}), while a tree is foldable by definition of its decoding.
  Finally, Proposition~\ref{prop:2apta-intersection} sums the state counts and keeps the index, so $\autT^{f}_{\varphi}$ still has index~$2$ and polynomially many states in~$|\varphi|$; the remaining bounds are as in Lemma~\ref{lemma:cactus-automaton}.
  This concludes the proof.
\end{proof}

\subsection{Proof of Theorem~\ref{thm:main-characterisation}}
\begin{proof}
A direct consequence of Lemma~\ref{lemma:if-finite-then-cactus-foldable} 
(completeness) and Lemma~\ref{lemma:soundness} (soundness).
\end{proof}

\subsection{Proof of Theorem~\ref{thm:main-result}}
\begin{proof}
The lower bound is folklore and comes from the global satisfiability problem for unimodal logic. That logic has the finite model properly, so the bound also holds for finite satisfiability.
For the upper bound, we follow the procedure from Section~\ref{sec:cactus-models}. Let $\varphi$ be the input formula.
First, we convert an input $\varphi$ into normal form in polynomial time; this step is correct by Lemma~\ref{lemma:normal-form}.
Second, we construct the 2APTA $\autT^{f}_{\varphi}$ that recognises exactly the foldable encoding-ready trees for~$\varphi$. It is the intersection of a 2APTA for encoding-ready trees with a 2APTA for the foldability condition.
This takes exponential time in $|\varphi|$. The resulting automaton has exponential size, but only polynomially many states and polynomial index. It is correct because each step of its construction in Section~\ref{app:foldable-cactus-automaton} is correct.
Third, we test $\autT^{f}_{\varphi}$ for emptiness. As our 2APTA has only polynomially many states and polynomially large index, its non-emptiness can be tested in a single exponential time in $|\varphi|$ (Proposition~\ref{prop:2apta-nonemptiness}).
Suppose $\autT^{f}_{\varphi}$ accepts some tree. Then it also accepts a regular tree~\cite[Thm.~3.7]{PerrinPin04InfiniteTrees}, which is a foldable regular encoding-ready tree for~$\varphi$. This tree encodes a foldable regular $\varphi$-cactus, which in turn yields a finite model of~$\varphi$.
Conversely, every finite model of $\varphi$ gives a foldable regular $\varphi$-cactus. The encoding of this cactus is a foldable encoding-ready tree, which $\autT^{f}_{\varphi}$ accepts.
Hence the algorithm is sound and complete, and runs in time exponential in~$|\varphi|$.
\end{proof}

\section{Appendix for Section~\ref{sec:traces} (Proof of Lemma~\ref{lemma:trace-finiteness})}\label{app:traces}

\begin{proof}
  Suppose, toward a contradiction, that no bound~$\rmB$ exists.
  As explained in Section~\ref{sec:traces}, we may then fix a single concatenation-closed bannable automaton $\autA_{q,q'}$ for which arbitrarily large values of~$m$ occur, and it suffices to exhibit a context-free grammar for the language~$\Lang$ of the walks certifying them.
  Recall that $\strC$ is regular, hence has only finitely many subtrees up to isomorphism; we call their isomorphism classes \emph{subtree types}.

  The terminals of the grammar are the elements of $\Sigma_\varphi$ together with a fresh separation mark~$\sharp$.
  Its nonterminals are
  \[
    \Pref{\mathfrak t, \downarrow},\qquad \Pref{q_1, \mathfrak t, \updownarrow},\qquad
    \Seg{q_1, \mathfrak t, \downarrow},\qquad \Seg{q_1, q_2, \mathfrak t, \updownarrow},
  \]
  one for every subtree type~$\mathfrak t$ of~$\strC$ and all states $q_1, q_2 \in Q$; Table~\ref{tab:nonterminals} says what each of them stands for.
  We derive only walks that start at the root, which is harmless: Definition~\ref{def:trace} constrains~$\xi_0$ only through its last world, and by the last part of Condition~\ref{phicactus:fini} the root reaches every vertex by a walk over~$\Sigma_{\varphi}$ crossing parent-to-child edges only, so replacing~$\xi_0$ by such a walk to its last world preserves every pair $(m, p)$ that the factorisation witnesses.
  The starting symbol is thus $\Pref{\mathfrak t_\varepsilon, \downarrow}$, where $\mathfrak t_\varepsilon$ is the subtree type of the root of~$\strC$.

\begin{table}[H]
  \centering
  \small
  \renewcommand{\arraystretch}{1.2}
  \begin{tabular}{@{}r@{\quad}p{10.4cm}@{}}
    \hline
    Nonterminal & Intended meaning\\\hline
    $\Pref{\mathfrak t, \downarrow}$ &
      Look for the beginning of $\xi_1$ in subtree $\mathfrak t$, expecting no return to an ancestor.\\
    $\Pref{q_1, \mathfrak t, \updownarrow}$ &
      Look for the beginning of $\xi_1$ in subtree $\mathfrak t$, expecting to return to the root
      of~$\mathfrak t$ with the current bad walk's run of $\hat\autA_{q,q'}$ in the state $q_1$, before returning to an ancestor.\\
    $\Seg{q_1, \mathfrak t, \downarrow}$ & 
      The current walk is in the state $q_1$ at the root of~$\mathfrak t$; continue the walk or finish
      and possibly start a new one, expecting no return to an ancestor.\\
    $\Seg{q_1, q_2, \mathfrak t, \updownarrow}$ &
      The current walk is in the state $q_1$ at the root of~$\mathfrak t$; continue the walk
      or finish and possibly start a new one, expecting to return to the root of~$\mathfrak t$ with the current bad walk's run
      of $\hat\autA_{q,q'}$ in the state $q_2$, before returning to an ancestor.\\
    \hline
  \end{tabular}
  \caption{The nonterminals of the grammar, for a subtree type~$\mathfrak t$ and states $q_1, q_2 \in Q$.}\label{tab:nonterminals}
\end{table}

  \vspace{-5mm}
  \noindent In the rules below, $\mathfrak t$ ranges over subtree types and the states $q_i, q_i'$ over~$Q$; each rule is included for every choice of them meeting the side condition stated next to it.

  \vspace{0.35em}
  \noindent \textbf{Tree edges.}\,
  For every subtree type~$\mathfrak t$ whose root has a $\sigma$-edge towards a child of subtree type~$\mathfrak t'$ and a $\tau$-edge back, we include:
  \begin{align*}
    \Pref{\mathfrak t, \downarrow} &\to \sigma\,\Pref{\mathfrak t', \downarrow},\\
    \Pref{\mathfrak t, \downarrow} &\to \sigma\,\Pref{q_1, \mathfrak t', \updownarrow}\,\tau\,\Seg{q_2, \mathfrak t, \downarrow}
      &&\text{for } (q_1, \tau, q_2) \in \hat\Delta,\\
    \Pref{q_1, \mathfrak t, \updownarrow} &\to \sigma\,\Pref{q_2, \mathfrak t', \updownarrow}\,\tau\,\Seg{q_3, q_1, \mathfrak t, \updownarrow}
      &&\text{for } (q_2, \tau, q_3) \in \hat\Delta,\\
    \Seg{q_1, \mathfrak t, \downarrow} &\to \sigma\,\Seg{q_2, \mathfrak t', \downarrow}
      &&\text{for } (q_1, \sigma, q_2) \in \hat\Delta,\\
    \Seg{q_1, \mathfrak t, \downarrow} &\to \sigma\,\Seg{q_2, q_3, \mathfrak t', \updownarrow}\,\tau\,\Seg{q_4, \mathfrak t, \downarrow}
      &&\text{for } (q_1, \sigma, q_2), (q_3, \tau, q_4) \in \hat\Delta,\\
    \Seg{q_1, q_2, \mathfrak t, \updownarrow} &\to \sigma\,\Seg{q_3, q_4, \mathfrak t', \updownarrow}\,\tau\,\Seg{q_5, q_2, \mathfrak t, \updownarrow}
      &&\text{for } (q_1, \sigma, q_3), (q_4, \tau, q_5) \in \hat\Delta.
  \end{align*}
  Note that the first three rules do not require anything of $\sigma$, since it belongs
  to $\xi_0$ where no run is tracked.

  \noindent \textbf{Anchored cycles.}\,
  Let $\elv_0, \sigma_1, \elv_1, \dots, \sigma_k, \elv_0$ be an anchored cycle at the root of a subtree type~$\mathfrak t$, and write $\mathfrak t_i$ for the subtree type of~$\elv_i$.
  For every $1 \leqslant i < k$, the walk may enter the cycle, wander off at~$\elv_i$,
  finish $\xi_0$ and find the beginning of $\xi_1$ there, and come back in some later $\xi_j$, carrying on around it; we thus include:
  \begin{align*}
     \Pref{\mathfrak t, \downarrow} &\to \sigma_1 \dots \sigma_i\, \Pref{q_i', \mathfrak t_i, \updownarrow}\,
     \sigma_{i+1}\, \Seg{q_{i+1}, q_{i+1}', \mathfrak t_{i+1}, \updownarrow}
     \dots \sigma_k\, \Seg{q_k, \mathfrak t, \downarrow},\\
     \Pref{q_0, \mathfrak t, \updownarrow} &\to \sigma_1 \dots \sigma_i\, \Pref{q_i', \mathfrak t_i, \updownarrow}\,
     \sigma_{i+1}\, \Seg{q_{i+1}, q_{i+1}', \mathfrak t_{i+1}, \updownarrow}
     \dots \sigma_k\, \Seg{q_k, q_0, \mathfrak t, \updownarrow},\\[2pt]
     &\hspace{-3.5em}\text{both for } (q_j', \sigma_{j+1}, q_{j+1}) \in \hat\Delta \text{ with } i \leqslant j < k.
  \end{align*}
  Again, the two above rules require nothing of $\sigma_1, \dots, \sigma_i$, since these actions do not
  belong to any bad walk, unlike $\sigma_{i+1}, \dots, \sigma_k$ which occur past the first $\sharp$ mark.
  A walk already at the root of~$\mathfrak t$ may travel around the whole cycle, which calls for:
  \begin{align*}
     \Seg{q_0', \mathfrak t, \downarrow} &\to \sigma_1\, \Seg{q_1, q_1', \mathfrak t_1, \updownarrow}\,
     \sigma_2\, \Seg{q_2, q_2', \mathfrak t_2, \updownarrow} \dots \sigma_k\, \Seg{q_k, \mathfrak t, \downarrow}\\[-1pt]
     &\quad\ \text{for } (q_i', \sigma_{i+1}, q_{i+1}) \in \hat\Delta \text{ with } 0 \leqslant i < k,\\
     \Seg{q_0, q_0', \mathfrak t, \updownarrow} &\to \sigma_1\, \Seg{q_1, q_1', \mathfrak t_1, \updownarrow}\,
     \sigma_2\, \Seg{q_2, q_2', \mathfrak t_2, \updownarrow} \dots \sigma_k\, \Seg{q_k, q_0', \mathfrak t, \updownarrow}\\[-1pt]
     &\quad\ \text{for } (q_0, \sigma_1, q_1), (q_i', \sigma_{i+1}, q_{i+1}) \in \hat\Delta \text{ with } 1 \leqslant i < k.
  \end{align*}
  The same travel taken in the opposite direction may ``break off'' if we do not expect to return;
  we must guess the breakoff point $\elv_{i-1}$ for some $1 \leqslant i \leqslant k$.
  \begin{align*}
     \Seg{q_k', \mathfrak t, \downarrow} &\to \inv{\sigma_k}\, \Seg{q_{k-1}, q_{k-1}', \mathfrak t_{k-1}, \updownarrow}\,
     \inv{\sigma_{k-1}}\, \Seg{q_{k-2}, q_{k-2}', \mathfrak t_{k-2}, \updownarrow}\\[-1pt]
     &\quad\ \dots \inv{\sigma_i}\, \Seg{q_{i-1}, \mathfrak t_{i-1}, \downarrow}\\[-1pt]
     &\quad\ \text{for } (q_j', \inv{\sigma_j}, q_{j-1}) \in \hat\Delta \text{ with } i \leqslant j \leqslant k,\\
     \Seg{q_k, q_k', \mathfrak t, \updownarrow} &\to \inv{\sigma_k}\, \Seg{q_{k-1}, q_{k-1}', \mathfrak t_{k-1}, \updownarrow}\,
     \inv{\sigma_{k-1}}\, \Seg{q_{k-2}, q_{k-2}', \mathfrak t_{k-2}, \updownarrow}\\[-1pt]
     &\quad\ \dots \inv{\sigma_1}\, \Seg{q_0, q_k', \mathfrak t, \updownarrow}\\[-1pt]
     &\quad\ \text{for } (q_k, \inv{\sigma_k}, q_{k-1}), (q_j', \inv{\sigma_j}, q_{j-1}) \in \hat\Delta \text{ with } 1 \leqslant j < k.
  \end{align*}
  A cycle may be entered via a back edge, traversed partly, and left via a back edge too.
  In~this~case, we~must guess the turning point~$\elv_{i-1}$ for some $1 < i \leqslant k$. Let \[\elv_0, \sigma_k, \elv_{k-1}, \sigma_{k-1},
  \elv_{k-2}, \dots, \sigma_i, \elv_{i-1}, \tau_i, \dots, \elv_{k-2}, \tau_{k-1}, \elv_{k-1}, \tau_k, \elv_0 \]
  be a traversal of this shape. There are finitely many such possibilities. We~include:
  \begin{align*}
     \Seg{q_k', \mathfrak t, \downarrow} &\to \sigma_k \Seg{q_{k-1}, q_{k-1}', \mathfrak t_{k-1}, \updownarrow}
     \dots {\sigma_{i+1}} \Seg{q_i, q_i', \mathfrak t_i, \updownarrow} {\sigma_i} \Seg{q_{i-1}, q_{i-1}', \mathfrak t_{i-1}, \updownarrow}\\
     &\tau_i \Seg{p_i, p_i', \mathfrak t_i, \updownarrow} \dots \tau_k \Seg{p_k, \mathfrak t, \downarrow}\\
     &\text{for $(q_j', \sigma_j, q_{j-1}) \in\hat\Delta$ with $i \leqslant j \leqslant k$ and $(p_{j-1}', \tau_j, p_{j}) \in \hat\Delta$ with $i < j \leqslant k$}\\
     &\text{and } (q_{i-1}', \tau_i, p_i) \in \hat\Delta.\\
     \Seg{q_k', q_k, \mathfrak t, \updownarrow} &\to {\sigma_k} \Seg{q_{k-1}, q_{k-1}', \mathfrak t_{k-1}, \updownarrow}
     \dots {\sigma_{i+1}} \Seg{q_i, q_i', \mathfrak t_i, \updownarrow} {\sigma_i} \Seg{q_{i-1}, q_{i-1}', \mathfrak t_{i-1}, \updownarrow}\\
     &\tau_i \Seg{p_i, p_i', \mathfrak t_i, \updownarrow} \dots \tau_k \Seg{p_k, q_k, \mathfrak t, \updownarrow}\\
     &\text{for $(q_j', \sigma_j, q_{j-1}) \in\hat\Delta$ with $i \leqslant j \leqslant k$ and $(p_{j-1}', \tau_j, p_{j}) \in \hat\Delta$ with $i < j \leqslant k$}\\
     &\text{and } (q_{i-1}', \tau_i, p_i) \in \hat\Delta.
 \end{align*}

  \noindent \textbf{Separation marks.}\,
  The next rules place the marks and restart the run of $\autA_{q,q'}$.
  For every subtree type~$\mathfrak t$ whose root does not satisfy $\ldem_{q,q'}$, we include:
  \begin{align*}
    \Pref{\mathfrak t, \downarrow} &\to \sharp\, \Seg{q, \mathfrak t, \downarrow}, &
    \Seg{q', \mathfrak t, \downarrow} &\to \sharp\, \Seg{q, \mathfrak t, \downarrow},\\
    \Pref{q_1, \mathfrak t, \updownarrow} &\to \sharp\, \Seg{q, q_1, \mathfrak t, \updownarrow}, &
    \Seg{q', q_1, \mathfrak t, \updownarrow} &\to \sharp\, \Seg{q, q_1, \mathfrak t, \updownarrow}.
  \end{align*}

  \noindent \textbf{Self-loops, tests, and standing still.}\,
  For every subtree type~$\mathfrak t$ and every letter~$\sigma$ labelling a self-loop at the root of~$\mathfrak t$ (a test being such a letter), we include:
  \begin{align*}
    \Seg{q_1, \mathfrak t, \downarrow} &\to \sigma\, \Seg{q_2, \mathfrak t, \downarrow}
      &&\text{for } (q_1, \sigma, q_2) \in \hat\Delta,\\
    \Seg{q_1, q_2, \mathfrak t, \updownarrow} &\to \sigma\, \Seg{q_3, q_2, \mathfrak t, \updownarrow}
      &&\text{for } (q_1, \sigma, q_3) \in \hat\Delta.
  \end{align*}
  Finally, a walk may simply stop where it is: for every subtree type~$\mathfrak t$ and every $q_1 \in Q$
  we~include $\Seg{q_1, \mathfrak t, \downarrow} \to \varepsilon$ and $\Seg{q_1, q_1, \mathfrak t, \updownarrow} \to \varepsilon$.

  The grammar described above consists of finitely many rules because $\strC$
  is a regular cactus of bounded arity with finitely many non-empty transition relations.

  Now, let us introduce some definitions that will allow us to formalise the semantics
  given in Table~\ref{tab:nonterminals}. For a walk $\xi$, we obtain its \emph{projection}
  $\tilde\xi$ by removing the vertices and leaving the edge actions. Let $\|\xi\|$ be the length
  of $\xi$ understood as the number of edges it crosses (clearly $\|\xi\| = |\tilde\xi|$).
  If a walk $\xi$ factorises as $\xi_0\xi_1\cdots\xi_m$ with $\xi_i$ bad $(q, q')$-walks
  for $0 < i < m$
  and $\tilde\xi = \sigma_1, \sigma_2, \dots, \sigma_n$,
  then a function
  $\eta\colon \{0, \dots, \|\xi\|\} \to Q \cup \{\,\sharp\,\}$
  is a \emph{labelling} of $\xi$ when
  \begin{enumerate}
    \item for all $k$, $\eta(k) = \sharp$ if and only if $k = \sum_{j=0}^i \|\xi_j\|$
    for some $0 \leqslant i < m$\\
    (boundary points are labelled with $\sharp$ marks),
    \item\label{run-labelling:same} for $i \geqslant \|\xi_0\|$,
    if $\eta(i), \eta(i+1) \in Q$, then $(\eta(i), \sigma_{i+1}, \eta(i+1)) \in \hat\Delta$,\\
    and if $\eta(i) = \eta(i+1) = \sharp$, then $(q, \sigma_{i+1}, q') \in \hat\Delta$,
    \item\label{run-labelling:different}
    if $\eta(i) = \sharp$, then $(\eta(i-1), \sigma_i, q') \in \hat\Delta$ for $i > \|\xi_0\|$, $\eta(i-1) \in Q$\\
    and $(q, \sigma_{i+1}, \eta(i+1)) \in \hat\Delta$ for $i < n$, $\eta(i+1) \in Q$.
  \end{enumerate}
  The labelling is called \emph{open} if (\ref{run-labelling:same}) holds for all $i$ and
  (\ref{run-labelling:different}) holds not just for $i > \|\xi_0\|$ but for all $i > 0$.
  The following proposition formalises the semantics.
  \begin{proposition}
    For all subtree types $\mathfrak t$, states $q_1, q_2$ and terminal strings $x$
    of the form $\tilde{\xi_0}\sharp\tilde{\xi_1}\sharp\cdots\sharp\tilde{\xi_m}$,
    \begin{itemize}
      \item $\Pref{\mathfrak t, \downarrow} \stackrel{*}{\Rightarrow} x$ if and only if
      there exists a walk $\xi$ that factorises as $\xi_0\xi_1\cdots\xi_m\xi_{m+1}$
      like in~Definition~\ref{def:trace}, such that $\xi_0$ uses only parent-to-child edges
      and $\xi$ is contained entirely within the subtree
      of type $\mathfrak t$ of its initial element;
      \item $\Pref{q_1, \mathfrak t, \updownarrow} \stackrel{*}{\Rightarrow} x$ if and only if
      there exists a cycle $\xi$ that factorises as $\xi_0\xi_1\cdots\xi_m\xi_{m+1}$
      like above, where $\xi_0$ uses only parent-to-child edges,
      with a labelling $\eta$ satisfying $\eta(\|\xi\|) = q_1$ (or possibly $\eta(\|\xi\|) = \sharp$
      if $q_1 = q$),
      and is contained entirely within the subtree of type $\mathfrak t$ of its initial element;
      \item $\Seg{q_1, \mathfrak t, \downarrow} \stackrel{*}{\Rightarrow} x$ if and only if
      there exists a walk $\xi$ that factorises as $\xi_0\xi_1\cdots\xi_m$
      with $\xi_i$ bad $(q, q')$-walks for $0 < i < m$, with an open labelling $\eta$
      satisfying $\eta(0) = q_1$ (or possibly $\eta(0) = \sharp$
      if $q_1 = q'$),
      and is contained entirely within the subtree of type $\mathfrak t$ of its initial element;
      \item $\Seg{q_1, q_2, \mathfrak t, \updownarrow} \stackrel{*}{\Rightarrow} x$ if and only if
      there exists a cycle $\xi$ that factorises as $\xi_0\xi_1\cdots\xi_m$ like above,
      with an open labelling $\eta$ satisfying $\eta(0) = q_1$ (or possibly $\eta(0) = \sharp$
      if $q_1 = q'$) and $\eta(\|\xi\|) = q_2$ (or possibly $\eta(\|\xi\|) = \sharp$ if $q_2 = q$),
      and is contained entirely within the subtree of type $\mathfrak t$ of its initial element.
    \end{itemize}
  \end{proposition}
  \begin{proof}[Proof idea]
    The directions $\Rightarrow$ can be proven by straightforward induction
    on the length of~a~derivation. In~the induction step, we~show that the right-hand side
    of every rule corresponds to a~walk that actually exists in~$\strC$.

    The directions $\Leftarrow$ require a little more thought. 
    We~proceed by induction on~$\|\xi\|$. We~consider a~walk~$\xi$
    (with a~labelling~$\eta$) that fits entirely within the subtree of its initial element~$\elv_0$.
    We~select $\mathfrak t$ to be the type of this subtree. The~choice of $q_1, q_2$
    is determined by~$\eta$. We~must show that $\tilde\xi$ can be derived from
    the appropriate nonterminals (those associated with $\mathfrak t$ and possibly
    the states $q_1$ and $q_2$), inserting $\sharp$ marks where $\eta$ has them.
    If~$\eta(0) = \sharp$, we~use a~rule that inserts~$\sharp$ and consider new $\eta'$
    with $\eta'(0) = q$ and other values unchanged.
    Now, if~$\elv_0$ occurs only at the beginning, $\tilde\xi$ must be derived
    from a~$\downarrow$-nonterminal by a~rule
    that represents a~motion down the cactus, either via a~parent-to-child edge
    or via a~back edge. We~choose the appropriate rule to perform the first step of $\tilde\xi$
    and use the fact that the rest of $\tilde\xi$ is derivable from the last nonterminal
    in the right-hand side of this rule by inductive hypothesis.
    If, however, $\elv_0$ occurs at least twice, we~let $\xi = \elv_0 \zeta_1 \elv_0 \zeta_2$,
    where $\zeta_1$ contains no occurrences of~$\elv_0$ except at the endpoints.
    The~walk~$\zeta_1$ is a~cycle, so~we~must use a~rule that enters an anchored cycle
    that belongs to~$\elv_0$ and returns (unless $\zeta_1$ is stationary, in~which case
    we just apply a~self-loop rule). The~vertex $\elv_0$ is connected
    with its anchored cycle in two ways: via a~parent-to-child edge and via a~back edge,
    resulting in two possibilities for entry and two possibilities for exit, thus~four
    possibilities for traversal in~total. Each of them is made possible by an appropriate rule.
    Further detours from this anchored cycle taken by $\zeta_1$ are shorter walks,
    and~therefore their projections are derivable by inductive hypothesis.
    Similarly, $\tilde\zeta_2$ is derivable
    from the last nonterminal symbol in the right-hand side of the chosen rule
    by inductive hypothesis.
  \end{proof}
  From the above, we may conclude that all \emph{normalised} walks are derivable: if $\xi$
  factorises as $\xi_0\xi_1\cdots\xi_m\xi_{m+1}$ (again, like in~Definition~\ref{def:trace}), then
  $\Pref{\mathfrak t_\varepsilon, \downarrow} \stackrel{*}{\Rightarrow} \tilde\zeta\sharp\tilde{\xi_1}\sharp\cdots\sharp\tilde{\xi_m}\sharp\tilde{\xi_{m+1}}$
  for $\zeta$ a walk from the root to the initial element of $\xi_1$ that passes through
  parent-to-child edges only, which exists by the last part of Condition~\ref{phicactus:fini}.
  Conversely, every derivable $x$ is a projection
  of a normalised walk with $\sharp$ marks inserted as indicated by the walk's labelling.
  Therefore, the fact that there exist walks witnessing arbitrary values of $m$
  implies that words with an arbitrary number of $\sharp$ marks are derivable.
  Then, by Ogden's lemma, there exists a word of the form $sztyx$, where $zy$
  contains at least one $\sharp$ mark, such that $sz^nty^nx$ is derivable for every $n$.
  Clearly, a fragment of a walk is a walk as well, so $z^n$ and $y^n$ correspond
  to walks existing in $\strC$ for every $n$.
  By inspecting the proof of Ogden's lemma, we see that the derivation involves
  the steps $A \stackrel{*}{\Rightarrow} zAy \stackrel{*}{\Rightarrow} zty$.
  Because the nonterminal $A$ repeats, some subtree type repeats together with the automaton's
  state, thus $z$ and $y$ can be looped into $\omega$-walks that can be labelled with
  some repeating run of $\hat{\autA}_{q,q'}$.
  If $z$ contains a $\sharp$ mark, we get an $\omega$-walk with all $\xi_i$ bad,
  and if $y$ contains a $\sharp$ mark, we get an $\omega$-walk with all $\inv{\xi_i}$ bad.
  The walks $\xi_i$ are non-trivial, since bad $(q, q')$-walks have distinct endpoints
  in $\varphi$-cacti. This contradicts the foldability of~$\strC$.
\end{proof}
\section{Appendix for Section~\ref{sec:soundness}}\label{app:soundness}

\subsection{Proof of Lemma~\ref{lemma:invariant}}

\begin{proof}
  Types are clearly preserved: a fold-operation never changes the atomic
  propositions satisfied at a vertex.
  Fix a concatenation-closed bannable automaton $\autA_{q, q'}$ and let
  $((\elv, \elv j), (\elw, \elw j))$ be the fold-operation performed: it adds, to the
  appropriate transition relations, new edges between $\elw$ and $\elv j$, and possibly
  between $\elt$ and $\elw$ (appropriately directed), where $\elt$ is the anchored
  descendant of $\elv$. For the inclusion
  $\vtr_{q, q'}^{\strD'}(\elu) \subseteq \vtr_{q, q'}^{\strD}(\elu)$ we prove by
  induction on $n$: for every $\elu \in \domD'$, if a walk $\xi$ witnesses
  $(m, q'') \in \vtr_{q, q'}^{\strD'}(\elu)$ and contains~$n$ new edges, then
  $(m, q'') \in \vtr^\strD_{q,q'}(\elu)$.
  For $n = 0$ the claim is immediate, as $\xi$ exists in $\strD$. Assume it for $n$, and
  let $\xi$ contain $n+1$ new edges and witness $(m, q'') \in \vtr_{q, q'}^{\strD'}(\elu)$.
  By definition, $\xi$ factorises as $\xi_0\xi_1\cdots\xi_m\xi_{m+1}$, where every $\xi_i$
  with $1 \leqslant i \leqslant m$ is a bad $(q, q')$-walk and $\xi_{m+1}$ witnesses
  $\hat\autA_{q, q''}$. Let $\elu', \sigma, \elu''$ be the last single-edge fragment of $\xi$
  carrying a new edge. If it lies within $\xi_0$ there is nothing to do, as all bad walks
  exist in $\strD$; hence assume it lies inside $\xi_i$ with $i \geqslant 1$. In the
  corresponding run $(\elu_0, q), \sigma_1, (\elu_1, q_1), \dots, \sigma_\ell, (\elu_\ell, q_\ell)$,
  with $q_\ell$ equal to $q'$ or $q''$, let $(\elu', q_k), \sigma_{k+1}, (\elu'', q_{k+1})$ be the
  transition crossing $\elu', \sigma, \elu''$. The prefix of $\xi$ up to $\elu'$ witnesses
  $(i-1, q_k) \in \vtr^{\strD'}_{q, q'}(\elu')$ and satisfies the induction hypothesis, so
  $(i-1, q_k) \in \vtr^\strD_{q, q'}(\elu')$. The new edges join two pairs of vertices,
  leaving four cases.

  \begin{enumerate}
    \item $\elu' = \elw$ and $\elu'' = \elv j$.\\
      Since $\elw \approx \elv$, we have $(i-1, q_k) \in \vtr^\strD_{q, q'}(\elv)$; let $\zeta$
  witness this. By definition of the fold-operation, an edge identical to
  $\elw, \sigma, \elv j$ joins $\elv$ and $\elv j$, so extending $\zeta$ by it continues the
  run with the suffix of $\xi$ starting at $\elv j$. The resulting walk exists in $\strD$,
  so~$(m, q'') \in \vtr^\strD_{q, q'}(\elu)$.

  \item $\elu' = \elv j$ and $\elu'' = \elw$.\\
  As in the previous case, but using $\elw j \approx \elv j$.

  \item $\elu' = \elt$ and $\elu'' = \elw$.\\ 
    By definition of the fold-operation, an edge identical to $\elt, \sigma, \elw$ joins
  $\elt$ and $\elv$, so $(i-1, q_{k+1}) \in \vtr^{\strD}_{q, q'}(\elv)$, and since
  $\elv \approx \elw$ also $(i-1, q_{k+1}) \in \vtr^\strD_{q, q'}(\elw)$. Extending a
  witnessing walk by the suffix of $\xi$ starting at $\elw$ gives a walk of $\strD$ that
  continues the run, so~$(m, q'') \in \vtr^\strD_{q, q'}(\elu)$.

  \item $\elu' = \elw$ and $\elu'' = \elt$.\\
Use $\elw \approx \elv$, then replace $\elw, \sigma, \elt$ by the edge
  $\elv, \sigma, \elt$ to continue the run.\qedhere
  \end{enumerate}
\end{proof}

\subsection{Proof of Lemma~\ref{lemma:quasicactus}}
\begin{proof}
  We check the conditions of Definition~\ref{def:phi-cactus}, which a $\varphi$-quasicactus
  must satisfy with ``quasi-anchored'' in place of ``anchored'' (Definition~\ref{def:quasicactus}).
  Condition~\ref{phicactus:fini} is immediate, since the fold-operation adds no vertices
  and removes a whole subtree, so every remaining vertex keeps its parent and the $\Sigma_{\varphi}$-edge to it;
  and so is Condition~\ref{phicactus:bool}, since types are not modified.
  We clarify that the finiteness of traces indeed implies that $\strD'$ is foldable.
  By Lemma~\ref{lemma:invariant}, we assert that the values of $m$ are bounded, whereas
  if there existed a bad $\omega$-walk, it would have infinitely many walks of increasing
  trace as prefixes.

  \textbf{Walk demands~\ref{phicactus:walk}.} As for positive demands, the only vertex that
  could have lost a witness is $\elw$, but by definition of the operation $\elv i$
  becomes its new witness.
  As for negative demands, we prove the following by induction on $n$:
  for all vertices $\elv \in \rmD'$ and automata $\autA_{q,q'}$,
  if a walk starting at $\elv$ has $n$ occurrences of the edges introduced
  by the fold-operation and witnesses $\autA_{q,q'}$, then $\strD', \elv \models \wdem_{q,q'}$.
  The base case is obvious since the walk then exists in $\strD$ which
  is a $\varphi$-quasicactus. Assume the hypothesis for $n$ and
  take a walk $\xi$ in $\strD'$ starting at $\elv$, witnessing
  $\autA_{q,q'}$ and having $n+1$ occurrences of the new edges.
  Let $\elv_i, \sigma_{i+1}, \elv_{i+1}$ be the first occurrence of a new edge
  on $\xi$. Consider the transition $(q_1, \sigma_{i+1}, q_2)$ of the run
  of $\autA_{q,q'}$ on $\xi$ that corresponds to crossing this edge.
  The suffix of $\xi$ starting from $\elv_{i+1}$ satisfies the induction
  hypothesis, implying that $\strD', \elv_{i+1} \models \wdem_{q_2, q'}$.
  Since $\lambda^{\strD'}(\elv_{i+1}) = \lambda^{\strD}(\elv_{i+1})$,
  we have $\strD, \elv_{i+1} \models \wdem_{q_2, q'}$.
  Now, in four cases like those in the proof of Lemma~\ref{lemma:invariant},
  we use the similarity of edges to conclude that $\strD, \elv_i \models \wdem_{q_1, q'}$,
  and hence some walk $\zeta$ realises this demand at $\elv_i$ in $\strD$.
  Since the prefix of $\xi$ up to $\elv_i$ exists in $\strD$,
  it follows that it can be extended by the walk $\zeta$, yielding
  a walk that witnesses $\autA_{q,q'}$ at $\elv$ in $\strD$.
  Therefore $\strD, \elv \models \wdem_{q,q'}$, and thus $\strD', \elv \models \wdem_{q,q'}$
  again by preservation of types.

  \textbf{Cycle demands~\ref{phicactus:cycle}.} Positive demands are immediate
  from the definition of the fold-operation. Denoting the vertices as in
  Definition~\ref{def:fold-op}, the only vertex that could have lost its witness
  when the subtree $\rmD_{\elw i}$ was removed is $\elw$; but $\elv i$ serves as
  its new witness, thanks to the added edges $\elw, \sigma, \elv i$ and
  $\elv i 0^k, \sigma', \elv$.
  Negative demands were handled in the main body of the paper.
\end{proof}

\subsection{Proof of Lemma~\ref{lemma:soundness}}

We prove a helpful proposition first.

\begin{proposition}\label{prop:anchored-paths-bounded}
  In a regular $\varphi$-cactus with $t$ distinct subtree types,
  the length of anchored paths is bounded by $t+1$.
\end{proposition}
\begin{proof}
  Assume toward a contradiction that there is an anchored path of length
  at least $t+2$. Its prefix is of the form $\elv, \elv i, \elv i0, \dots, \elv i0^t$.
  Observe that $\elv i, \elv i0, \dots, \elv i0^t$ is a $0$-path of length $t+1$.
  Since there are only $t$ distinct subtree types, some two vertices
  on this $0$-path have the same subtree type. This implies the existence
  of an infinite $0$-path, forbidden by~\ref{phicactus:fini}.
\end{proof}

\begin{proof}
  Let $\strC$ be a foldable regular $\varphi$-cactus (Definitions~\ref{def:phi-cactus}
  and~\ref{def:foldability}), let $t$ be the number of its subtree types, and let
  $n_\lambda = \left|\left\{ \lambda^\strC(\elu) \mid \elu \in \rmC \right\}\right|$,
  which is finite because $\Sigma_\varphi$ is.

  We make one more preprocessing assumption, which simplifies the arguments below.
  Namely, for every vertex~$\elv$ of~$\strC$ and every demand~$\delta$, we assume that $\elv$ has a $\kappa(\delta)$-th child only if the witness~$\xi_\delta^{\elv}$ exists and is anchored.
  This is without loss of generality.
  Whenever a pair $(\elv, \delta)$ violates the assumption, we remove the whole subtree rooted at the $\kappa(\delta)$-th child of~$\elv$.
  We claim that the resulting structure is still a foldable regular $\varphi$-cactus.
  Note that only whole subtrees strictly below the root are removed, so the remaining vertices keep their labels, and every remaining vertex other than the root keeps its parent and the $\Sigma_{\varphi}$-edge to it.
  Hence Conditions~\ref{phicactus:bool} and~\ref{phicactus:fini} of Definition~\ref{def:phi-cactus} are preserved.
  The ``otherwise'' clauses of Conditions~\ref{phicactus:walk} and~\ref{phicactus:cycle} are preserved as well.
  Every witnessing walk or cycle in the resulting structure is also one in~$\strC$, so a vertex that had none still has none.
  It remains to show that the positive parts of Conditions~\ref{phicactus:walk} and~\ref{phicactus:cycle} still hold.
  That is, every remaining vertex~$\elw$ that satisfies a demand~$\delta'$ keeps its witness~$\xi_{\delta'}^{\elw}$.
  If $\xi_{\delta'}^{\elw}$ is stationary, it visits no vertex other than~$\elw$, and we are done.
  Otherwise $\xi_{\delta'}^{\elw}$ is anchored, so apart from~$\elw$ itself it visits only vertices of the form $\elw\kappa(\delta')0^{i}$.
  Every removed vertex, in turn, has a prefix $\elv\kappa(\delta)$ for some violating pair $(\elv, \delta)$.
  Suppose that some vertex $\elw\kappa(\delta')0^{i}$ has such a prefix.
  Since $\elw$ was not removed, $\elv\kappa(\delta)$ is not a prefix of~$\elw$, and so $\elw$ is a proper prefix of~$\elv\kappa(\delta)$.
  Thus $\elv\kappa(\delta)$ is a prefix of $\elw\kappa(\delta')0^{i}$ that is longer than~$\elw$ and ends with the non-zero letter~$\kappa(\delta)$.
  The only such prefix is $\elw\kappa(\delta')$, hence $\elv = \elw$ and $\kappa(\delta) = \kappa(\delta')$, and so $\delta = \delta'$ because $\kappa$~is~injective.
  But then $\xi_{\delta}^{\elv}$ exists and is anchored, so the pair $(\elv, \delta)$ is not violating, which is a contradiction.
  Therefore $\xi_{\delta'}^{\elw}$ avoids every removed vertex, and it is still a witness in the resulting structure.

  We now continue with the actual proof.
  Let $n_\mathrm{pair}$ be the number of $(m, q)$-pairs belonging to
  $\vtr^\strC_{q, q'}(\elu)$ for some $\elu$ and some concatenation-closed
  bannable $\autA_{q, q'}$; this is a natural number by Lemma~\ref{lemma:trace-finiteness}.
  Put $n_\vtr \deff 2^{n_\mathrm{pair}|Q|^2}$ and
  \[ M \deff (t + 1)(2n_\vtr^2 n_\lambda^2 \cdot 2|Q|^2|\Sigma_\varphi| + 1). \]
  We construct a finite sequence $\strD_0, \strD_1, \dots, \strD_k$ of structures
  with $\strD_0 \deff \strC$ by applying successive fold-operations.
  We claim that, for every $i$, every path from the root to a vertex at depth $M$
  in $\strD_i$ carries a pair of similar edges admitting a fold-operation.
  Consider any such path $\xi = \varepsilon, \elv_1, \elv_2, \dots, \elv_M$.
  Recall that anchored paths have length at most $t + 1$
  (Proposition~\ref{prop:anchored-paths-bounded}). Hence non-zero children
  are separated by fragments of length at most $t + 1$
  on~$\xi$. Every child of the root is a non-zero child, so $\xi$ contains
  at least $2n_\vtr^2 n_\lambda^2 \cdot 2|Q|^2|\Sigma_\varphi| + 1$ non-zero children.
  A \emph{trace configuration} $C$ assigns to every concatenation-closed bannable
  automaton $\autA_{q, q'}$ a trace, \ie a subset of $\N \times Q$; every vertex
  $\elu \in \rmD_i$ carries the configuration $C_\elu$ with
  $C_\elu(\autA_{q, q'}) = \vtr_{q, q'}^{\strD_i}(\elu)$.
  Let the \emph{extended type} of a vertex $\elv i$ be
  $(\lambda^\strC(\elv i), \lambda^\strC(\elv), C_{\elv i}, C_\elv, i, \sigma)$,
  where the witnessing walk of $\elv$ going through $\elv i$ contains
  the fragment $\elv, \sigma, \elv i$ (\cf the initial transition
  in Definition~\ref{def:similar}).
  By Lemma~\ref{lemma:invariant}, fold-operations preserve types,
  but can ``dilute'' traces into smaller subsets.
  The bound $n_\vtr$ nonetheless works, as it counts all subsets of the
  $(m, q)$-pairs of the original $\varphi$-cactus $\strC$.
  So some extended type occurs at least three times on $\xi$, and we can pick
  two vertices $\elv i, \elw i$ of that type with $\elw \notin \elv i0^*$.
  Hence $(\elv, \elv i) \approx (\elw, \elw i)$, and
  $((\elv, \elv i), (\elw, \elw i))$ is a permissible fold-operation.

  So as long as $\strD_i$ contains a vertex at depth $M$, a fold-operation removing it
  is available, which yields the sequence $\strD_0, \strD_1, \dots, \strD_k$.
  It stops at some finite index $k$ because the number of vertices at depth $M$
  is finite, as $\rmC$ is a tree of bounded arity.
  By~Lemma~\ref{lemma:quasicactus}, every $\strD_i$ is a $\varphi$-quasicactus.
  The $\varphi$-quasicactus $\strD_k$ is finite, since it has finite degree and no
  path in it reaches depth $M$. Therefore $\strD_k$ is a finite model of $\varphi$.
\end{proof}

\section{Appendix for Section~\ref{sec:conclusions}}\label{app:conclusions}

\subsection{How robust is our current approach?}

Concatenation-closedness gives a simple correspondence between undesirable
cycles and infinite walks of a specific shape. Our strategy has been to ensure
that an infinite model admits no such walk, and to finitise it by a process
preserving this property (traces proved a helpful tool here). The crucial
application of concatenation-closedness is in the proof of
Lemma~\ref{lemma:if-finite-then-cactus-foldable}, which certifies that the
foldability condition is not \emph{too strong}---if a finite model exists, our
algorithm will detect it.

\subsection{Why does our approach not extend so easily to the full logic?}

Generalised to other languages, the condition does become too strong.
Consider the formula $\mlbox{\pe^*}(\mldiamond{\pe}\top \land \neg \lop(\pe(\pe\pe)^*))$.
It has a finite model: a cycle of even length. Yet every model of it contains
an infinite walk divisible into fragments of odd length, violating the
(generalised) foldability condition! One may try to modify the condition and
adapt the notion of trace accordingly. An interesting direction would be
avoiding walks that have an infinite subsequence of vertices, no two of which
have between them a walk acceptable to them as a cycle. The trace of a walk
could then be the size of the longest such subsequence. This generalisation,
however, troubles the finiteness of trace: no straightforward pumping argument
works, since enough repetitions of a ``bad'' walk may yield a ``good'' one, so
a walk of large trace need not be transformable into an infinite bad walk with
an infinite ``conflicting subsequence''. It remains open how to generalise
traces and the foldability condition so as to solve finite satisfiability of
$\LoopPDL$ in general.
We are actively working on this problem, and we hope to report on it in the near future.
\end{document}